\documentclass[a4paper,onecolumn,unpublished]{quantumarticle}
\pdfoutput=1
\usepackage{hyperref}
\usepackage{amssymb}
\usepackage{amsmath}
\usepackage{amsthm}
\usepackage{tikzit}
\usepackage[ruled,algosection]{algorithm2e}

\usepackage{silence}
\input{quantum.tikzdefs}
\tikzstyle{gate}=[shape=rectangle, text height=1.5ex, text depth=0.25ex, yshift=0.5mm, fill=white, draw=black, minimum height=3mm, yshift=-0.5mm, minimum width=3mm, font={\small}, tikzit category=circuit, inner sep=2pt]
\tikzstyle{big gate}=[shape=rectangle, text height=1.5ex, text depth=0.25ex, yshift=0.5mm, fill=white, draw=black, minimum height=10mm, yshift=-0.5mm, minimum width=5mm, font={\small}, tikzit category=circuit]
\tikzstyle{Z dot}=[inner sep=0mm, minimum size=2mm, shape=circle, draw=black, fill={rgb,255: red,221; green,255; blue,221}, tikzit category=zx]
\tikzstyle{Z phase dot}=[minimum size=5mm, font={\footnotesize\boldmath}, shape=rectangle, rounded corners=2mm, inner sep=0.2mm, outer sep=-2mm, scale=0.8, tikzit shape=circle, draw=black, fill={rgb,255: red,221; green,255; blue,221}, tikzit draw=blue, tikzit category=zx]
\tikzstyle{X dot}=[Z dot, shape=circle, draw=black, fill={rgb,255: red,255; green,136; blue,136}, tikzit category=zx]
\tikzstyle{X phase dot}=[Z phase dot, tikzit shape=circle, tikzit draw=blue, fill={rgb,255: red,255; green,136; blue,136}, font={\footnotesize\boldmath}, tikzit category=zx]
\tikzstyle{hadamard}=[fill=yellow, draw=black, shape=rectangle, inner sep=0.6mm, minimum height=1.5mm, minimum width=1.5mm, tikzit category=zx]
\tikzstyle{paulibox}=[fill={rgb,255: red,221; green,221; blue,255}, draw=black, shape=rectangle, inner sep=0.6mm, minimum height=5mm, minimum width=5mm, font={\footnotesize}, text height=1.5ex, text depth=0.25ex, tikzit category=zx]
\tikzstyle{vertex}=[inner sep=0mm, minimum size=1mm, shape=circle, draw=black, fill=black, tikzit category=misc]
\tikzstyle{vertex set}=[inner sep=0mm, minimum size=1mm, shape=circle, draw=black, fill=white, font={\footnotesize\boldmath}, tikzit category=misc]
\tikzstyle{small black dot}=[fill=black, draw=black, shape=circle, inner sep=0pt, minimum width=1.2mm, tikzit category=circuit]
\tikzstyle{cnot ctrl}=[fill=black, draw=black, shape=circle, inner sep=0pt, minimum width=1.2mm, tikzit category=circuit]
\tikzstyle{cnot targ}=[fill=white, draw=white, shape=circle, tikzit category=circuit, label={center:$\oplus$}, inner sep=0pt, minimum width=2.1mm, tikzit fill={rgb,255: red,102; green,204; blue,255}, tikzit draw=black]
\tikzstyle{ket}=[fill=white, draw=black, shape=regular polygon, regular polygon sides=3, regular polygon rotate=-30, scale=0.7, inner sep=1pt, tikzit category=circuit, tikzit shape=rectangle, tikzit fill=green]
\tikzstyle{bra}=[fill=white, draw=black, shape=regular polygon, regular polygon sides=3, regular polygon rotate=30, scale=0.7, inner sep=1pt, tikzit category=circuit, tikzit shape=rectangle, tikzit fill=red]
\tikzstyle{scalar}=[shape=rectangle, text height=1.5ex, text depth=0.25ex, yshift=0.5mm, fill=white, draw=black, minimum height=5mm, yshift=-0.5mm, minimum width=5mm, font={\small}]
\tikzstyle{clabel}=[fill=white, draw=none, shape=rectangle, tikzit fill={rgb,255: red,56; green,255; blue,242}, font={\footnotesize}, inner sep=1pt, tikzit category=labels]
\tikzstyle{empty diagram}=[draw={gray!40!white}, dashed, shape=rectangle, minimum width=1cm, minimum height=1cm, tikzit category=misc]
\tikzstyle{amap}=[fill=white, draw=black, shape=NEbox, tikzit category=asymmetric, tikzit fill=yellow, tikzit shape=rectangle]
\tikzstyle{amap conj}=[fill=white, draw=black, shape=NWbox, tikzit category=asymmetric, tikzit fill=green, tikzit shape=rectangle]
\tikzstyle{amap adj}=[fill=white, draw=black, shape=SEbox, tikzit category=asymmetric, tikzit fill=red, tikzit shape=rectangle]
\tikzstyle{amap trans}=[fill=white, draw=black, shape=SWbox, tikzit category=asymmetric, tikzit fill=orange, tikzit shape=rectangle]
\tikzstyle{astate}=[fill=white, draw=black, shape=NEtriangle, tikzit category=asymmetric, tikzit shape=circle, tikzit fill=yellow]
\tikzstyle{astate conj}=[fill=white, draw=black, shape=NWtriangle, tikzit category=asymmetric, tikzit shape=circle, tikzit fill=green]
\tikzstyle{astate adj}=[fill=white, draw=black, shape=SEtriangle, tikzit category=asymmetric, tikzit shape=circle, tikzit fill=red]
\tikzstyle{astate trans}=[fill=white, draw=black, shape=SWtriangle, tikzit category=asymmetric, tikzit shape=circle, tikzit fill=orange]
\tikzstyle{white dot}=[inner sep=0mm, minimum size=2mm, shape=circle, draw=black, fill={rgb,255: red,250; green,250; blue,250}]
\tikzstyle{white phase dot}=[minimum size=5mm, font={\footnotesize\boldmath}, shape=rectangle, rounded corners=2mm, inner sep=0.2mm, outer sep=-2mm, scale=0.8, tikzit shape=circle, draw=black, fill={rgb,255: red,250; green,250; blue,250}, tikzit draw=blue]
\tikzstyle{hbox}=[shape=rectangle, text height=2mm, fill={rgb,255: red,255; green,235; blue,61}, draw=black, minimum height=2mm, minimum width=2mm, font={\small}, tikzit category=zh, inner sep=0pt, rounded corners=0.5mm]
\tikzstyle{Z dot (zh)}=[inner sep=0mm, minimum size=2mm, shape=circle, draw=black, fill={rgb,255: red,250; green,250; blue,250}, tikzit category=zh]
\tikzstyle{X dot (zh)}=[Z dot, shape=circle, draw=black, fill={rgb,255: red,193; green,193; blue,193}, tikzit category=zh]
\tikzstyle{triangle}=[fill={rgb,255: red,255; green,136; blue,136}, draw=black, shape=isosceles triangle, isosceles triangle apex angle=60, minimum size=2.5mm, inner sep=0mm]
\tikzstyle{labelled hbox}=[shape=rectangle, text height=1.75ex, text depth=0.5ex, fill={rgb,255: red,255; green,235; blue,61}, draw=black, minimum height=3mm, minimum width=4mm, font={\small}, tikzit category=zh, inner sep=1.3pt, rounded corners=0.5mm]
\tikzstyle{Z phase dot (zh)}=[Z phase dot, tikzit shape=circle, tikzit draw=blue, fill={rgb,255: red,250; green,250; blue,250}, font={\footnotesize\boldmath}, tikzit category=zh]
\tikzstyle{X phase dot (zh)}=[Z phase dot, tikzit shape=circle, tikzit draw=blue, fill={rgb,255: red,193; green,193; blue,193}, font={\footnotesize\boldmath}, tikzit category=zh]
\tikzstyle{W node}=[fill=black, draw=black, shape=regular polygon, regular polygon sides=3, minimum size=2mm]
\tikzstyle{Z dot (zw)}=[fill=white, draw=black, shape=circle, minimum width=1.2mm, inner sep=0pt]
\tikzstyle{Z phase dot XL}=[Z phase dot, fill={rgb,255: red,250; green,250; blue,250}, draw=black, shape=circle, tikzit draw={rgb,255: red,191; green,0; blue,64}, tikzit shape=circle, font={\large\boldmath}, inner sep=0.0mm]
\tikzstyle{gn}=[fill=green, draw=black, shape=circle, tikzit category=ZX, tikzit fill=green, tikzit draw=black, tikzit shape=circle, inner sep=2pt]
\tikzstyle{rn}=[fill=red, draw=black, shape=circle, tikzit fill=red, tikzit draw=black, tikzit category=ZX, tikzit shape=circle, inner sep=2pt]
\tikzstyle{divide}=[regular polygon, regular polygon sides=3, shape border rotate=90, draw=black, fill=gray, inner sep=1.6pt, tikzit category=scal, rounded corners=0.8mm]
\tikzstyle{black}=[fill=black, draw=black, shape=circle, tikzit fill=black, tikzit draw=black, tikzit shape=circle, tikzit category=IH, inner sep=2pt]
\tikzstyle{gather}=[fill=gray, draw=black, tikzit category=scal, rounded corners=0.8mm, regular polygon, regular polygon sides=3, shape border rotate=-90, inner sep=1.6pt]
\tikzstyle{ggen}=[fill=white, draw=black, shape=rectangle, rounded corners=2mm, line width=1pt, tikzit draw=red, tikzit category=scal]
\tikzstyle{white}=[fill=white, draw=black, shape=circle, inner sep=2pt, tikzit category=IH]
\tikzstyle{mbox}=[fill=white, draw=black, rounded rectangle, rounded rectangle west arc=none, tikzit category=scal, tikzit shape=rectangle]
\tikzstyle{A}=[fill=white, shape=circle, tikzit category=scal, inner sep=1pt]
\tikzstyle{ggreen}=[fill=green, draw=black, shape=circle, tikzit category=SZX, tikzit fill=green, tikzit draw=black, line width=1pt, inner sep=2pt]
\tikzstyle{gred}=[fill=red, draw=black, shape=circle, rounded corners=2mm, tikzit category=SZX, inner sep=2pt, tikzit fill=red, line width=1pt]
\tikzstyle{ghad}=[fill=yellow, draw=black, shape=rectangle, tikzit category=SZX, tikzit shape=rectangle, tikzit fill=yellow, inner sep=2pt, line width=1pt]
\tikzstyle{boxm}=[fill=white, draw=black, rounded rectangle, tikzit category=scal, tikzit shape=rectangle, rounded rectangle east arc=none]
\tikzstyle{box}=[fill=white, draw=black, shape=rectangle]
\tikzstyle{had}=[fill=yellow, draw=black, shape=rectangle, tikzit category=ZX, tikzit fill=yellow, tikzit draw=black, inner sep=2pt]
\tikzstyle{gwhite}=[fill=white, draw=black, shape=circle, tikzit fill=white, tikzit shape=circle, line width=1 pt, inner sep=2 pt, tikzit draw=red]
\tikzstyle{gblack}=[fill=black, draw=black, shape=circle, tikzit fill=black, tikzit shape=circle, line width=1 pt, inner sep=2 pt, tikzit draw=red]
\tikzstyle{antipode}=[fill=red, draw=black, shape=rectangle, tikzit fill=red, tikzit draw=black, tikzit shape=rectangle, inner sep=2pt]
\tikzstyle{diamond}=[fill=white, draw=black, shape=diamond, inner sep=2pt]
\tikzstyle{mongr}=[fill=green, draw=green, shape=circle, inner sep=2pt]
\tikzstyle{monbl}=[fill=blue, draw=black, shape=circle, inner sep=2pt]
\tikzstyle{bg}=[inner sep=0.7mm, minimum width=0pt, minimum height=0pt, fill=green, draw=white, very thick, shape=circle]
\tikzstyle{br}=[inner sep=0.7mm, minimum width=0pt, minimum height=0pt, fill=red, draw=white, very thick, shape=circle]
\tikzstyle{rmat}=[draw, signal, fill=gray, signal to=east, signal from=west, inner sep=1pt, minimum height=6pt]
\tikzstyle{lmat}=[draw, signal, fill=gray, signal to=west, signal from=east, inner sep=1pt, minimum height=6pt]
\tikzstyle{umat}=[draw, signal, fill=gray, signal to=north, signal from=south, inner sep=1pt, minimum width=6pt]
\tikzstyle{dmat}=[draw, signal, fill=gray, signal to=south, signal from=north, inner sep=1pt, minimum width=6pt]
\tikzstyle{new style 0}=[fill=white, draw=black, shape=circle, inner sep=2pt, minimum width=5mm]

\tikzstyle{simple}=[-]
\tikzstyle{hadamard edge}=[-, dashed, dash pattern=on 2pt off 0.5pt, thick, draw={rgb,255: red,68; green,136; blue,255}]
\tikzstyle{box edge}=[-, dashed, dash pattern=on 2pt off 0.5pt, thick, draw={rgb,255: red,203; green,192; blue,225}]
\tikzstyle{brace edge}=[-, tikzit draw=blue, decorate, decoration={brace,amplitude=1mm,raise=-1mm}]
\tikzstyle{diredge}=[->, thick]
\tikzstyle{double edge}=[-, double, shorten <=-1mm, shorten >=-1mm, double distance=2pt]
\tikzstyle{gray edge}=[-, {gray!60!white}]
\tikzstyle{pointer edge}=[->, very thick, gray]
\tikzstyle{boldedge}=[-, line width=1.0pt, shorten <=-0.17mm, shorten >=-0.17mm]
\tikzstyle{bidir edge}=[<->, very thick, draw={rgb,255: red,191; green,191; blue,191}]
\tikzstyle{purple edge}=[->, thick, draw={rgb,255: red,225; green,117; blue,216}]
\tikzstyle{green edge}=[->, thick, draw={rgb,255: red,167; green,231; blue,137}]
\tikzstyle{orange edge}=[->, thick, draw={rgb,255: red,245; green,170; blue,63}]
\tikzstyle{blue edge}=[->, thick, draw={rgb,255: red,68; green,136; blue,255}]
\tikzstyle{any edge}=[->, thick, draw=cyan]
\tikzstyle{red edge}=[->, thick, draw={rgb,255: red,255; green,136; blue,136}]
\tikzstyle{bidiredge}=[<->, thick]
\tikzstyle{dashed diredge}=[->, dashed, dash pattern=on 1pt off 0.5pt]
\tikzstyle{bidashed diredge}=[<->, dashed, dash pattern=on 1pt off 0.5pt]
\tikzstyle{gray fill}=[-, fill={rgb,255: red,234; green,234; blue,234}, draw=black]
\tikzstyle{white fill}=[-, fill=white]
\tikzstyle{arrow}=[->]
\tikzstyle{very thick}=[-, line width=1pt, tikzit draw=red]
\tikzstyle{pointille}=[dashed, -]
\tikzstyle{red}=[-, draw=red]
\tikzstyle{blue}=[-, draw=blue]
\tikzstyle{green}=[-, draw=green]
\tikzstyle{arrow}=[->]
\tikzstyle{strike}=[-, tikzit draw={rgb,255: red,191; green,0; blue,64}, strike through]
\tikzstyle{strike'}=[-, tikzit draw=cyan, strike bend]
\tikzstyle{dashed no arrow}=[-, dashed, dash pattern=on 1pt off 0.5pt]

\newtheorem{lemma}{Lemma}[section]
\newtheorem{theorem}{Theorem}[section]
\newtheorem{definition}{Definition}[section]
\newtheorem*{problem}{Problem}
\newtheorem*{theorem*}{Theorem}

\title{Clifford Circuit Synthesis for Distributed Quantum \newline Architectures with Arbitrary Network Topology}
\author{Tuomas Laakkonen {\small (\href{mailto:tsrl@mit.edu}{tsrl@mit.edu})}}
\affiliation{Massachusetts Institute of Technology}

\begin{document}

\maketitle

\begin{abstract}
    \noindent To achieve large-scale fault-tolerant quantum computation, it may be easier to combine many small sets of qubits than to construct a single large set. For example via quantum error correction with block codes, or distributed quantum processors utilizing shared entanglement. In these regimes, the time or error budget of the overall quantum computation may be dominated by non-local operations. Hence, it is worthwhile to minimize the number of these operations. We consider the case where both non-local and local connectivity may be arbitrarily restricted, and give an asymptotically optimal synthesis method for distributed CNOT and Clifford circuits, based on block-matrix Gaussian elimination. We extend this to all Clifford+R\textsubscript{Z} circuits by generalizing the Pauli exponential circuit representation; this naturally integrates with existing methods for optimizing T-count. As an application, we show how to implement CNOT circuits in a CSS code encoding $n$ logical qubits in $k$ blocks using $O(nk)$ inter-block transversal CNOTs and intra-block Pauli measurements.
\end{abstract}

\section{Introduction}

In preparation for the availability of fault-tolerant quantum computers, it is important to study how quantum circuits can be best compiled for these systems. This is a highly non-trivial process which will include many levels, from error-minimizing physical circuit synthesis, to fault-tolerant implementations of operations for error-correcting codes, to the efficient synthesis of logical-level application circuits. In this work we will look at a particular problem that sits between these last two levels, which is to compile circuits for \emph{distributed} quantum architectures. This is motivated by the following reasoning: to achieve large-scale fault-tolerant quantum computation, it may be easier to combine many small sets of qubits than to construct a single large set. 

Candidates for architectures which have this structure include those based on distributed quantum computing \cite{caleffi2024distributed}, which has been recently demonstrated practically \cite{main2025distributed}, where physically distinct quantum processors are networked together via shared entanglement, and those based on error correction with block codes, where groups of logical qubits are encoded together, and there is a distinction between operations performed within a code block versus those that are performed between code blocks. One recent proposal for an early fault-tolerant quantum processor is the bicycle architecture \cite{yoder2025tour}, which combines both of these by utilizing blocks of the bivariate bicycle code that are networked together with Bell pairs. In systems with qubits partitioned in this way, which we call \emph{distributed}, the quantum gates that operate between partitions may require special handling, and in many cases should be avoided where possible (for instance, in the bicycle architecture these operations have error rates two orders of magnitude higher than other operations \cite[Table 2]{yoder2025tour}). We call such gates \emph{non-local}, and aim to optimize circuits so as to minimize them. 

\subsection{Prior Work}

Previous approaches to optimizing distributed circuits have focused mainly on the scheduling and placement of qubits and gates. In such schemes, the circuits itself is not drastically modified, but instead quantum teleportation of both qubits, gates, and subcircuits are performed using Bell pairs. This is known by various names, for example `tele-gate' and `tele-data' in \cite{vanmeter2006distributed}, and `embedding' in \cite{andresmartinez2024distributing}. The aim of this approach is then to reduce the number of Bell pairs necessary, and an important part of this is deciding on how the qubits should be placed, and given such a placement, when and to where gates or qubits should be teleported. A variety of techniques have been proposed to solve this, for example using (hyper)graph partitioning to optimize qubit placement \cite{andresmartinez2024distributing,wu2023entanglement}, covering problems to optimize where to teleport gates \cite{wu2023entanglement,sundaram2022distribution}, and Steiner trees to handle architectures where not all partitions can communicate \cite{andresmartinez2024distributing}. 

By contrast, our approach is fundamentally different, and aims to generate circuits with \emph{fewer} non-local gates in the first instance, rather than to provide an optimal implementation of the non-local gates already present in a circuit. The non-local gates used in our representation can be implemented using a teleportation protocol with a single Bell pair, but we do not attempt to optimize it further than this. Therefore, we do not aim to replace these approaches but to complement them; our algorithms can be easily combined with existing distribution methods.

\subsection{Our Approach}

\begin{figure}[t]
    \resizebox{\textwidth}{!}{$$\input{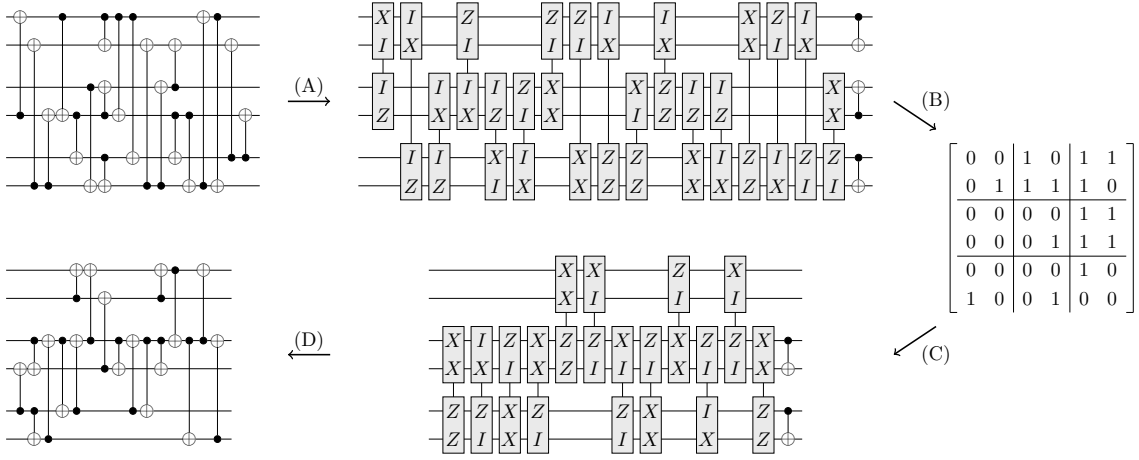}$$}
    \caption{An example of how the a CNOT circuit is optimized using the methods presented in this paper. We have an architecture with three partitions of two qubits, and the partitions have linear nearest neighbor connectivity. \emph{(A)} We start with a CNOT circuit that is converted into our circuit representation using Algorithm \ref{alg:localcliffpush}. Grey boxes represent generalized CNOT gates, and are the only non-local gates remaining. Note that at this point there may be non-local gates that do not respect the connectivity constraints. \emph{(B)} The non-local portion of the circuit is converted into a block parity matrix. \emph{(C)} Using \textsc{BlockRowCol}, a new circuit is synthesized from the parity matrix. This reduces the number of non-local gates while also respecting the connectivity constraints. \emph{(D)} Using Algorithm \ref{alg:resynth}, this is converted back into a circuit of CNOT gates, without increasing the number of non-local gates.}\label{fig:overview-cnot}
\end{figure}

In this work, we tackle the following problem for CNOT, Clifford, and Clifford+R\textsubscript{Z} circuits:
\begin{problem}
    Given a circuit $C$, find an equivalent circuit $C_\mathrm{opt} \equiv C$ which has fewer non-local gates than $C$ and respects the connectivity constraints of the architecture.
\end{problem}

The way that we will accomplish this is to pick a circuit representation which abstracts away \emph{all} local Clifford gates, leaving us with only non-local or non-Clifford gates, and see how existing algorithms for circuit resynthesis can be adapted to operate directly on this representation. In particular, we generalize the Pauli exponential representation of circuits \cite[Section 7.3]{kissingerwetering2024book}, and consider circuits made of two types of gates: \emph{Pauli exponentials} (below left) which will encode non-Clifford operations, and \emph{generalized CNOT gates} (below right) which will encode non-local operations. These are multi-qubit gates parameterized by two Pauli strings (tensor products of Pauli matrices) $A$ and $B$, and an angle $\theta$ (drawn in the style of \cite{litinski2019game}):
$$ P(A, \theta) = e^{-i\frac{\theta}{2}A} = \begin{tikzpicture}
	\begin{pgfonlayer}{nodelayer}
		\node [style=none] (0) at (-0.375, 1.25) {};
		\node [style=none] (1) at (0.625, 1.25) {};
		\node [style=none] (2) at (-0.375, -1.25) {};
		\node [style=none] (3) at (0.625, -1.25) {};
		\node [style=none] (5) at (0.125, 0) {$A$};
		\node [style=none] (6) at (-0.375, 0.75) {};
		\node [style=none] (7) at (-0.375, -0.75) {};
		\node [style=none] (8) at (0.625, 0.75) {};
		\node [style=none] (9) at (0.625, -0.75) {};
		\node [style=none] (10) at (0.85, 0) {$\theta$};
		\node [style=none] (11) at (0.625, 0.375) {};
		\node [style=none] (12) at (0.625, -0.375) {};
		\node [style=none] (13) at (1.125, 0) {};
		\node [style=none] (14) at (1.75, -0.75) {};
		\node [style=none] (15) at (-1, -0.75) {};
		\node [style=none] (16) at (-1, 0.75) {};
		\node [style=none] (17) at (1.75, 0.75) {};
		\node [style=none] (18) at (1.5, 0.25) {$\vdots$};
		\node [style=none] (19) at (-0.75, 0.25) {$\vdots$};
	\end{pgfonlayer}
	\begin{pgfonlayer}{edgelayer}
		\draw [style=gray fill] (1.center)
			 to (3.center)
			 to (2.center)
			 to (0.center)
			 to cycle;
		\draw [style=gray fill] (12.center)
			 to [in=-90, out=0, looseness=1.50] (13.center)
			 to [in=0, out=90, looseness=1.50] (11.center);
		\draw (14.center) to (9.center);
		\draw (15.center) to (7.center);
		\draw (16.center) to (6.center);
		\draw (17.center) to (8.center);
		\draw (11.center) to (12.center);
	\end{pgfonlayer}
\end{tikzpicture} \qquad C(A, B) = e^{i\frac{\pi}{4}(A\otimes I + I \otimes B - A \otimes B - I \otimes I)} = \begin{tikzpicture}
	\begin{pgfonlayer}{nodelayer}
		\node [style=none] (0) at (-1.25, 1.5) {};
		\node [style=none] (1) at (-1.25, 0.5) {};
		\node [style=none] (2) at (-0.875, 1.2) {$\vdots$};
		\node [style=none] (3) at (-0.5, 1.75) {};
		\node [style=none] (4) at (0.5, 1.75) {};
		\node [style=none] (5) at (-0.5, 0.25) {};
		\node [style=none] (6) at (0.5, 0.25) {};
		\node [style=none] (7) at (0, 1) {$A$};
		\node [style=none] (8) at (1.25, 1.5) {};
		\node [style=none] (9) at (1.25, 0.5) {};
		\node [style=none] (10) at (0.875, 1.2) {$\vdots$};
		\node [style=none] (11) at (-0.5, 1.5) {};
		\node [style=none] (12) at (-0.5, 0.5) {};
		\node [style=none] (13) at (0.5, 0.5) {};
		\node [style=none] (14) at (0.5, 1.5) {};
		\node [style=none] (45) at (0, 0.25) {};
		\node [style=none] (46) at (-1.25, -1.5) {};
		\node [style=none] (47) at (-1.25, -0.5) {};
		\node [style=none] (48) at (-0.875, -0.8) {$\vdots$};
		\node [style=none] (49) at (-0.5, -1.75) {};
		\node [style=none] (50) at (0.5, -1.75) {};
		\node [style=none] (51) at (-0.5, -0.25) {};
		\node [style=none] (52) at (0.5, -0.25) {};
		\node [style=none] (53) at (0, -1) {$B$};
		\node [style=none] (54) at (1.25, -1.5) {};
		\node [style=none] (55) at (1.25, -0.5) {};
		\node [style=none] (56) at (0.875, -0.8) {$\vdots$};
		\node [style=none] (57) at (-0.5, -1.5) {};
		\node [style=none] (58) at (-0.5, -0.5) {};
		\node [style=none] (59) at (0.5, -0.5) {};
		\node [style=none] (60) at (0.5, -1.5) {};
		\node [style=none] (61) at (0, -0.25) {};
	\end{pgfonlayer}
	\begin{pgfonlayer}{edgelayer}
		\draw [style=gray fill] (5.center)
			 to (3.center)
			 to (4.center)
			 to (6.center)
			 to cycle;
		\draw (11.center) to (0.center);
		\draw (12.center) to (1.center);
		\draw (13.center) to (9.center);
		\draw (14.center) to (8.center);
		\draw [style=gray fill] (52.center)
			 to (51.center)
			 to (49.center)
			 to (50.center)
			 to cycle;
		\draw (57.center) to (46.center);
		\draw (58.center) to (47.center);
		\draw (59.center) to (55.center);
		\draw (60.center) to (54.center);
		\draw (45.center) to (61.center);
	\end{pgfonlayer}
\end{tikzpicture} $$
In Section \ref{sec:representation}, we will show how to convert any Clifford+R\textsubscript{Z} circuit into this representation. With this in hand, we can study how these gates act when applied to parity matrices, stabilizer tableaux, and phase polynomials \cite[Sections 4.1, 6.4, 7.1]{kissingerwetering2024book}\cite{aaronson2004improved}, and use this to generalize three existing algorithms to output generalized CNOT gates directly:
\begin{itemize}
    \item For CNOT circuits, we show how \textsc{RowCol} \cite{wu2023optimization} can be made to operate on block matrices, and how non-local operations connect with low-rank submatrices, giving an algorithm we call \textsc{BlockRowCol} (Section \ref{sec:cnot}).
    \item For Clifford circuits, we adapt the method of \cite{winderl2023architecture} by showing how stabilizer tableaux can be disentangled into per-partition sub-tableaux, leading to an algorithm we call \textsc{DistRowCol} (Section \ref{sec:clifford}).
    \item For Clifford+R\textsubscript{Z} circuits, we give a modification of the standard Pauli exponential resynthesis method \cite{nash2020quantum}, with heuristics adapted from \cite{vandaele2022phase}, and generalize the phase-folding optimization \cite[Section 7.4.3]{kissingerwetering2024book} (Section \ref{sec:clifft}).
\end{itemize}
Our goal with this is not just to provide specific optimized routines for circuit resynthesis in distributed settings, but also to convince you that essentially \emph{any} existing method based on parity matrices, stabilizer tableaux, or phase polynomials ought to be easily adaptable to use these block-wise operations.

\subsection{Overview of Results}

Our most important theoretical result is the following:
\begin{theorem*}
    \textsc{BlockRowCol} and \textsc{DistRowCol} use at most $2n(k - 1)$ non-local gates to synthesize a CNOT or Clifford circuit on $n$ qubits divided into $k$ equal-sized partitions, regardless of the connectivity between partitions. Moreover, this is asymptotically optimal whenever $k = o(n / \log n)$. For $k = 2$, with any CNOT circuit as input, \textsc{BlockRowCol} uses at most twice as many non-local gates as the optimal equivalent CNOT circuit.
\end{theorem*}

In practice, we benchmark our algorithms against the implementation of \cite{andresmartinez2024distributing} given in the \texttt{pytket-dqc} Python package (Section \ref{sec:benchmarks}). We find that for CNOT and Clifford circuits our algorithms usually have fewer or comparable numbers of non-local gates to the circuits obtained by \texttt{pytket-dqc} both on its own and when combined with a standard resynthesis method \cite{wu2023optimization,nash2020quantum,vandenberg2020simple}. For Clifford+T our method can be significantly worse for some architectures, but our results suggest that this is due to the general inefficiency of any kind of circuit resynthesis in this setting.

We also see that because our method is simple and based on tableaux, which are memory-efficient and lend themselves to efficient implementation, our methods are fast enough to scale to large instances. For instance, we successfully ran \textsc{BlockRowCol} on instances with up to $512$ qubits and $128$ partitions, with circuits of more than $2 \cdot 10^5$ gates. Our implementation is available publicly \cite{github}, and is \emph{not} particularly well-optimized; we expect that with more effort, it would be possible to reach the scale of thousands of qubits, which is the scale required for interesting quantum algorithms \cite{gidney2025how}. Finally, to provide further motivation for why these techniques may be interesting, we discuss in Section \ref{sec:discussion} some applications to practical architectures, including:
\begin{enumerate}
    \item We show how to synthesize a logical CNOT circuit in a CSS code using inter-block transversal CNOTs and intra-block Pauli measurements. In particular, for the recently introduced \emph{phantom codes} \cite{koh2026entangling}, this has a complexity that scales quadratically in the number of code blocks, not the number of logical qubits.
    \item We consider the problem of classically simulating quantum circuits with tensor networks, and show that a variation on \textsc{BlockRowCol} can enable asymptotically improved application of CNOT circuits to a tree tensor network \cite{seitz2023simulating}.
\end{enumerate}
We note that the methods presented here are preliminary, and focus on how this abstracted representation of non-local gates can be applied to existing methods. There are many practical details that could be optimized heuristically to provide better results, or extensions that could be made. We conclude in Section \ref{sec:future} by discussing some of these.

\clearpage

\tableofcontents 

\clearpage

\section{Circuit Representation}
\label{sec:representation}

In this section, we define a representation of distributed Clifford+R\textsubscript{Z} circuits that abstracts away all local Clifford operations. This is a generalization of the usual Pauli exponential representation which has been used in various compilation procedures (e.g \cite{nash2020quantum}). We will represent a circuit as a combination of two types of operations: Pauli exponentials that act locally, and generalized CNOT gates that act non-locally. These then are organized in a directed acyclic graph (DAG) which captures the structure of the circuit modulo any possible pairwise commutation of gates. The defining motivation of this representation is that we will show that any generalized CNOT gate can be implemented with \emph{exactly one non-local CNOT gate}.  

Throughout, we will assume that we are compiling for a distributed quantum computer where the qubits are divided into $P$ \emph{partitions}. Each partition $i$ contains some number of qubits, and we write $q_{ij}$ for the $j$th qubit of partition $i$. We call a gate \emph{non-local} if it acts on some qubits $q_{i_1j_1}$ and $q_{i_2j_2}$ where $i_1 \neq \pm I_2$ (i.e not within the same partition), and \emph{local} otherwise. For any circuit $C$ acting on $\{q_{i,j}\}$ we will define in the following subsections its representation, given by
$$ C \equiv (\mathcal{D}(C_N), C_L) $$
where $\mathcal{D}(C_N)$ is a DAG that encodes the non-local and non-Clifford operations, and $C_L$ is a circuit containing only local Clifford operations.

\subsection{Pauli Exponentials}

Before defining the generalized gates that we will use to build our circuit representation, we will first set up our notation for the usual Pauli exponential representation. Let $\mathcal{P} = \{ i^k M \mid k \in \mathbb{Z}, M \in \{I, X, Y, Z\}\}$ be the group generated by the Pauli matrices:
$$ I = \begin{pmatrix}1&0\\0&1\end{pmatrix},\quad X = \begin{pmatrix}0&1\\1&0\end{pmatrix},\quad Y = \begin{pmatrix}0&-i\\i&0\end{pmatrix},\quad Z = \begin{pmatrix}1&0\\0&-1\end{pmatrix} $$
Furthermore, let $\mathcal{P}_n = \{ P_1 \otimes \cdots \otimes P_n \mid P_1, \dots, P_n \in \mathcal{P} \}$ be the group of $n$-fold tensor products of elements of $\mathcal{P}$. The $n$-qubit Clifford group $\mathcal{C}_n$ is defined \cite[Section 3.5]{gottesman2009introduction} as the normalizer of $\mathcal{P}_n$ within the $2^n$-dimensional unitary group $U(2^n)$:
$$ \mathcal{C}_n = \{ V \in U(2^n) \mid \forall P \in \mathcal{P}_n, ~V P V^\dagger \in \mathcal{P}_n \} $$ 
Up to a global phase factor, this group is equivalent to the circuits consisting of CNOT, Hadamard, and $S$-gates, hence we will call elements of $\mathcal{C}_n$ Clifford operations or circuits.

Now note that $\mathcal{P}$ consists of four copies of the Pauli matrices, multiplied by the phases $\pm 1$ and $\pm i$; in the same way, $\mathcal{P}_n$ consists of four copies of all $n$-fold tensor products of Pauli matrices. An $i$-phase will appear when multiplying any pair of elements that do not commute, e.g $XY = iZ$. However, a crucial property of $\mathcal{C}_n$ is that, when an element $V \in \mathcal{C}_n$ acts on $P \in \mathcal{P}_n$ as $VPV^\dagger = Q \in \mathcal{P}_n$, an $i$-phase is \emph{never} introduced. That is, for Pauli matrices $P_1, \dots, P_n$, we have
$$ V(P_1 \otimes \cdots \otimes P_n)V^\dagger = \pm Q_1 \otimes \cdots Q_n $$
for some Pauli matrices $Q_1, \dots, Q_n$. This property is used often when constructing stabilizer tableaux \cite{aaronson2004improved}, which track a set of Pauli matrices and signs $\pm 1$. 

This property then motivates us to define a smaller group $\mathcal{P}^\pm = \{ \pm I, \pm X, \pm Y, \pm Z \}$ of Pauli matrices with only $\pm 1$ phases. This is not a group with respect to standard matrix multiplication, but we instead define \emph{phaseless multiplication} as follows:
$$ \overline{PQ} = \begin{cases}
    PQ & \text{if } PQ = QP \\
    -iPQ & \text{if } PQ \neq QP
\end{cases} $$
As before, we can define the set $\mathcal{P}^\pm_n = \{ \pm M_1 \otimes \cdots \otimes M_n \mid M_i \in \{I, X, Y, Z\} \}$ of $n$-fold tensor products of $\mathcal{P}^\pm$, which is a group under phaseless multiplication. This is the smallest group for which $\mathcal{C}_n$ is the normalizer; if we were to remove the $\pm 1$ signs and consider only tensor products of Pauli matrices, this would no longer be true (e.g $HYH = -Y$). We will call the elements of $\mathcal{P}^\pm_n$ \emph{($n$-qubit) Pauli strings}. For Pauli strings with a single non-identity qubit we will write $P_k$ to indicate a $P$ matrix on the $k$th qubit and identity elsewhere.

Given a Pauli string $A$ and an angle $\theta$, we define the \emph{Pauli exponential} $P(A, \theta)$:
$$ P(A, \theta) = e^{-i\frac{\theta}{2}A}$$
We will write $P_i(A, \theta)$ when we want to mark that $A$ acts on the $i$th partition. Note that any circuit $C$ written \cite[Section 7.4.2]{kissingerwetering2024book} in the Clifford+T gateset can be presented as a product of Pauli exponentials $C = P(A_1, \frac{\pi}{4}k_1) \cdots P(A_m, \frac{\pi}{4}k_m)$ where each angle $\theta$ is an integer multiple of $\frac{\pi}{4}$ (or a multiple of $\frac{\pi}{2}$ in the case of Clifford circuits). In circuit diagrams we draw these operations in the style of \cite{litinski2019game}:
$$ P(\pm A_1 \otimes \cdots \otimes A_n, \theta) ~~=~~ \begin{tikzpicture}
	\begin{pgfonlayer}{nodelayer}
		\node [style=none] (0) at (0, 1) {};
		\node [style=none] (1) at (0, -1) {};
		\node [style=none] (2) at (0.5, 0.25) {$\vdots$};
		\node [style=none] (3) at (1, 1.5) {};
		\node [style=none] (4) at (2.25, 1.5) {};
		\node [style=none] (5) at (1, -1.5) {};
		\node [style=none] (6) at (2.25, -1.5) {};
		\node [style=none] (7) at (1.625, 1) {$A_1$};
		\node [style=none] (9) at (3.75, 1) {};
		\node [style=none] (10) at (3.75, -1) {};
		\node [style=none] (11) at (3.25, 0.25) {$\vdots$};
		\node [style=none] (12) at (1, 1) {};
		\node [style=none] (13) at (1, -1) {};
		\node [style=none] (14) at (2.25, -1) {};
		\node [style=none] (15) at (2.25, 1) {};
		\node [style=none] (16) at (1.625, -1) {$A_n$};
		\node [style=none] (17) at (1.625, 0.25) {$\vdots$};
		\node [style=none] (18) at (2.55, -0.375) {$\theta$};
		\node [style=none] (19) at (2.25, 0.75) {};
		\node [style=none] (20) at (2.85, 0.5) {};
		\node [style=none] (21) at (2.85, -0.5) {};
		\node [style=none] (22) at (2.25, -0.75) {};
		\node [style=none] (23) at (2.55, 0.375) {$\pm$};
		\node [style=none] (24) at (2.25, 0) {};
		\node [style=none] (25) at (2.85, 0) {};
	\end{pgfonlayer}
	\begin{pgfonlayer}{edgelayer}
		\draw [style=gray fill] (5.center)
			 to (3.center)
			 to (4.center)
			 to (6.center)
			 to cycle;
		\draw (12.center) to (0.center);
		\draw (13.center) to (1.center);
		\draw (14.center) to (10.center);
		\draw (15.center) to (9.center);
		\draw [style=gray fill] (19.center)
			 to (22.center)
			 to [in=-90, out=0, looseness=1.25] (21.center)
			 to (20.center)
			 to [in=0, out=90, looseness=1.25] cycle;
		\draw [style=gray fill] (25.center) to (24.center);
	\end{pgfonlayer}
\end{tikzpicture} $$
Pauli exponentials satisfy some general relations, including:
\begin{align*}
    P(A, \theta_1)P(A, \theta_2) &= P(A, \theta_1 + \theta_2) \\
    P(-A, \theta) &= P(A, -\theta) \\
    P(A, 0) &= I 
\end{align*}
And they commute whenever the corresponding Pauli strings commute:
$$ P(A_1, \theta_1)P(A_2, \theta_2) = P(A_2, \theta_2)P(A_1, \theta_1) \quad\iff\quad A_1A_2 = A_2A_1 $$
The properties have been used to optimize Clifford+T circuits via the phase-folding \cite[Section 7.4.3]{kissingerwetering2024book} and phase-teleportation \cite{kissinger2020reducing,zhang2019optimizing} methods.

\subsection{Generalized CNOT Gates}

Now we define generalized CNOT gates as follows. Let $A \in \mathcal{P}^\pm_n$ and $B \in \mathcal{P}^\pm_m$ be Pauli strings acting on two separate partitions of sizes $n$ and $m$, then we define the generalized CNOT gate $C(A, B)$ as
$$ C(A, B) = e^{i\frac{\pi}{4} (A \land B)} \qquad A \land B = A \otimes I + I \otimes B - A \otimes B - I \otimes I$$
which, since all of the terms of $A \land B$ commute, can also be written as:
\begin{align*}
C(A, B) &= \frac{e^{-i\frac{\pi}{4}}}{2\sqrt{2}}\left[I\otimes I + iA \otimes I\right]\left[I\otimes I + iI \otimes B\right]\left[I \otimes I - iA \otimes B\right] \\ 
&= \frac{I \otimes I + A \otimes I + I \otimes B - A \otimes B}{2} = I \otimes I + \frac{A \land B}{2} 
\end{align*}
If $A$ and $B$ act on the $i$th and $j$th partitions, then we will write $C_{ij}(A, B)$. When $A$ and $B$ are single-qubit, this generalizes the CNOT gate - for example, we have:
$$ C(Z, Z) = CZ \quad C(Z, X) = CNOT \quad C(X, Z) = SWAP \cdot CNOT \cdot SWAP$$
These operations have been defined previously (up to a global phase) in \cite{grier2022classification}, and they are called $A$-controlled-$B$ gates in \cite{litinski2019game}. In circuit diagrams we will draw these operations again in the style of \cite{litinski2019game}:
$$ C(A, B) ~~=~~ \begin{tikzpicture}
	\begin{pgfonlayer}{nodelayer}
		\node [style=none] (0) at (-1.25, 1.5) {};
		\node [style=none] (1) at (-1.25, 0.5) {};
		\node [style=none] (2) at (-0.875, 1.2) {$\vdots$};
		\node [style=none] (3) at (-0.5, 1.75) {};
		\node [style=none] (4) at (0.5, 1.75) {};
		\node [style=none] (5) at (-0.5, 0.25) {};
		\node [style=none] (6) at (0.5, 0.25) {};
		\node [style=none] (7) at (0, 1) {$A$};
		\node [style=none] (8) at (1.25, 1.5) {};
		\node [style=none] (9) at (1.25, 0.5) {};
		\node [style=none] (10) at (0.875, 1.2) {$\vdots$};
		\node [style=none] (11) at (-0.5, 1.5) {};
		\node [style=none] (12) at (-0.5, 0.5) {};
		\node [style=none] (13) at (0.5, 0.5) {};
		\node [style=none] (14) at (0.5, 1.5) {};
		\node [style=none] (45) at (0, 0.25) {};
		\node [style=none] (46) at (-1.25, -1.5) {};
		\node [style=none] (47) at (-1.25, -0.5) {};
		\node [style=none] (48) at (-0.875, -0.8) {$\vdots$};
		\node [style=none] (49) at (-0.5, -1.75) {};
		\node [style=none] (50) at (0.5, -1.75) {};
		\node [style=none] (51) at (-0.5, -0.25) {};
		\node [style=none] (52) at (0.5, -0.25) {};
		\node [style=none] (53) at (0, -1) {$B$};
		\node [style=none] (54) at (1.25, -1.5) {};
		\node [style=none] (55) at (1.25, -0.5) {};
		\node [style=none] (56) at (0.875, -0.8) {$\vdots$};
		\node [style=none] (57) at (-0.5, -1.5) {};
		\node [style=none] (58) at (-0.5, -0.5) {};
		\node [style=none] (59) at (0.5, -0.5) {};
		\node [style=none] (60) at (0.5, -1.5) {};
		\node [style=none] (61) at (0, -0.25) {};
	\end{pgfonlayer}
	\begin{pgfonlayer}{edgelayer}
		\draw [style=gray fill] (5.center)
			 to (3.center)
			 to (4.center)
			 to (6.center)
			 to cycle;
		\draw (11.center) to (0.center);
		\draw (12.center) to (1.center);
		\draw (13.center) to (9.center);
		\draw (14.center) to (8.center);
		\draw [style=gray fill] (52.center)
			 to (51.center)
			 to (49.center)
			 to (50.center)
			 to cycle;
		\draw (57.center) to (46.center);
		\draw (58.center) to (47.center);
		\draw (59.center) to (55.center);
		\draw (60.center) to (54.center);
		\draw (45.center) to (61.center);
	\end{pgfonlayer}
\end{tikzpicture} $$

The most important property of generalized CNOT gates is that they can be implemented using a single non-local CNOT operation. Given $C(A, B)$ we can always find \cite{vandenberg2020simple} local Clifford operations $U$ and $V$ such that $UAU^\dagger = I^{\otimes n-1} \otimes Z$ and $VBV^\dagger = X \otimes I^{\otimes m - 1}$. Then we have
\ctikzfig{gen-cnot-into-cz}
where only the CNOT operation is non-local. Therefore, we will aim to minimize the number of generalized CNOT gates in our circuits, since this minimizes the number of non-local gates.

We will now discuss some other useful properties of generalized CNOT gates. Since $C(A, B)$ gates are Clifford, they conjugate Pauli strings to Pauli strings, and in particular they act as
$$ C(A, B) \cdot [P \otimes Q] \cdot C(A, B) = \begin{cases}
    P \otimes Q & \text{if } AP = PA \land BQ = QB \\
    P \otimes \overline{QB} & \text{if } AP \neq PA \land BQ = QB \\
    \overline{PA} \otimes Q & \text{if } AP = PA \land BQ \neq QB \\
    \overline{PA} \otimes \overline{QB} & \text{if } AP \neq PA \land BQ \neq QB
\end{cases}$$
so we can interpret $C(A, B)$ as applying $B$ to the second partition conditioned on the first partition anticommuting with $A$, and vice-versa. In Sections \ref{sec:cnot} and \ref{sec:clifford} we will use this to determine how they act on parity matrices and stabilizer tableaux. Additionally, they commute when the Pauli strings in each partition commute individually
\begin{align*}
    C(A, B) \cdot C(P, Q) = C(P, Q) \cdot C(A, B) &\iff AP = PA \land BQ = QB \\
    [C(A, B) \otimes I] \cdot [I \otimes C(P, Q)] = [I \otimes C(P, Q)] \cdot [C(A, B) \otimes I] &\iff BP = PB
\end{align*}
and they commute with Pauli exponentials if the Pauli strings on the shared partition commute:
$$C(A, B) \cdot [P(C, \theta) \otimes I] = [P(C, \theta) \otimes I] \otimes C(A, B) \iff AC = CA $$
Any generalized CNOT gate where (at least) one Pauli string is the identity can be removed or simplified:
$$ C(A, I) = I \otimes I \qquad C(A, -I) = P(A, \pi) \otimes I = A \otimes I$$
Finally, any two generalized CNOT gates on the same partitions which are the same on one partition can be combined into a single generalized CNOT gate:
$$ C(A, B_1) \cdot C(A, B_2) = \begin{cases}
    C(A, \overline{B_1B_2}) & \text{if } B_1B_2 = B_2B_1 \\
    [P(A, \frac{\pi}{2}) \otimes I] \cdot C(A, \overline{B_1B_2}) & \text{if } B_1B_2 \neq B_2B_1
\end{cases}$$
The case where $B_1$ and $B_2$ commute can be thought of as a extension of the fact that CNOT is self-inverse, while the non-commuting case is a extension of the following identity
\ctikzfig{cnot-cz-combine}
which is written in terms of generalized CNOT gates $C(Z, Z) \cdot C(Z, X) = [P(Z, \frac{\pi}{2}) \otimes I] \cdot C(Z, Y)$. In Section \ref{sec:gencnotfold} we will use this to optimize the number of non-local gates with a method similar to phase-folding.

\subsection{Circuits up to Commutation via DAGs}

Directed acyclic graphs are commonly used to represent generic quantum circuits \cite{meijer2023comparison} by associating to each gate a vertex and connecting adjacent gates which share qubits via edges. This has the advantage of being invariant to trivial commutations of the gates. Similar to \cite{zhang2019optimizing} we will represent circuits of Pauli exponentials and generalized CNOT gates with a looser variant of this that captures all the pairwise commutation relations of the gates. In particular, suppose we have a circuit 
$$C = U_m U_{m-1} \cdots U_1 $$ 
comprised of $m$ Pauli exponential and generalized CNOT gates, then we define the commutation DAG of $C$ as the directed graph $\mathcal{D}(C) = (V, E)$ of $m$ vertices where:
$$V = \{ U_1, U_2, \dots, U_m \} \qquad E = \{ U_{i} \to U_{j} \mid i < j \text{ and } U_iU_{j} \neq U_{j}U_i \}$$
Whether each pair of gates commutes can be decided easily by the relations given in the previous subsection. We will now show that arranging the gates of $C$ into any topological order of $\mathcal{D}(C)$ yields a circuit that is equal to $C$. Hence, we can always compute $\mathcal{D}(C)$ and discard the original circuit $C$, since this keeps enough information to recover an equivalent circuit.

\begin{lemma}
    \label{thm:commdag}
    Suppose $\mathcal{T} = [U_{i_1}, U_{i_2}, \dots, U_{i_m}]$ is a topological order of $\mathcal{D}(C)$, then $U_{i_m} \cdots U_{i_1} = C$.
\end{lemma}

\begin{proof}
    We can perform a bubble sort to transform $C' = U_{i_m} \cdots U_{i_2}U_{i_1}$ into $C$: suppose that $i_b < i_a$ but $b > a$, then we must have that $U_{i_a}U_{i_b} = U_{i_b}U_{i_a}$. If this were not the case, then $\mathcal{D}(C)$ would contain an edge $U_{i_b} \to U_{i_a}$, and so $\mathcal{T}$ could not be a topological order of $\mathcal{D}(C)$. Now while there is any $k$ such that $i_{k+1} < i_k$, swap the values of $i_k$ and $i_{k+1}$ and swap $U_{i_k}$ and $U_{i_{k+1}}$ in $C'$. This must keep $C'$ the same at every step since each $U_{i_k}$ and $U_{i_{k+1}}$ commute. When there are no such $k$ remaining then we must have $i_j < i_{j + 1}$ for all $j$; since each $i_j$ is distinct, this implies that $i_j = j$ and hence $C' = C$. 
\end{proof}

Note that the converse is \emph{not} true: there exists equivalent permutations of the gates of $C$ that are not captured by pairwise commutation relations; this is explored further in \cite{mori2025nontrivial}. For example, we have the following identity
$$ P(-ZI, \theta) P(IZ, \theta) \cdot P(-YX, \theta) P(-XY, \theta) = P(-YX, \theta) P(-XY, \theta) \cdot P(-ZI, \theta) P(IZ, \theta)$$
but at the same time $P(IZ, \theta)\cdot P(-YX, \theta) \neq P(-YX, \theta) \cdot P(IZ, \theta)$. 

\subsection{Translating Clifford+R\texorpdfstring{\textsubscript{Z}}{Z} Circuits}

With this in place, we can describe how to obtain the representation $(\mathcal{D}(C_N), C_L)$ from any Clifford+R\textsubscript{Z} circuit $C$. Similarly to the usual Pauli exponential compilation, we will first `push' all the local Clifford gates to the end of the circuit to obtain $C = C_L \cdot C_N$ where $C_N$ contains only Pauli exponentials and generalized CNOT gates. To do this we first replace every non-Clifford or non-local gate with a Pauli exponential or generalized CNOT gate, and then make use of the following relations
\begin{align*}
    U^\dagger \cdot P(A, \theta) \cdot U &= P(U^\dagger AU, \theta) \\
    [U^\dagger \otimes I] \cdot C(A, B) \cdot [U \otimes I] &= C(U^\dagger A U, B)
\end{align*}
where $U$ is any local Clifford operation. After we have obtained $C_N$ and $C_L$, $\mathcal{D}(C_N)$ can be computed directly from its definition. This procedure is formalized in Algorithm \ref{alg:localcliffpush}, and an example is given in Figure \ref{fig:cliff-push-example}. In Section \ref{sec:resynth}, we discuss how the opposite transformation is performed.

\begin{algorithm}
\caption{A procedure to determine $(\mathcal{D}(C_N), C_L)$ given a Clifford+R\textsubscript{Z} circuit $C$.} \label{alg:localcliffpush}
\KwIn{A Clifford+R\textsubscript{Z} circuit $C = U_m\cdots U_2U_1$.}
\KwOut{A DAG $\mathcal{D}(C_N)$ and a circuit $C_L$ of local Clifford gates.}

$C_N \gets$ an empty circuit\;
$C_L \gets$ an empty circuit\;

\For{$k = m$ \KwTo $1$}{
    \uIf{$U_k = CNOT(q_{i_1j_1},q_{i_2j_2})$ is a non-local CNOT gate}{
        Prepend $C_{i_1i_2}(Z_{j_1}, X_{j_2})$ to $C_N$\;
    }\uElseIf{$U_k = R_Z(q_{ij}, \theta)$ is a non-Clifford R\textsubscript{Z} gate}{
        Prepend $P_{i}(Z_j, \theta)$ to $C_N$\;
    }\ElseIf{$U_k$ is a local Clifford gate}{
        \For{$V_l \in C_N$}{
            $V_l \gets U_k^\dagger V_l U_k$\;
        }
        Prepend $U_k$ to $C_L$\;
    }
}

$V \gets \{ V_l \mid V_l \in C_N \}$\;
$E \gets \{ V_i \to V_j \mid i < j \text{ and } V_iV_j \neq V_jV_i \}$\;
$\mathcal{D}(C_N) \gets (V, E)$\;
\Return{$(\mathcal{D}(C_N), C_L)$}
\end{algorithm}

\begin{figure}
    \resizebox{\textwidth}{!}{$$\input{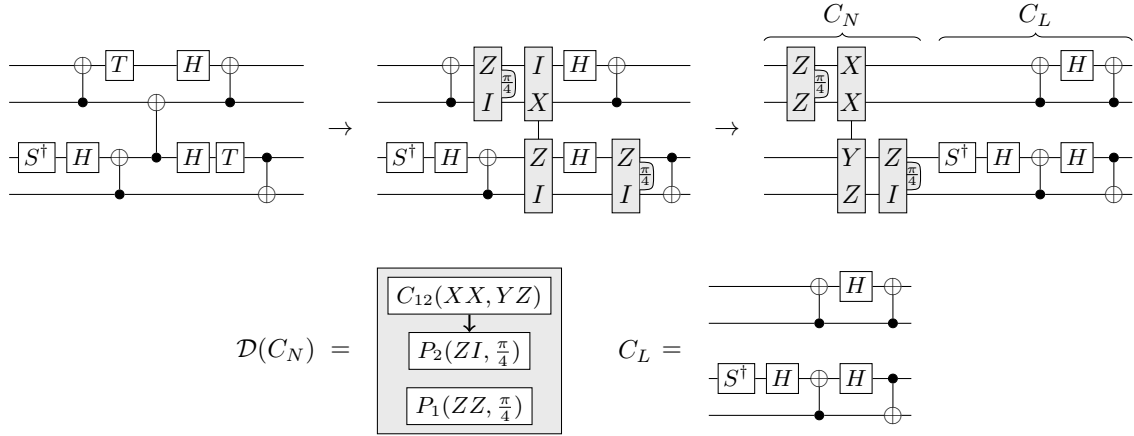}$$}
    \caption{An example of translating a Clifford+T circuit into the representation defined in Section \ref{sec:representation}, following Algorithm \ref{alg:localcliffpush}. On the top left the input circuit $C$ is shown, consisting of two partitions of two qubits. First, we replace non-Clifford and non-local CNOT gates with Pauli exponentials and generalized CNOT gates, respectively. Then we `push' all of the local Clifford gates to the end of the circuit. This splits $C$ into a local Clifford part $C_L$ and a non-local or non-Clifford part $C_N$, from which we construct the representation $(\mathcal{D}(C_N), C_L)$.}
    \label{fig:cliff-push-example}
\end{figure}

\section{CNOT Circuit Synthesis}
\label{sec:cnot}

In this section, we introduce \textsc{BlockRowCol}, an algorithm for synthesizing distributed CNOT circuits using generalized CNOT gates. It is based on the \textsc{RowCol} algorithm for CNOT circuits introduced in \cite{wu2023optimization} and can synthesize circuits for distributed architectures with equal-sized partitions and arbitrarily restricted connectivity using the technique of Steiner trees introduced by \cite{kissinger2019cnot}. We will presume familiarity with the parity matrix formulation of CNOT circuit synthesis, see \cite[Section 4.1]{kissingerwetering2024book} for an introduction. The key insight is in three parts:
\begin{enumerate}
    \item The parity matrix corresponding to a distributed CNOT circuit can be seen as a block matrix where each block represents the interaction between two partitions.
    \item If we can reduce the parity matrix so that every off-diagonal block is zero, then local gates suffice to synthesize the remaining circuit.
    \item A subset of generalized CNOT gates, that we will call \emph{ZX-type}, act on the parity matrix by performing a rank-one update to a row of blocks.
\end{enumerate}
Combining these ideas, we show that Gaussian elimination can be adapted to perform elimination of the off-diagonal blocks directly. Whenever the number of partitions is small compared to the number of qubits, the resulting algorithm is asymptotically more efficient in terms of non-local gate count than Patel-Markov-Hayes' algorithm \cite{patel2003efficient} for standard CNOT circuit synthesis. 

We start in Section \ref{sec:rowcolinterp} by giving a reinterpretation and simplification of \textsc{RowCol} before showing in Section \ref{sec:blockrowcol} how to adapt this to construct \textsc{BlockRowCol}. Finally in Section \ref{sec:blockrowcol-optimality} we will see that it is optimal in certain settings.

\subsection{Reinterpreting RowCol}
\label{sec:rowcolinterp}

The \textsc{RowCol} algorithm for synthesizing CNOT circuits in constrained architectures was introduced in \cite{wu2023optimization}. The algorithm proceeds by disentangling a qubit from the remaining qubits, one at a time. In terms of the parity matrix, the disentangling operation corresponds to eliminating a column and the corresponding row.

While columns are eliminated using standard Gaussian elimination, eliminating rows in \textsc{RowCol} is more complicated, and requires setting up and solving a linear system for each row. In \cite{winderl2023architecture}, a variant of \textsc{RowCol} for Clifford circuits was introduced which sidesteps this complication, by relying on the both the stabilizer and destabilizer portions of the stabilizer tableau. In this section we will reformulate \textsc{RowCol} for CNOT circuits using this method. The advantage is that the final algorithm only performs standard Gaussian elimination, which can be adapted to distributed circuits, as we will see in in Section \ref{sec:blockrowcol}. 

To start, we will see how row operations applied to a parity matrix affect the transpose of the inverse matrix. The motivation here is that for any CNOT circuit with parity matrix $M$, the stabilizer tableau has the form:
$$ \left(\begin{array}{cc}
    (M^{-1})^T & 0 \\
    0 & M
\end{array}\right) $$

\begin{lemma}
    \label{thm:rowopsinvtransp}
    Applying a row operation $M_i \gets M_i \oplus M_j$ to a parity matrix $M$ applies the opposite operation $M'_j \gets M'_j \oplus M'_i$ to the inverse transpose $M' = (M^{-1})^T$. 
\end{lemma}

\begin{proof}
    Let $R_{ij}$ be the elementary matrix that represents adding row $j$ to row $i$ when multiplied to the left of a matrix. Then we have $R_{ij}^{-1} = R_{ij}$ and $R_{ij}^T = R_{ji}$. Thus $((R_{ij} \cdot M)^{-1})^T = (M^{-1} \cdot R_{ij}^{-1})^T = (M^{-1} \cdot R_{ij})^T = R_{ij}^T \cdot (M^{-1})^T = R_{ji} \cdot M'$.
\end{proof}

Next we see that block structure in $M$ will reflect block structure in $M'$: in particular, if a block is zero in $M$, then the block reflected about the diagonal in $M'$ will also be zero.
\begin{lemma}
    \label{thm:zeroblockinv}
    For compatible matrices $A$, $B$, and $C$ with $A$ and $C$ invertible, we have
    $$ M = \begin{pmatrix} A & B \\ 0 & C \end{pmatrix} \iff (M^{-1})^T = \begin{pmatrix} A' & 0 \\ B' & C' \end{pmatrix}$$
    where $A' = (A^{-1})^T$, $C' = (C^{-1})^T$, and $B' = C'B^TA'$.
\end{lemma}

\begin{proof}
    By block multiplication:
    \begin{equation*}
        \begin{pmatrix} A & B \\ 0 & C \end{pmatrix} \begin{pmatrix} A' & 0 \\ B' & C' \end{pmatrix}^T = \begin{pmatrix} A & B \\ 0 & C \end{pmatrix} \begin{pmatrix} A^{-1} & A^{-1}BC^{-1} \\ 0 & C^{-1} \end{pmatrix} = \begin{pmatrix} I & 0 \\ 0 & I \end{pmatrix} \qedhere
    \end{equation*}
\end{proof}

From this follows the key result that eliminating a column of $M$ will eliminate a row of $(M^{-1})^T$ and vice-versa. Hence, to eliminate \emph{both} a row and column of a parity matrix $M$, which disentangles one qubit from the rest, we eliminate a column of $M$ and then eliminate the same column of $(M^{-1})^T$. Lemma \ref{thm:zeroblockinv} guarantees that these two operations do not interfere with each other. This idea is formalized in Algorithm \ref{alg:rowcolcnot}, which depends on the subroutine \textsc{ElimCol} that performs the Gaussian elimination step, keeping both $M$ and $(M^{-1})^T$ up to date using Lemma \ref{thm:rowopsinvtransp}. In Algorithm \ref{alg:elimcolalltoall} we give a straightforward implementation of \textsc{ElimCol} for all-to-all connected architectures. For an $n$-qubit parity matrix, Algorithm \ref{alg:rowcolcnot} runs in $O(n^3)$ time, and produces a circuit of at most $n(n-1)$ CNOT gates. 

\begin{algorithm}
    \caption{A simplified version of the \textsc{RowCol} CNOT circuit synthesis algorithm.} \label{alg:rowcolcnot}
    \KwIn{A parity matrix $M$ representing the desired CNOT circuit.}
    \KwOut{A CNOT circuit $C$.}
    
    $C \gets$ an empty circuit\;
    $M' \gets (M^{-1})^T$\;

    \For{each qubit $q$}{
        $C_q \gets \textsc{ElimCol}(M, M', q, flip=\mathrm{false})$\;
        Prepend $C_q$ to $C$\;
        $C'_q \gets \textsc{ElimCol}(M', M, q, flip=\mathrm{true})$\;
        Prepend $C'_q$ to $C$\;
    }

    \Return{$C$}
\end{algorithm}

\begin{algorithm}
    \caption{The \textsc{ElimCol} subroutine for all-to-all connected architectures.}\label{alg:elimcolalltoall}
    \KwIn{A parity matrix $M$, its inverse transpose $M' = (M^{-1})^T$, and a column $q$. A flag $flip$ that determines whether CNOT gates should be flipped.}
    \KwOut{A CNOT circuit $C$ which eliminates the $q$th column of $M$.}

    $C \gets$ an empty circuit\;
    
    \If{$M_{qq} = 0$}{
        Let $k = \min\{ j \mid j \neq q \land M_{jq} = 1 \}$ be the pivot of the $q$th column of $M$\;
        Add row $k$ to row $q$ in $M$\;
        Add row $q$ to row $k$ in $M'$\;
        \lIf{$\lnot flip$}{prepend $CNOT_{k,q}$ to $C$ \lElse*{prepend $CNOT_{q,k}$ to $C$}}
    }
    \For{all $j \neq q$ such that $M_{jq} = 1$}{
        Add row $q$ to row $j$ in $M$\;
        Add row $j$ to row $q$ in $M'$\;
        \lIf{$\lnot flip$}{prepend $CNOT_{q,j}$ to $C$ \lElse*{prepend $CNOT_{j,q}$ to $C$}}
    }
    \Return{$C$}
\end{algorithm}

Adapting this to the case of restricted connectivity can be done using Steiner trees, as in \cite{kissinger2019cnot}. Since we will use Steiner trees throughout the following sections, we define them formally here. 

\begin{definition}
    Given a connected undirected graph $G = (V, E)$ and a subset $S \subseteq V$, a Steiner tree $T = (V_T, E_T)$ for $S$ in $G$ is a connected subgraph of $G$ with the fewest edges such that $S \subseteq V_T$.
\end{definition}

While computing Steiner trees is NP-hard in general, polynomial-time algorithms exist \cite{kou1981fast,robins2000improved} for finding a $c$-approximation (that is, a tree $T \subseteq G$ such that $V_T \subseteq S$ with at most $c$ times as many edges as a Steiner tree) for various values of $c \leq 2$. Therefore, whenever we write in an algorithm to calculate a Steiner tree, we will assume that an approximation may be used instead. 

Steiner trees can be used to eliminate columns of $M$ while respecting the qubit connectivity as follows \cite{kissinger2019cnot}. Let $G_Q$ be the connectivity graph of the qubits $Q$ that have not yet been disentangled, and $S = \{j \mid M_{jq} = 1\} \cup \{ q \}$ the set of qubits that are non-zero in the $q$th column of $M$, together with the pivot qubit $q \in Q$. Then we calculate a Steiner tree $T = (V_T, E_T)$ of $S$ in $G_Q$, and define an orientation of $T$ with all edges pointing away from $q$. Our goal is to eliminate every non-zero entry $M_{jq}$ for $j \in V_T \setminus \{q\}$ by performing a row operation between $j$ and its parent in $T$. Since $T \subseteq G$ this will only use physically allowed CNOT gates, and since $S \subseteq V_T$ then at the end of this process we will have $M_{jq} = 0$ for every $j \neq q$, exactly as desired. The elimination proceeds in two steps:
\begin{enumerate}
    \item First we must ensure that every row $j$ is able to be eliminated by its parent $p_j$. In this scenario, that means that we need $M_{p_jq} = 1$ (we will generalize this in Section \ref{sec:blockrowcol}). Therefore, while there exists any $j$ such that $M_{p_jq} = 0$ and $M_{jq} = 1$, we add row $j$ to row $p_j$. Since every leaf $l$ of $T$ must have $M_{lq} = 1$ (otherwise $T$ would not be minimal), this will result in $M_{jq} = 1$ for every $j \in V_T$.
    \item Then we eliminate every row $j \in V_T$ except the pivot $q$ by performing a row addition from $p_j$ to $j$. In order to avoid eliminating a parent before we can eliminate its children, we take $j$ in order of decreasing depth, from the leaves to the root.
\end{enumerate}
In Algorithm \ref{alg:elimcolsteiner} we implement \textsc{ElimCol} for arbitrary qubit connectivity using this method. See \cite{kissinger2019cnot} for detailed examples of the Steiner tree procedure. Combined with Algorithm \ref{alg:rowcolcnot}, this is a complete method to synthesize CNOT circuits. For an $n$-qubit parity matrix it produces a circuit of at most $2n(n-1)$ CNOT gates, regardless of the connectivity. Note that the ordering of qubits to disentangle must be chosen to ensure that $G_Q$ remains connected throughout the algorithm; this can be accomplished by calculating the articulation points of $G_Q$ at each step and picking the next $q$ to be any vertex that is not an articulation point \cite{wu2023optimization}. 

\begin{algorithm}
    \caption{The \textsc{ElimCol} subroutine for arbitrary qubit connectivity.}\label{alg:elimcolsteiner}
    \KwIn{The qubit connectivity graph $G$, the set of entangled qubits $Q$, a parity matrix $M$, its inverse transpose $M' = (M^{-1})^T$, and a column $q \in Q$. A flag $flip$ that determines whether CNOT gates should be flipped.}
    \KwOut{A CNOT circuit $C$ which eliminates the $q$th column of $M$.}

    $G_Q \gets$ the subgraph of $G$ induced by $Q$\;
    $S \gets \{ j \mid M_{jq} = 1\} \cup \{q\}$\;
    $T = (V_T, E_T) \gets$ a Steiner tree of $S$ in $G_Q$\;
    
    $C \gets$ an empty circuit\;
    \For{each edge $p_j \to j \in E_T$ in post-order from $q$}{
        \If{$M_{p_jq} = 0$}{
            Add row $j$ to row $p_j$ in $M$\;
            Add row $p_j$ to row $j$ in $M'$\;
            \lIf{$\lnot flip$}{prepend $CNOT_{j,p_j}$ to $C$ \lElse*{prepend $CNOT_{p_j,j}$ to $C$}}
        }
    }

    \For{each edge $p_j \to j \in E_T$ in post-order from $q$}{
        Add row $p_j$ to row $j$ in $M$\;
        Add row $j$ to row $p_j$ in $M'$\;
        \lIf{$\lnot flip$}{prepend $CNOT_{p_j,j}$ to $C$ \lElse*{prepend $CNOT_{j,p_j}$ to $C$}}
    }

    \Return{$C$}
\end{algorithm}

\subsection{Block Row Operations}

We will now describe the framework in which \textsc{RowCol} can be adapted to \textsc{BlockRowCol}. Suppose that we are given a parity matrix $M$ acting on $n$ qubits. We wish to synthesize a circuit that implements $M$ for a distributed architecture of $k$ equal-sized partitions of $n/k$ qubits with limited connectivity between partitions. Suppose that the rows and columns of $M$ are ordered so that partitions are contiguous. Then we can view $M$ as a block matrix
$$ M = \begin{pmatrix}
    M_{11} & \cdots & M_{1k} \\
    \vdots &        & \vdots \\
    M_{k1} & \cdots & M_{kk}
\end{pmatrix}$$
where each $\frac{n}{k} \times \frac{n}{k}$ block $M_{BA}$ represents the action of partition $a$ on partition $b$. Note that if we apply Gaussian elimination to $M$ to produce a CNOT circuit, this will require only local CNOT gates if and only if $M_{BA} = 0$ for every $a \neq b$. Therefore, our goal will be to eliminate all the off-diagonal blocks of $M$. 

First, we will discuss how to generalize the notion of row operations to this scenario. Consider the row operation which adds row $q_{aj}$ to row $q_{bi}$. This acts on the blocks of $M$ as follows:
$$ \forall c: ~~(M_{bc})_i \gets (M_{bc})_i + (M_{ac})_j $$
If multiple row operations act on the same pair of partitions $a$ and $b$, we may batch them together. Let the coefficient matrix $R$ be a binary matrix where each value $R_{ij} = 1$ indicates that row $j$ in partition $a$ is to be added to row $i$ in partition $b$. Then the action on the blocks of $M$ is
$$ \forall c, i, j: ~~(M_{bc})_i \gets (M_{bc})_i + R_{ij}\cdot (M_{ac})_j $$
but because all the row operations commute, we can group them together and instead write it as:
$$ \forall c: ~~M_{bc} \gets M_{bc} + RM_{ac}$$
This is precisely the block-level analog of a row operation, so we call it a \emph{block row addition}.

Next, we will see how generalized CNOT gates act on $M$. Suppose we have a generalized CNOT gate $U$ of the form
$$ U = C_{ab}\left(\prod_i Z_i^{z_i}, ~~\prod_i X_i^{x_i}\right) $$
where each $\vec{z}, \vec{x} \in \mathbb{F}_2^{n/k}$. We call these \emph{ZX-type generalized CNOT gates}. 
\begin{lemma}
    A ZX-type generalized CNOT gate $C_{ab}(\prod_i Z_i^{z_i}, ~~\prod_i X_i^{x_i})$ acts on $M$ as a block row addition from partition $a$ to partition $b$ where the coefficient matrix $R = \vec{x}\vec{z}^T$ has rank one.
\end{lemma}

\begin{proof}
    Using the fact that $C(A, B_1)C(A, B_2) = C(A, \overline{B_1B_2})$ whenever $B_1$ and $B_2$ commute and that $C(A, I) = I$, we can split this as a product of generalized CNOT gates that each have only one non-identity Pauli matrix in each partition:
    \begin{align*}
        U &= C_{ab}(\prod_j Z_j^{z_j}, \prod_i X_i^{x_i}) = \prod_j C_{ab}(Z_j^{z_j}, \prod_i X_i^{x_i}) \\
        &= \prod_{i,j} C_{ab}(Z_j^{z_j}, X_i^{x_i}) = \prod_{i,j} C_{ab}(Z_j, X_i)^{z_jx_i}
    \end{align*}
    Note that a generalized CNOT gate $C_{ab}(Z_j, X_i)$ is exactly a normal CNOT gate controlled on $q_{aj}$ targeting $q_{bi}$. Hence the action of $U$ on the blocks of $M$ will be that of a block row addition where $R_{ij} = z_jx_i$, which implies $R = \vec{x}\vec{z}^T$. 
\end{proof}

We can use this property to perform any block row addition with exactly $\mathrm{rank}(R) \leq \frac{n}{k}$ generalized CNOT gates. This is asymptotically more efficient than the naive decomposition into $(n/k)^2$ CNOT gates.

\begin{lemma}
    \label{thm:blockrowaddgates}
    Any block row addition from partition $a$ to partition $b$ with coefficient matrix $R$ can be realized using exactly $\mathrm{rank}(R)$ ZX-type generalized CNOT gates.
\end{lemma}

\begin{proof}
    Let $U$ be the unitary performing the block row addition. Any matrix binary $R$ can be written in the form:
    $$ R = \sum_{i = 1}^{\mathrm{rank}(R)} \vec{u_i}\vec{v_i}^T $$
    This is known as the \emph{rank factorization} of $R$, and it can be computed via Gaussian elimination \cite{piziak1999full}. Then the action of the row addition on the blocks of $M$ can be written as:
    $$ \forall c: ~~M_{bc} \gets M_{bc} \oplus \left(\sum_{i = 1}^{\mathrm{rank}(R)} \vec{u_i}\vec{v_i}^T\right)M_{ac} $$
    Or equivalently, we have
    $$ \forall c, i: ~~M_{bc} \gets M_{bc} \oplus \vec{u_i}\vec{v_i}^TM_{ac}$$
    since all of these row operations commute. We can thus rewrite it as a product of individual ZX-type generalized CNOT gates:
    \begin{equation*}
        U = \prod_{i} C_{ab}\left(\prod_j Z_j^{v_{ij}}, \prod_k X_k^{u_{ik}}\right) \qedhere
    \end{equation*}
\end{proof}

In \textsc{BlockRowCol}, as in \textsc{RowCol}, we will need to perform block row additions on both the parity matrix $M$ and its inverse transpose $M' = (M^{-1})^T$. Since from Lemma \ref{thm:rowopsinvtransp} we know that the effect of row additions on the inverse transpose matrix is again a row addition but with the source and target flipped, we can conclude the block row addition from $a$ to $b$ with $R$ in $M$ corresponds to a block row addition from $b$ to $a$ with coefficient matrix $R^T$ in $M'$.

\subsection{BlockRowCol}
\label{sec:blockrowcol}

We now have everything that we need to construct \textsc{BlockRowCol}. We will use the setup of the previous subsection, with the parity matrix $M$ divided into square $\frac{n}{k} \times \frac{n}{k}$ blocks $M_{BA}$. First we consider \textsc{BlockElimCol}, the blocked version of \textsc{ElimCol}. Suppose we want to eliminate the block $M_{bq}$ using the block $M_{aq}$, then we need to find a coefficient matrix $R$ such that $M_{bq} + R M_{aq} = 0$, and we would like to minimize the rank of $R$ so that its implementation uses as few generalized CNOT gates as possible.

\begin{definition}
    For any matrix $A$, let $\mathrm{rsp}(A)$ be the linear subspace spanned by the rows of $A$.
\end{definition}

\begin{lemma}
    \label{thm:blockpivotelim}
    If $\mathrm{rsp}(B) \subseteq \mathrm{rsp}(A)$ then there exists $R$ such that $B + RA = 0$ and $\mathrm{rank}(R) = \mathrm{rank}(B)$.
\end{lemma}

\begin{proof}
    Let $\{ a_j \}$ be the indices of the rows of $A$ that form a basis for $\mathrm{rsp}(A)$. By definition, for every row $B_i$ of $B$, we may write $B_i$ as a linear combination of the rows $A_{a_j}$, $ B_i = \sum_{j = 0}^{\mathrm{rank}(A)} R_{ia_j} A_{a_j} $.
    The matrix $R$ given by $R_{ij}$ then satisfies $B + RA = 0$. Since the $A_{a_j}$ are linearly independent, $R$ is unique and any linear dependence of $B$ must also be a linear dependence of $R$, which forces $\mathrm{rank}(R) \leq \mathrm{rank}(B)$. On the other hand since $\mathrm{rank}(B) = \mathrm{rank}(RA) \leq \min \{ \mathrm{rank}(R), \mathrm{rank}(A) \} \leq \mathrm{rank}(R)$ we have $\mathrm{rank}(R) = \mathrm{rank}(B)$.
\end{proof}

This is a generalization of the usual pivot condition for Gaussian elimination that a row can only be used to eliminate an entry of the matrix if its value in the column to be eliminated is non-zero: a $1 \times 1$ matrix is invertible, and hence its row space is full, if and only if its only entry is non-zero. To make a row into a pivot, we need to enforce the row space condition. In Gaussian elimination, this is done by using a row operation to make the pivot entry non-zero; similarly, in the blocked version, we can use a block row operation.

\begin{lemma}
    \label{thm:blockmakepivot}
    Suppose matrices $A$ and $B$ have the same number of columns, and $A$ has at least as many rows as $B$. Then there exists $R$ such that
    $$ \mathrm{rsp}(A + RB) \supseteq \mathrm{rsp}(B) \cup \mathrm{rsp}(A) $$
    with $\mathrm{rank}(R) \leq \mathrm{rank}(B)$.
\end{lemma}

\begin{proof}
    Let $M$ be the matrix obtained by stacking $A$ and $B$ vertically and then reducing to column echelon form. Let $\{ b_i \}$ be the list of rows of $B$ that are pivot rows in $M$. Let $|\{ b_i \}| = k$, then pick any $k$ distinct rows $\{ a_i \}$ of $A$ that are \emph{not} pivot rows in $M$. Form $R$ by setting $R_{a_ib_i} = 1$ for all $1 \leq i \leq k$, and zero otherwise. The block row operation from $B$ to $A$ with coefficient matrix $R$ acts on $M$ such that all the pivot rows are rows of $A$, which implies that $\mathrm{rsp}(A + RB) \supseteq \mathrm{rsp}(B)$. Since we did not remove any pivot rows from $A$, we also see that $\mathrm{rsp}(A + RB) \supseteq \mathrm{rsp}(A)$. We have $\mathrm{rank}(R) = k \leq \mathrm{rank}(B)$ since $R$ has $k$ non-zero rows, all of which are distinct and have only one non-zero element.
\end{proof}

These two operations are sufficient for \textsc{BlockElimCol}, which will be done using Steiner trees, as in Section \ref{sec:rowcolinterp}. Let $G_P$ be the connectivity graph of the partitions $P$ that have not yet been disentangled, and let $S = \{ j \mid M_{jp} \neq 0 \} \cup \{ p \}$ be the set of partitions which have a non-zero block in the column corresponding to the $p$th partition, along with the pivot partition $p \in P$. Calculate a Steiner tree $T = (V_T, E_T)$ of $S$ in $G_P$ and orient the edges pointing away from $p$. Then the elimination proceeds in two steps:
\begin{enumerate}
    \item While there is a partition $j \neq p$ with $p_j$ in $T$ such that $\mathrm{rsp}(M_{jp}) \nsubseteq \mathrm{rsp}(M_{p_j p})$, use Lemma \ref{thm:blockmakepivot} to find $R$ such that $\mathrm{rsp}(M_{jp}) \cup \mathrm{rsp}(M_{p_jp}) \subseteq \mathrm{rsp}(M_{p_jp} + RM_{jp})$. Apply a block row addition from $j$ to $p_j$ with $R$.
    \item Now we eliminate every block except the pivot. Proceeding from the leaves to the root, use Lemma \ref{thm:blockpivotelim} to find $R$ such that $M_{jp} + RM_{p_jp} = 0$, and apply a block row addition from $p_j$ to $j$ with $R$.
\end{enumerate}
This is formalized in Algorithm \ref{alg:blockelimcol}, relying on the subroutine \textsc{BRowOpCirc} which constructs a circuit of generalized CNOT gates implementing a block row operation according to Lemma \ref{thm:blockrowaddgates}. An example is given in Figure \ref{fig:blockelimcol}. 

\begin{algorithm}
    \caption{The \textsc{BlockElimCol} subroutine.}\label{alg:blockelimcol}
    \KwIn{The partition connectivity graph $G$, the set of entangled partitions $P$, a blocked parity matrix $M$, its inverse transpose $M' = (M^{-1})^T$, and a partition $p \in P$. A flag $flip$ that determines whether generalized CNOT gates should be flipped.}
    \KwOut{A generalized CNOT gate circuit $C$ eliminating the $p$th (block) column of $M$.}

    $G_P \gets$ the subgraph of $G$ induced by $P$\;
    $S \gets \{ j \mid M_{jp} \neq 0\} \cup \{p\}$\;
    $T = (V_T, E_T) \gets$ a Steiner tree of $S$ in $G_P$\;
    
    $C \gets$ an empty circuit\;
    \For{each edge $p_j \to j \in E_T$ in post-order from $p$}{
        \If{$\mathrm{rsp}(M_{jp}) \nsubseteq \mathrm{rsp}(M_{p_jp})$}{
            Find $R$ such that $\mathrm{rsp}(M_{jp}) \cup \mathrm{rsp}(M_{p_jp}) \subseteq \mathrm{rsp}(M_{p_jp} + RM_{jp})$ (Lemma \ref{thm:blockmakepivot})\;
            Block row addition from $j$ to $p_j$ with $R$ in $M$\;
            Block row addition from $p_j$ to $j$ with $R^T$ in $M'$\;
            $C_{jp_j} \gets$ \lIf{$\lnot flip$}{$\textsc{BRowOpCirc}(j, p_j, R)$ \lElse*{$\textsc{BRowOpCirc}(p_j, j, R^T)$}} 
            Prepend $C_{jp_j}$ to $C$\;
        }
    }

    \For{each edge $p_j \to j \in E_T$ in post-order from $p$}{
        Find $R$ such that $M_{jp} + RM_{p_jp} = 0$ (Lemma \ref{thm:blockpivotelim})\;
        Block row addition from $p_j$ to $j$ with $R$ in $M$\;
        Block row addition from $j$ to $p_j$ with $R^T$ in $M'$\;
        $C_{p_jj} \gets$ \lIf{$\lnot flip$}{$\textsc{BRowOpCirc}(p_j, j, R)$ \lElse*{$\textsc{BRowOpCirc}(j, p_j, R^T)$}}
        Prepend $C_{p_jj}$ to $C$\;
    }

    \Return{$C$}
\end{algorithm}

\begin{figure}
    \resizebox{\textwidth}{!}{$$\begin{tikzpicture}
	\begin{pgfonlayer}{nodelayer}
		\node [style=vertex] (0) at (0, 0) {};
		\node [style=vertex] (1) at (2, 0) {};
		\node [style=vertex] (2) at (4, 0) {};
		\node [style=vertex] (3) at (0, -2) {};
		\node [style=white dot] (4) at (2, -2) {};
		\node [style=vertex] (5) at (4, -2) {};
		\node [style=vertex] (6) at (4, -4) {};
		\node [style=vertex] (7) at (2, -4) {};
		\node [style=vertex] (8) at (0, -4) {};
		\node [style=none] (9) at (0, 0.375) {\tiny $A$};
		\node [style=none] (10) at (2, 0.375) {\tiny $B$};
		\node [style=none] (11) at (4, 0.375) {\tiny $C$};
		\node [style=none] (12) at (2.25, -3.75) {\tiny $E$};
		\node [style=none] (13) at (4.375, -2) {\tiny $D$};
		\node [style=none] (14) at (4.375, -4) {\tiny $0$};
		\node [style=none] (15) at (2.375, -1.75) {\tiny $0$};
		\node [style=none] (16) at (0.375, -3.75) {\tiny $0$};
		\node [style=none] (17) at (0.375, -1.75) {\tiny $0$};
		\node [style=vertex] (30) at (16, -3.75) {};
		\node [style=vertex] (31) at (17.75, -2) {};
		\node [style=vertex] (32) at (19.5, -3.75) {};
		\node [style=white dot] (33) at (15.75, 0) {};
		\node [style=vertex] (34) at (15.75, -2.5) {};
		\node [style=vertex] (35) at (13.75, -2) {};
		\node [style=none] (36) at (16, -4.125) {\tiny $A$};
		\node [style=none] (37) at (17.875, -1.625) {\tiny $G$};
		\node [style=none] (38) at (19.5, -4.125) {\tiny $C$};
		\node [style=none] (39) at (13.75, -2.375) {\tiny $E$};
		\node [style=none] (40) at (16.125, 0.25) {\tiny $F$};
		\node [style=none] (41) at (15.75, -2.875) {\tiny $D$};
		\node [style=none] (42) at (14.625, -0.75) {\rotatebox{45}{\scalebox{0.35}{$\mathrm{rsp}(E) \subseteq \mathrm{rsp}(F)$}}};
		\node [style=none] (43) at (16.875, -0.75) {\rotatebox{-45}{\scalebox{0.35}{$\mathrm{rsp}(F) \supseteq \mathrm{rsp}(G)$}}};
		\node [style=none] (44) at (15.95, -1.25) {\rotatebox{-90}{\scalebox{0.35}{$\mathrm{rsp}(F) \supseteq \mathrm{rsp}(D)$}}};
		\node [style=none] (45) at (18.75, -2.625) {\rotatebox{-45}{\scalebox{0.35}{$\mathrm{rsp}(G) \supseteq \mathrm{rsp}(C)$}}};
		\node [style=none] (46) at (16.75, -2.625) {\rotatebox{45}{\scalebox{0.35}{$\mathrm{rsp}(A) \subseteq \mathrm{rsp}(G)$}}};
		\node [style=vertex] (47) at (8.25, -3.75) {};
		\node [style=vertex] (48) at (10, -2) {};
		\node [style=vertex] (49) at (11.75, -3.75) {};
		\node [style=white dot] (50) at (8, 0) {};
		\node [style=vertex] (51) at (8, -2.5) {};
		\node [style=vertex] (52) at (6, -2) {};
		\node [style=none] (53) at (8.25, -4.125) {\tiny $A$};
		\node [style=none] (54) at (10.125, -1.625) {\tiny $B$};
		\node [style=none] (55) at (11.75, -4.125) {\tiny $C$};
		\node [style=none] (56) at (6, -2.375) {\tiny $E$};
		\node [style=none] (57) at (8.375, 0.25) {\tiny $0$};
		\node [style=none] (58) at (8, -2.875) {\tiny $D$};
		\node [style=vertex] (80) at (23.75, -3.75) {};
		\node [style=vertex] (81) at (25.5, -2) {};
		\node [style=vertex] (82) at (27.25, -3.75) {};
		\node [style=white dot] (83) at (23.5, 0) {};
		\node [style=vertex] (84) at (23.5, -2.5) {};
		\node [style=vertex] (85) at (21.5, -2) {};
		\node [style=none] (86) at (23.75, -4.125) {\tiny $0$};
		\node [style=none] (87) at (25.625, -1.625) {\tiny $0$};
		\node [style=none] (88) at (27.25, -4.125) {\tiny $0$};
		\node [style=none] (89) at (21.5, -2.375) {\tiny $0$};
		\node [style=none] (90) at (23.875, 0.25) {\tiny $F$};
		\node [style=none] (91) at (23.5, -2.875) {\tiny $0$};
		\node [style=none] (92) at (22, -0.75) {};
		\node [style=none] (93) at (22.5, -1) {};
		\node [style=none] (94) at (23.5, -1.5) {};
		\node [style=none] (95) at (24.5, -3) {};
		\node [style=none] (96) at (26.5, -3) {};
		\node [style=none] (97) at (24.5, -1) {};
		\node [style=none] (98) at (2, -5) {\tiny (a)};
		\node [style=none] (99) at (8, -5) {\tiny (b)};
		\node [style=none] (100) at (15.75, -5) {\tiny (c)};
		\node [style=none] (101) at (23.5, -5) {\tiny (d)};
	\end{pgfonlayer}
	\begin{pgfonlayer}{edgelayer}
		\draw (4) to (3);
		\draw (3) to (0);
		\draw (3) to (8);
		\draw (8) to (7);
		\draw (5) to (2);
		\draw (5) to (6);
		\draw (6) to (7);
		\draw (4) to (7);
		\draw (4) to (5);
		\draw (4) to (1);
		\draw (1) to (0);
		\draw (1) to (2);
		\draw [style=green edge] (33) to (35);
		\draw [style=green edge] (33) to (34);
		\draw [style=green edge] (33) to (31);
		\draw [style=green edge] (31) to (30);
		\draw [style=green edge] (31) to (32);
		\draw [style=green edge] (50) to (52);
		\draw [style=green edge] (50) to (51);
		\draw [style=green edge] (50) to (48);
		\draw [style=green edge] (48) to (47);
		\draw [style=green edge] (48) to (49);
		\draw [style=green edge] (83) to (85);
		\draw [style=green edge] (83) to (84);
		\draw [style=green edge] (83) to (81);
		\draw [style=green edge] (81) to (80);
		\draw [style=green edge] (81) to (82);
		\draw [style=dashed diredge, in=150, out=-30] (92.center) to (93.center);
		\draw [style=dashed diredge, in=135, out=-30] (93.center) to (94.center);
		\draw [style=dashed diredge, in=150, out=-45] (94.center) to (95.center);
		\draw [style=dashed diredge, in=-150, out=-30] (95.center) to (96.center);
		\draw [style=dashed diredge, in=30, out=45] (96.center) to (97.center);
	\end{pgfonlayer}
\end{tikzpicture}$$}
    \caption{An example of applying Algorithm \ref{alg:blockelimcol} to a set of nine partitions connected in a 3x3 grid. \emph{(a)} Suppose we have a parity matrix with non-zero blocks on five partitions, which we wish to eliminate using the pivot vertex (shown in white). \emph{(b)} First, we construct the Steiner tree connecting the non-zero blocks and the pivot, and orient the edges away from the pivot. \emph{(c)} By applying Lemma \ref{thm:blockmakepivot} repeatedly, we modify the blocks at $B$ and the pivot to respect the row space inclusion relations listed on each edge, resulting in the new block matrices $G$ and $F$. \emph{(d)} Finally we can eliminate all the non-pivot blocks by traversing the tree edges in post-order (shown with dashed arrows), and applying Lemma \ref{thm:blockpivotelim} to each edge.}\label{fig:blockelimcol}
\end{figure}
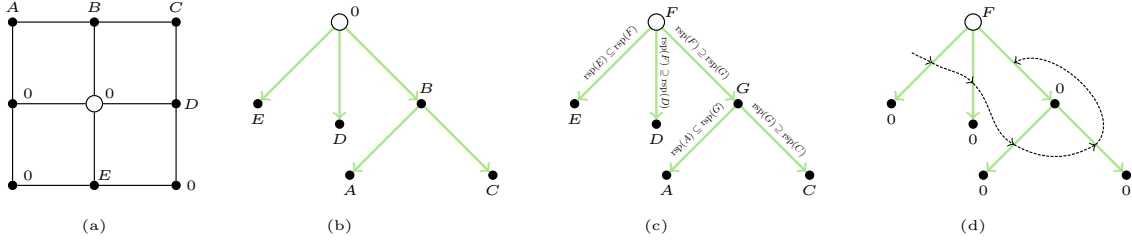

From this, \textsc{BlockRowCol} follows the same template as \textsc{RowCol}, in two phases. First, we disentangle each partition from the rest one by one using generalized CNOT gates. This is done by eliminating the corresponding block of columns from the parity matrix $M$, and then the corresponding block of columns in $(M^{-1})^T$, using \textsc{BlockElimCol}. This works in the same way as \textsc{RowCol} because Lemma \ref{thm:zeroblockinv} still applies for the case of block matrices. We pick the ordering of partitions to ensure that the connectivity graph stays connected at each step. After all of the partitions are disentangled, we can perform \textsc{RowCol} on the diagonal blocks of $M$; this will only use local CNOT gates that respect the qubit connectivity and reduces $M$ to the identity matrix. This procedure is formalized in Algorithm \ref{alg:blockrowcol}, and an example is given in Figure \ref{fig:blockrowcol}.

\begin{algorithm}
    \caption{The \textsc{BlockRowCol} algorithm for distributed CNOT circuit synthesis.}\label{alg:blockrowcol}
    \KwIn{The partition connectivity graph $G$, the qubit connectivity graphs of each partition $\{ G_p \}$, and the parity matrix $M$.}
    \KwOut{A representation $(\mathcal{D}(C_N), C_L)$ of a distributed circuit implementing $M$.}

    $C \gets$ an empty circuit\;
    $P \gets$ the set of all partitions\;
    \While{$P$ is non-empty}{
        $G_P \gets$ the induced subgraph of $P$ in $G$\; 
        Pick $p \in P$ which is not a articulation point of $G_P$\;
        $C_p \gets \textsc{BlockElimCol}(G, P, M, M', p, flip=\mathrm{false})$\;
        Prepend $C_p$ to $C$\;
        $C'_p \gets \textsc{BlockElimCol}(G, P, M', M, p, flip=\mathrm{true})$\;
        Prepend $C'_p$ to $C$\;
        Remove $p$ from $P$\;
    }

    \For{each partition $p$}{
        $L_p \gets \textsc{RowCol}(G_p, M_{pp})$\;
        Prepend $L_p$ to $C$\;
    }

    Find $(\mathcal{D}(C_N), C_L)$ using Algorithm \ref{alg:localcliffpush}\;
    \Return{$(\mathcal{D}(C_N), C_L)$}
\end{algorithm}

\begin{figure}
    \resizebox{\textwidth}{!}{$$\begin{tikzpicture}
	\begin{pgfonlayer}{nodelayer}
		\node [style=none] (0) at (0, 0) {$\left(\begin{array}{cc|cc} 1 & 1 & 0 & 1 \\ 1 & 1 & 1 & 0 \\ \hline 0 & 1 & 1 & 0 \\ 1 & 0 & 0 & 1 \end{array}\right)$};
		\node [style=none] (1) at (10, 0) {$\left(\begin{array}{cc|cc} 1 & 1 & 0 & 1 \\ 1 & 0 & 0 & 0 \\ \hline 0 & 1 & 1 & 0 \\ 1 & 0 & 0 & 1 \end{array}\right)$};
		\node [style=none] (2) at (5, 0) {$\longrightarrow$};
		\node [style=none] (3) at (5, 1) {\tiny $R_{BA} = \begin{pmatrix} 0 & 0 \\ 1 & 0\end{pmatrix}$};
		\node [style=none] (4) at (20, 0) {$\left(\begin{array}{cc|cc} 1 & 1 & 0 & 1 \\ 1 & 0 & 0 & 0 \\ \hline 0 & 0 & 1 & 1 \\ 0 & 0 & 0 & 1 \end{array}\right)$};
		\node [style=none] (5) at (15, 0) {$\longrightarrow$};
		\node [style=none] (6) at (15, 1) {\tiny $R_{AB} = \begin{pmatrix} 1 & 1 \\ 0 & 1\end{pmatrix}$};
		\node [style=none] (7) at (-4.5, 0) {$M$};
		\node [style=none] (8) at (0, -4) {$\left(\begin{array}{cc|cc} 0 & 1 & 0 & 0 \\ 1 & 1 & 0 & 0 \\ \hline 0 & 1 & 1 & 1 \\ 0 & 1 & 0 & 1 \end{array}\right)$};
		\node [style=none] (13) at (5, -4) {$\longrightarrow$};
		\node [style=none] (14) at (5, -3) {\tiny $R_{AB} = \begin{pmatrix} 1 & 0 \\ 1 & 0\end{pmatrix}$};
		\node [style=none] (15) at (-5.5, -4) {$(M^{-1})^T$};
		\node [style=none] (16) at (-3.5, 0) {$=$};
		\node [style=none] (17) at (-3.5, -4) {$=$};
		\node [style=none] (18) at (10, -4) {$\left(\begin{array}{cc|cc} 0 & 1 & 0 & 0 \\ 1 & 1 & 0 & 0 \\ \hline 0 & 0 & 1 & 1 \\ 0 & 0 & 0 & 1 \end{array}\right)$};
		\node [style=none] (20) at (-4.5, -8) {$M$};
		\node [style=none] (21) at (-3.5, -8) {$=$};
		\node [style=none] (22) at (0, -8) {$\left(\begin{array}{cc|cc} 1 & 1 & 0 & 0 \\ 1 & 0 & 0 & 0 \\ \hline 0 & 0 & 1 & 1 \\ 0 & 0 & 0 & 1 \end{array}\right)$};
		\node [style=none] (23) at (10, -8) {$\left(\begin{array}{cc|cc} 1 & 0 & 0 & 0 \\ 0 & 1 & 0 & 0 \\ \hline 0 & 0 & 1 & 1 \\ 0 & 0 & 0 & 1 \end{array}\right)$};
		\node [style=none] (24) at (5, -8) {$\longrightarrow$};
		\node [style=none] (25) at (5, -7.5) {\tiny $\mathrm{CNOT}_{21}$};
		\node [style=none] (26) at (5, -7) {\tiny $\mathrm{CNOT}_{12}$};
		\node [style=none] (27) at (20, -8) {$\left(\begin{array}{cc|cc} 1 & 0 & 0 & 0 \\ 0 & 1 & 0 & 0 \\ \hline 0 & 0 & 1 & 0 \\ 0 & 0 & 0 & 1 \end{array}\right)$};
		\node [style=none] (28) at (15, -8) {$\longrightarrow$};
		\node [style=none] (29) at (15, -7.5) {\tiny $\mathrm{CNOT}_{43}$};
		\node [style=none] (30) at (-8, 0) {$(a)$};
		\node [style=none] (31) at (-8, -4) {$(b)$};
		\node [style=none] (32) at (-8, -8) {$(c)$};
	\end{pgfonlayer}
\end{tikzpicture}$$}
    \caption{An example of applying Algorithm \ref{alg:blockrowcol} to a $4 \times 4$ parity matrix $M$ with two partitions. \emph{(a)} First, we eliminate the off-diagonal block in the first column of $M$ using two block row additions; one to ensure that $\mathrm{rsp}(M_{21}) \subseteq \mathrm{rsp}(M_{11})$ and the second to eliminate $M_{21}$. \emph{(b)} Now we eliminate the off-diagonal block in the first column of $(M^{-1})^T$ with one block row addition. From Lemma \ref{thm:zeroblockinv}, we already know that the top-right block must be zero. \emph{(c)} As the two partitions are now unentangled, we can apply Algorithm \ref{alg:rowcolcnot} to them individually.}\label{fig:blockrowcol}
\end{figure}

Finally, let us consider the performance of \textsc{BlockRowCol}. Since it acts as \textsc{RowCol} on the blocks of the parity matrix, it will perform at most $2k(k-1)$ block row addition operations. From Lemma \ref{thm:blockrowaddgates}, we know that each of these uses at most $\frac{n}{k}$ non-local gates; therefore:
\begin{theorem}
    \textsc{BlockRowCol} synthesizes a CNOT circuit on $n$ qubits for a distributed architecture of $k$ equal partitions using at most $2n(k - 1)$ non-local gates, regardless of the connectivity between the partitions.
\end{theorem}

\subsection{Optimality}
\label{sec:blockrowcol-optimality}

In this section, we will show that \textsc{BlockRowCol} is optimal in a certain sense. We begin by giving a lower bound on the number of generalized CNOT gates required to synthesize any distributed CNOT circuit.

\begin{lemma}
    For $k \geq 2$ and $n \geq 4$, there exist parity matrices that cannot be reduced to block-diagonal form using fewer than $\frac{1}{32}\min\{nk, \frac{n^2}{\log_2 n}\}$ generalized CNOT gates.
\end{lemma}

\begin{proof}
    We use a counting argument similar to \cite{patel2003efficient}. Let $\#I_m$ be the number of invertible $m \times m$ binary matrices, then we have:
    $$ 2^{m(m-1)} \leq \#I_m = \prod_{i=0}^{m-1} (2^m - 2^i) \leq 2^{m^2} $$
    Now let $N$ be the number of distinct unitaries comprised of generalized CNOT gates on $n$ qubits divided in $k$ equal size partitions. Since any parity matrix on $n$ qubits can be realized as a circuit of generalized CNOT gates followed by a local CNOT circuit on each partition, we must have:
    $$ \#I_n \leq N \cdot \#I_{n/k}^k \implies N \geq \frac{\#I_n}{\#I_{n/k}^k} \geq \frac{2^{n(n-1)}}{(2^{(n/k)^2})^k} = 2^{n^2(1 - 1/k) - n} \geq 2^{\frac{1}{2}n^2 - n}$$
    The number of possible generalized CNOT gates is at most $k(k - 1)\cdot 4^{n/k} \cdot 4^{n/k}$ (the $k(k-1)$ factor represents the choice of partitions, and each $4^{n/k}$ is a choice of Pauli string). Now suppose that $d$ is the minimum number of generalized CNOT gates needed to reduce any parity matrix to block-diagonal form, then we will have:
    $$ (k(k-1)16^{n/k})^d \geq N \geq 2^{\frac{1}{2}n^2 - n} $$
    Then by taking logarithms and applying the mediant inequality, we see that:
    \begin{align*}
        d &\geq \frac{\frac{1}{2}n^2 - n}{4\frac{n}{k} + \log_2(k(k-1))} \geq \frac{1}{2}\min \left\{ \frac{\frac{1}{2}n^2 - n}{4\frac{n}{k}}, \frac{\frac{1}{2}n^2 - n}{\log_2(k(k-1))}\right\} \\
        &\geq \min\left\{\frac{1}{16}(n - 2)k, \frac{n(n-2)}{8\log_2 k}\right\} \geq \min\left\{\frac{1}{16}(n - 2)k, \frac{n(n-2)}{8\log_2 n}\right\} \geq \frac{1}{32}\min\left\{nk, \frac{n^2}{\log_2 n}\right\}\qedhere
    \end{align*}
\end{proof}

This leads immediately to the asymptotic optimality of \textsc{BlockRowCol} - that is, it performs within a constant factor of optimal for the worst-case inputs.

\begin{theorem}
    \textsc{BlockRowCol} uses an asymptotically optimal number of non-local gates whenever $k = O(n / \log n)$. The strategy that runs both \textsc{BlockRowCol} and Patel-Markov-Hayes' algorithm \cite{patel2003efficient} and chooses the result with the fewest non-local gates is asymptotically optimal for any value of $k$.
\end{theorem}

However, for the case of $k = 2$ we can actually do better, and prove that the output of \textsc{BlockRowCol} is not just asymptotically, but \emph{approximately} optimal amongst circuits that use only ZX-type generalized CNOT gates; that is, within a constant factor of optimal for \emph{every} input. Note however that all (non-distributed) CNOT circuits become ZX-type generalized CNOT gate circuits once distributed, so this is equivalent to saying that it is optimal amongst all CNOT circuits. It was shown in \cite{bravyi2021hadamard} that generally speaking, Clifford gates do not provide a benefit to CNOT gates when implementing any unitary that \emph{could} be implemented with CNOT gates. There we would not expect non-ZX-type generalized CNOT gates to provide a benefit in this case, either.

\begin{theorem}
    \label{thm:blockrowcolapproxopt}
    For $k = 2$ and any input parity matrix  $M$, (a variation of) \textsc{BlockRowCol} uses at most twice as many non-local gates as any circuit implementing $M$ with ZX-type generalized CNOT gates.
\end{theorem}

\begin{proof}
    Let the two partitions be called $A$ and $B$, and divide $M$ into the four $\frac{n}{2} \times \frac{n}{2}$ blocks $M_{AA}$, $M_{BB}$, $M_{BA}$, $M_{AB}$. Then \textsc{BlockRowCol} uses generalized CNOT gates in three steps:
    \begin{enumerate}
        \item First, it is ensured that $\mathrm{rsp}(M_{BA}) \subseteq \mathrm{rsp}(M_{AA})$. Since $\mathrm{dim}(\mathrm{rsp}(M_{BA}) \cup \mathrm{rsp}(M_{AA})) = \frac{n}{2}$ this will result in $M_{AA}$ being invertible. This requires at most $\mathrm{rank}(M_{BA})$ ZX-type generalized CNOT gates acting from $B$ to $A$. Since each of these acts like a rank-one update to the $A$ partition of $M$, the rank of $M_{AB}$ may increase by at most $\mathrm{rank}(M_{BA})$.
        \item Then $M_{BA}$ is eliminated using $M_{AA}$. This requires exactly $\mathrm{rank}(M_{BA})$ gates acting from $A$ to $B$ and does not affect the rank of $M_{AB}$. After this step $M_{BB}$ must be invertible, as otherwise the $(M_{BA}, M_{BB})$ submatrix of $M$ would not have full rank, which is impossible because $M$ is invertible.
        \item Finally, $M'_{AB}$ is eliminated in $M'$ using $M'_{AA}$. Since $M_{AA}$ and $M_{BB}$ are both invertible, Lemma \ref{thm:zeroblockinv} guarantees that $M'_{AB}$ has the same rank as $M_{AB}$. Hence this requires at most $\mathrm{rank}(M_{BA}) + \mathrm{rank}(M_{AB})$ gates, because of the potential increase in rank in the first step.
    \end{enumerate}
    In total, this requires at most $3\cdot\mathrm{rank}(M_{BA}) + \mathrm{rank}(M_{AB})$ gates. Now consider repeating this procedure with the $A$ and $B$ partitions swapped; this will require $3\cdot\mathrm{rank}(M_{AB}) + \mathrm{rank}(M_{BA})$ gates. By picking the smallest of the two circuits, the number of gates is bounded above by:
    \begin{align*}
        &3\cdot \min \{ \mathrm{rank}(M_{BA}), \mathrm{rank}(M_{AB}) \} + \max \{ \mathrm{rank}(M_{BA}), \mathrm{rank}(M_{AB}) \} \\
        &\leq  2\cdot \min \{ \mathrm{rank}(M_{BA}), \mathrm{rank}(M_{AB}) \} + 2\cdot \max \{ \mathrm{rank}(M_{BA}), \mathrm{rank}(M_{AB}) \} \\
        &= 2\cdot (\mathrm{rank}(M_{BA}) + \mathrm{rank}(M_{AB}))
    \end{align*}
    On the other hand, since each ZX-type generalized CNOT gate acts as a rank-one update on one partition of $M$, at least $\mathrm{rank}(M_{BA}) + \mathrm{rank}(M_{AB})$ such gates are required to implement $M$, as the ranks of $M_{BA}$ and $M_{AB}$ must be reduced to zero.
\end{proof}

\section{Clifford Circuit Synthesis}
\label{sec:clifford}

Now we build on the previous section and introduce \textsc{DistRowCol}, an algorithm to synthesize distributed Clifford circuits using generalized CNOT gates. We will use the same setup as for CNOT circuits; that is, $k$ equal partitions totaling $n$ qubits (although in principle this method can also be applied to unequal partitions). We will assume familiarity with the stabilizer tableau representation of Clifford circuits; see \cite{aaronson2004improved} for an introduction. As before, there are three key parts to our method:
\begin{enumerate}
    \item The stabilizer tableau corresponding to a distributed Clifford circuits can be organized into a block form, where each block represents the interaction between two partitions.
    \item Once the tableau has been reduced to a block-diagonal form, it can be reduced to the identity using only local Clifford gates.
    \item Generalized CNOT gates will act on the tableau by multiplying into the blocks.
\end{enumerate}
Therefore, we will describe how to reduce a stabilizer tableau into the block-diagonal form while attempting to minimize the number of generalized CNOT gates by modifying the method of \cite{vandenberg2020simple}, and its architecture-aware variation in \cite{winderl2023architecture}. The resulting is very similar to \textsc{BlockRowCol} and has the same complexity (and hence is asymptotically efficient), but cannot take advantage of low-rank structures in the same way, so it is not approximately optimal for $k = 2$ like we showed for \textsc{BlockRowCol} in Theorem \ref{thm:blockrowcolapproxopt}.

\subsection{Block Structure and Generalized CNOT Gates}
\label{sec:cliffblock}

We begin by describing the block representation for a Clifford circuit $C$; note that this is \emph{not} the same as the block structure for CNOT circuits, in particular we will have $k \times n$ blocks instead of $k \times k$ blocks. Let $S = [S_1, S_2, \dots, S_n]$ be $n$-qubit Pauli strings that are the stabilizers of $C$, and $D = [D_1, D_2, \dots, D_n]$ be $n$-qubit Pauli strings that are the destabilizers of $C$. Then the definition of a stabilizer tableau implies that
$$ \forall i \neq j: \quad S_iS_j = S_jS_i, \quad D_iD_j = D_jD_i, \quad S_iD_j = D_jS_i, \quad S_iD_i \neq D_iS_i  $$
We organize the blocks as follows. Supposing that the qubits are laid out so that the partitions are contiguous, we relabel the stabilizers and destabilizers of $C$ so that $S_{ia}$ and $D_{ia}$ are stabilizer and destabilizer corresponding to the $i$th qubit of partition $a$. Furthermore, we divide each of these $n$-qubit Pauli strings into substrings of size $\frac{k}{n}$, so that $S_{iab}$ is a $\frac{n}{k}$-qubit Pauli string representing the action of the $i$th qubit in partition $a$ on partition $b$:
$$ S_{ia} = \bigotimes_b S_{iab} \qquad D_{ia} = \bigotimes_b D_{iab} $$
We will use the term $(a, b)$-block to mean the sets of Pauli strings $\{ S_{iab} \mid \forall i\}$ and $\{ D_{iab} \mid \forall i\}$. We refer to a block as \emph{trivial} if $S_{iab} = D_{iab} = I$ for all $i$, and non-trivial otherwise. The goal of our algorithm will be to make the $(a, b)$-block trivial for every pair of partitions $a$ and $b$ with $a \neq b$; this is what we called \emph{the block-diagonal form} above.

We can describe the action of a generalized CNOT gate $C_{ab}(A, B)$ in terms of the block representation. In Section \ref{sec:representation}, we saw that
$$ C(A, B) \cdot [P \otimes Q] \cdot C(A, B) = \begin{cases}
    P \otimes Q & \text{if } AP = PA \land BQ = QB \\
    P \otimes \overline{QB} & \text{if } AP \neq PA \land BQ = QB \\
    \overline{PA} \otimes Q & \text{if } AP = PA \land BQ \neq QB \\
    \overline{PA} \otimes \overline{QB} & \text{if } AP \neq PA \land BQ \neq QB
\end{cases} $$
and therefore in terms of the block structure, applying $C_{ab}(A, B)$ acts as
$$ \forall p, i: \qquad S_{ipa} \gets \begin{cases}
    S_{ipa} & \text{if } BS_{ipb} = S_{ipb}B \\
    \overline{S_{ipa}A} & \text{if } BS_{ipb} \neq S_{ipb}B 
\end{cases} \qquad S_{ipb} \gets \begin{cases}
    S_{ipb} & \text{if } AS_{ipa} = S_{ipa}A \\
    \overline{S_{ipb}B} & \text{if } AS_{ipa} \neq S_{ipa}A 
\end{cases}$$ 
and likewise for $D_{ipa}$ and $D_{ipb}$. Note that this operation may create a sign on any of the updated substrings, for example $S_{ipa}$, and so, unlike a regular stabilizer tableau, multiple signs are tracked per stabilizer and destabilizer, one for each partition. The signs of each substring are multiplied together to give the sign of the overall string. 

\subsection{Disentangling Partitions}
\label{sec:disentangle}

In \cite{winderl2023architecture}, as in \textsc{RowCol} explained in Section \ref{sec:rowcolinterp}, a Clifford circuit is synthesized by disentangling each of the qubits one at a time. We proceed similarly, but by disentangling partitions. The above update rule leads to an important observation: suppose that we wish to set some substring $S_{ipb}$ to the identity, then we can pick any Pauli string $A$ which anticommutes with $S_{ipa}$, and perform the generalized CNOT gate $C_{ab}(A, S_{ipb})$. The main difficulty then becomes picking $A$ so that this operation maintains the structure that we will build in the stabilizer tableau. To do this, we will make use of the following result:

\begin{lemma}
    \label{thm:blockfindanticommute}
    Suppose we have two sets of $n$-qubit Pauli strings $U = \{ U_1, \dots, U_k \}$ and $V = \{ V_1, \dots, V_k \}$ with $k < n$ such that $U_i U_j = U_j U_i$, $V_i V_j = V_j V_i$, and $U_i V_j = V_j U_i$ for all $i \neq j$ and $U_i V_i \neq V_i U_i$ for all $i$. Then for any Pauli string $P \neq \pm I$ such that $P U_i = U_i P$ and $P V_i = V_i P$ for all $i$, there exists a Pauli string $Q$ also with $Q U_i = U_i Q$ and $Q V_i = V_i Q$ for all $i$, but $PQ \neq QP$. 
\end{lemma}

\begin{proof}
    We proceed by induction on $k$ and $n$. If $k = 0$ and $n > 0$, we can pick any $Q$ such that $PQ \neq QP$, which exists since $P \neq \pm I$. Otherwise if $k > 1$ and $n > 1$, let $C$ be a Clifford unitary such that $C U_1 C^\dagger = Z_1$ and $C V_1 C^\dagger = X_1$, which always exists since they anticommute \cite{vandenberg2020simple}. Now let $U' = \{ C U_i C^\dagger \mid i > 1 \}$, $V' = \{ C V_i C^\dagger \mid i > 1 \}$, and $P' = C P C^\dagger$. We must have each of $U'_i$, $V'_i$, and $P'$ as identity on the first qubit, since they commute with both $C U_1 C^\dagger = Z_1$ and $C U_1 C^\dagger = X_1$, so let $U'_i = I \otimes U''_i$, $V'_i = I \otimes V''_i$ and $P' = I \otimes P''$. Then $U''$ and $V''$ must obey the same commutation conditions as $U$ and $V$, and $P''$ must commute with $U''$ and $V''$. Therefore, by the induction hypothesis, there exists a $Q''$ such that $P''Q'' \neq Q''P''$ and $Q''$ commutes with the sets of $(n-1)$-qubit Pauli strings $U''$ and $V''$. Now $Q = C^\dagger (I \otimes Q'') C$ must commute with $U$ and $V$, and satisfies $QP \neq PQ$ as required. Since we assumed $k < n$, the induction never hits the case of $n = 0$, and so this procedure is well-defined.
\end{proof}

We will use this to eliminate a whole column of blocks. That is, for a partition $a$, we will make the $(a, b)$-block trivial for every $b \neq a$. This is done one qubit at a time: suppose that for some $i$, we have $S_{jab} = D_{jab} = I$ for all $b \neq a$ and $j < i$. Then we can set $S_{iab} = D_{iab} = I$ for all $b \neq a$ in two steps:
\begin{enumerate}
    \item Because of the commutation relations of the stabilizer tableau and the fact that $S_{ja}$ and $D_{ja}$ are identity everywhere except partition $a$, we must have that $S_{jaa}S_{kaa} = S_{kaa}S_{jaa}$, $D_{jaa}D_{kaa} = D_{kaa}D_{jaa}$, $S_{jaa}D_{kaa} = D_{kaa}S_{jaa}$, and $S_{jaa}D_{jaa} \neq D_{jaa}S_{jaa}$ for all $j, k < i$. Furthermore, $S_{iaa}$ must commute with $S_{jaa}$ and $D_{jaa}$ for all $j < i$. Therefore, we can apply Lemma \ref{thm:blockfindanticommute} to find $A$ which commutes with $S_{jaa}$ and $D_{jaa}$ for all $j < i$, but anticommutes with $S_{iaa}$. Now apply $C_{ab}(A, S_{iab})$ for all $b \neq a$, then this sets $S_{iab} = I$ for all $b \neq a$ while keeping $S_{jab} = D_{jab} = I$ for all $j < i$.
    \item Now since $S_{iab} = I$ for all $b \neq a$, we must have $S_{iaa}D_{iaa} \neq D_{iaa}S_{iaa}$. Additionally, $S_{iaa}$ still commutes with $S_{jaa}$ and $D_{jaa}$ for all $j < i$. Therefore, we apply $C_{ab}(S_{iaa}, D_{iab})$ for all $b\neq a$, and this sets $D_{iab} = I$ for all $b \neq a$ while keeping $S_{iab} = I$ and $S_{jab} = D_{jab} = I$ for all $j < i$.
\end{enumerate}
Therefore, applying this method for each $1 \leq i \leq \frac{n}{k}$ in order will eliminate a column of blocks as desired. However, similar to Lemma \ref{thm:zeroblockinv}, we can also show that it eliminates the corresponding row of blocks: 

\begin{lemma}
    Given a stabilizer tableau in block form and a partition $a$, if the $(a, b)$-block is trivial for every $b \neq a$, then the $(b, a)$-block is also trivial for every $b \neq a$.
\end{lemma}

\begin{proof}
    By assumption, we have $S_{iab} = D_{iab} = I$ for every $b \neq a$ and all $i$. Therefore, the commutativity properties of the sets $\{S_{ia} \mid \forall i\}$ and $\{ D_{ia} \mid \forall i \}$ guaranteed by the stabilizer tableau must be encoded by $\{S_{iaa} \mid \forall i\}$ and $\{ D_{iaa} \mid \forall i \}$. These form the stabilizer tableau of some local Clifford unitary $C$ on the $a$ partition which maps $C S_{iaa} C^\dagger = Z_{q_{ai}}$ and $C D_{iaa} C^\dagger = X_{q_{ai}}$ for all $i$. For any $b \neq a$, the stabilizer tableau also implies that $S_{jb}$ and $D_{jb}$ commute with $S_{ia}$ and $D_{ia}$ for every $j$, and this implies that $S_{jb}$ and $D_{jb}$ must be identity on the qubit $q_{ia}$ for all $i$, which is exactly equivalent to $S_{jba} = D_{jba} = I$ for all $j$, and hence the $(b, a)$-block is trivial.
\end{proof}

Therefore, this procedure acts to completely disentangle one partition from the rest of the partitions. It is formalized in in Algorithm \ref{alg:disentanglealltoall} as the subroutine \textsc{Disentangle}. In practice, the step depending on Lemma \ref{thm:blockfindanticommute} can be done without constructing the basis changes explicitly, by phrasing the commutativity constraints as a linear algebra problem that is solved with Gaussian elimination. Similar to \textsc{BlockRowCol}, it is easy to see that doing this for each partition in order produces a stabilizer tableau in block-diagonal form. We will define this formally in the next section, after discussing how to adapt Algorithm \ref{alg:disentanglealltoall} for arbitrary network topologies.

\begin{algorithm}
    \caption{The \textsc{Disentangle} subroutine for all-to-all connected architectures.}\label{alg:disentanglealltoall}
    \KwIn{A stabilizer tableau $(S, D)$ in block form, and a partition $a$.}
    \KwOut{A generalized CNOT gate circuit which disentangles partition $a$ from the rest.}

    $C \gets$ an empty circuit\;
    \For{$i = 1$ \KwTo $\frac{n}{k}$}{
        $A \gets $ Apply Lemma \ref{thm:blockfindanticommute} with $U = \{ S_{jaa} \mid j < i\}$, $V = \{ D_{jaa} \mid j < i \}$, and $P = S_{iaa}$\;
        \For{each partition $b \neq a$}{
            Apply $C_{ab}(A, S_{iab})$ to $S$ and $D$\;
            Prepend $C_{ab}(A, S_{iab})$ to $C$\;
        }

        \For{each partition $b \neq a$}{
            Apply $C_{ab}(S_{iaa}, D_{iab})$ to $S$ and $D$\;
            Prepend $C_{ab}(S_{iaa}, D_{iab})$ to $C$\;
        }
    }
    
    \Return{$C$}
\end{algorithm}

\subsection{DistRowCol}
\label{sec:distrowcol}

Adapting Algorithm \ref{alg:disentanglealltoall} for arbitrary connectivities follows the same strategy as in Algorithm \ref{alg:blockelimcol} for \textsc{BlockElimCol}; we will find a Steiner tree, make sure that every node in the tree can eliminate its children, and then perform the elimination. While for \textsc{BlockElimCol} the condition that one partition could eliminate another was phrased in terms of the row-spaces of blocks of the parity matrix, it is considerably simpler in this case. In particular, if we wish to eliminate $P$ using $Q$, we can apply $C(A, Q)$ for any $A$ that anticommutes with $P$. We can always find such an $A$ whenever $P \neq \pm I$, and so this is the condition that needs to be enforced. However, care must be taken that performing this operation does not alter the structure we will build in the tableau. For this, we need the following result:

\begin{lemma}
    \label{thm:blockfindcommute}
    For any set of $n$-qubit Pauli strings $U = \{ U_1, U_2, \cdots, U_k \}$ with $k < 2n$, there exists a $n$-qubit Pauli string $V \neq \pm I$ such that $U_iV = VU_i$ for all $1 \leq i \leq k$.
\end{lemma}

\begin{proof}
    If $k = 0$ or $U_i = I$ for all $i$, pick any $V \neq \pm I$. If $U$ is mutually commuting, then pick $V$ to be any non-identity element of $U$. Otherwise we proceed by induction on $k$ and $n$. Suppose without loss of generality that $U_1$ and $U_2$ anticommute, then there is a Clifford unitary $C$ that maps $C U_1 C^\dagger = Z_1$ and $C U_2 C^\dagger = X_1$ \cite{vandenberg2020simple}. For every $i$, write $C U_i C^\dagger = P_i \otimes U'_i$ for some single-qubit Pauli $P_i$. Now by induction there is an $(n - 1)$-qubit Pauli string $V'$ which commutes with $U'_i$ for all $i > 2$. Since $U'_1 = U'_2 = I$, $V'$ also commutes with these, and hence $V = C^\dagger (I \otimes V') C$ commutes with all $U_i$. Because $k < 2n$ and we decrease $k$ by two and $n$ by one at each induction step, we never hit the case of $n = 0$ and the induction is well-defined.
\end{proof}

Let $G$ be the connectivity graph of the partitions. For each qubit $i$ within a partition $a$, suppose that $S_{jab} = D_{jab} = I$ for all $j < i$ and $b \neq a$, then we will eliminate $S_{iab}$ for all $b \neq a$ in three steps:
\begin{enumerate}
    \item First, let $S = \{ b \mid S_{iab} \neq \pm I \} \cup \{ a \}$ and find a Steiner tree $T = (V, E)$ for $S$ in $G$, and orient $E$ pointing away from $a$.
    \item For every $b \in V $, let $p_b \in V$ be its parent in $T$. Find $b$ such that $S_{iab} \neq \pm I$ and $S_{iap_b} = I$ and pick any $A$ which anticommutes with $S_{iab}$. If $p_b = a$ then find $B \neq \pm I$ which commutes with $S_{jaa}$ and $D_{jaa}$ for all $j < i$ using Lemma \ref{thm:blockfindcommute}, otherwise if $p_b \neq a$, then pick any $B \neq \pm I$. Apply $C_{bp_b}(A, B)$, then this sets $S_{iap_b} = B$. Repeat this until $S_{iap_b} \neq \pm I$ for every $b$.
    \item Now traverse the edges of the tree $p_b \to b$ from the leaves to the root. If $p_b = a$, then pick $A$ which anticommutes with $S_{iaa}$ and commutes with $S_{jaa}$ and $D_{jaa}$ for all $j < i$ using Lemma \ref{thm:blockfindanticommute}. If $p_b \neq a$, then pick any $A$ which anticommutes with $S_{iap_b}$. Apply $C_{p_bb}(A, S_{iab})$, then this sets $S_{iab} = I$.
\end{enumerate}
This clearly results in $S_{iab} = I$ for all $b \neq a$, but we must also justify why we maintain $S_{jab} = D_{jab} = I$ for all $j < i$ and $b \neq a$: since we started with $S_{jab} = D_{jab} = I$ for $b \neq a$, this can only change if we apply generalized CNOT gates which target the $a$ partition. In both the second and third steps, we always ensure that such gates act on $a$ with a Pauli string that commutes with $S_{jaa}$ and $D_{jaa}$ for all $j < i$, and hence $S_{jab}$ and $D_{jab}$ are unchanged. The procedure for eliminating $D_{iab}$ is essentially the same, except we can avoid the use of Lemmas \ref{thm:blockfindcommute} and \ref{thm:blockfindanticommute}: in the second step, we know that $D_{iaa} \neq \pm I$ since it must anticommute with $S_{iaa}$, and so the case $p_b = a$ can never occur. In the third step, in the case of $p_b = a$ we can always take $A = S_{iaa}$, since $S_{iaa}$ commutes with $S_{jaa}$ and $D_{jaa}$ for all $j < i$ and anticommutes with $D_{iaa}$ by the properties of the stabilizer tableau. This is formalized in Algorithm \ref{alg:disentanglearbitrary}. 

From \textsc{Disentangle}, we can construct \textsc{DistRowCol} in the same way that \textsc{BlockRowCol} is constructed from \textsc{BlockElimCol}. This is formalized in Algorithm \ref{alg:distrowcol}. The maximum number of non-local gates is the same as \textsc{BlockRowCol}, and we have:
\begin{theorem}
    \textsc{DistRowCol} synthesizes a Clifford circuit on $n$ qubits for a distributed architecture of $k$ equal partitions using at most $2n(k - 1)$ non-local gates, regardless of the connectivity between partitions. 
    
    \textsc{DistRowCol} is asymptotically optimal whenever $k = O(n / \log n)$, and the strategy that also runs Aaronson-Gottesman's algorithm \cite{aaronson2004improved} and chooses the result with the fewest non-local gates is asymptotically optimal for any value of $k$.
\end{theorem}

\begin{algorithm}
    \caption{The \textsc{Disentangle} subroutine for arbitrary connectivity.}\label{alg:disentanglearbitrary}
    \KwIn{The partition connectivity graph $G$, the set of entangled partitions $P$, a stabilizer tableau $(S, D)$ in block form, and a partition $a \in P$.}
    \KwOut{A generalized CNOT gate circuit which disentangles partition $a$ from the rest.}

    $C \gets$ an empty circuit\;
    $G_P \gets$ the subgraph of $G$ induced by $P$\;
    \For{$i = 1$ \KwTo $\frac{n}{k}$}{
        \tcc{Eliminate stabilizers:}
        $Q \gets \{ b \in P \mid S_{iab} \neq \pm I \} \cup \{ a \}$\;
        $T = (V_T, E_T) \gets$ the Steiner tree of $Q$ in $G_P$ rooted at $a$\;

        $B' \gets $ apply Lemma \ref{thm:blockfindcommute} with $U = \{ S_{jaa} \mid j < i\} \cup \{ D_{jaa} \mid j < i \}$\;
        \For{each edge $p_b \to b \in E_T$ in post-order from $a$}{
            \lIf{$S_{iap_b} \neq \pm I$}{continue}
            $A \gets$ any Pauli string which anticommutes with $S_{iab}$\;
            $B \gets$ \lIf{$p_b = a$}{$B'$ \lElse*{any $B \neq \pm I$ (e.g $Z_1$)}}
            Apply $C_{bp_b}(A, B)$ to $S$ and $D$\;
            Prepend $C_{bp_b}(A, B)$ to $C$\;
        }

        $A' \gets $ apply Lemma \ref{thm:blockfindanticommute} with $U = \{ S_{jaa} \mid j < i\}$, $V = \{ D_{jaa} \mid j < i \}$, and $P = S_{iaa}$\;
        \For{each edge $p_b \to b \in E_T$ in post-order from $a$}{
            $A \gets $ \lIf{$p_b = a$}{$A'$ \lElse*{any $A$ which anticommutes with $S_{iap_b}$}}
            Apply $C_{p_bb}(A, S_{iab})$ to $S$ and $D$\;
            Prepend $C_{p_bb}(A, S_{iab})$ to $C$\;
        }

        \tcc{Eliminate destabilizers:}
        $Q \gets \{ b \in P \mid D_{iab} \neq \pm I \} \cup \{ a \}$\;
        $T = (V_T, E_T) \gets$ the Steiner tree of $Q$ in $G_P$ rooted at $a$\;
        
        \For{each edge $p_b \to b \in E_T$ in post-order from $a$}{
            \lIf{$D_{iap_b} \neq \pm I$}{continue}
            $A \gets$ any Pauli string which anticommutes with $D_{iab}$\;
            $B \gets$ any $B \neq \pm I$ (e.g $Z_1$)\;
            Apply $C_{bp_b}(A, B)$ to $S$ and $D$\;
            Prepend $C_{bp_b}(A, B)$ to $C$\;
        }

        \For{each edge $p_b \to b \in E_T$ in post-order from $a$}{
            $A \gets $ \lIf{$p_b = a$}{$S_{iaa}$ \lElse*{any $A$ which anticommutes with $D_{iap_b}$}}
            Apply $C_{p_bb}(A, D_{iab})$ to $S$ and $D$\;
            Prepend $C_{p_bb}(A, D_{iab})$ to $C$\;
        }
    }
    
    \Return{$C$}
\end{algorithm}

\begin{algorithm}
    \caption{The \textsc{DistRowCol} algorithm for distributed Clifford circuit synthesis.}\label{alg:distrowcol}
    \KwIn{The partition connectivity graph $G$, the qubit connectivity graphs of each partition $\{ G_p \}$, and a stabilizer tableau $(S, D)$.}
    \KwOut{A representation $(\mathcal{D}(C_N), C_L)$ of a distributed circuit implementing $(S, D)$.}

    $C \gets$ an empty circuit\;
    $P \gets$ the set of all partitions\;
    \While{$P$ is non-empty}{
        $G_P \gets$ the induced subgraph of $P$ in $G$\; 
        Pick $p \in P$ which is not a articulation point of $G_P$\;
        $C_p \gets \textsc{Disentangle}(G, P, S, D, p)$\;
        Prepend $C_p$ to $C$\;
        Remove $p$ from $P$\;
    }

    \For{each partition $p$}{
        $S_p \gets \{ S_{ipp} \mid i\}$\;
        $D_p \gets \{ D_{ipp} \mid i\}$\;
        $L_p \gets$ synthesize $(S_p, D_p)$ using the method of \cite{winderl2023architecture} with connectivity $G_p$\;
        Prepend $L_p$ to $C$\;
    }

    Find $(\mathcal{D}(C_N), C_L)$ using Algorithm \ref{alg:localcliffpush}\;
    \Return{$(\mathcal{D}(C_N), C_L)$}
\end{algorithm}

\section{Clifford+R\texorpdfstring{\textsubscript{Z}}{Z} Circuit Optimization}
\label{sec:clifft}

In this section, we will extend our analysis to Clifford+R\textsubscript{Z} circuits. While these circuits do not have a normal form in the same way as CNOT of Clifford circuits, and hence cannot be optimized in the same way, we will show how to adapt several important and common operations to the distributed case. In particular, we will first show in Section \ref{sec:distpauliexp} how to resynthesize such a circuit in a way that non-local gates respect the connectivity constraints of the partitions, even if this was not the case for the original circuit. We do this while attempting to minimize the number of non-local gates introduced by adapting the method of \cite{vandaele2022phase} for synthesizing phase polynomial circuits for architectures with restricted qubit connectivity. We will then show two methods of optimizing the number of non-local gates. The first of these in Section \ref{sec:gencnotfold} is similar to the common phase-folding \cite[Section 7.4.3]{kissingerwetering2024book} or phase-teleportation \cite{zhang2019optimizing,kissinger2020reducing} optimization for reducing the number of phase gates, but instead applied to reduce the number of generalized CNOT gates. The second detects Clifford and CNOT subcircuits up to reordering of gates and optimizes them; this may be beneficial in circuits with few phase gates. We note that since our circuit representation is built on top of Pauli exponentials, methods such as phase-folding and other T-count optimization techniques \cite{heyfron2018efficient} can be directly applied on any subcircuit which does not contain non-local gates. These subcircuits can be identified using the method in Section \ref{sec:findsubcirc}.

\subsection{Distributing Pauli Exponentials}
\label{sec:distpauliexp}

In \cite{vandaele2022phase}, a method is introduced to synthesize phase polynomial circuits for architectures with restricted qubit connectivity. It consists of the following main steps: first, the circuit is rewritten into a parity matrix that encodes a Pauli exponential representation of the circuit. Then the circuit is processed one Pauli exponential at a time, with the order selected via a heuristic that is designed to minimize the number of gates generated. For each Pauli exponential, it is reduced to have support on only one qubit using a Steiner tree method, and then it can be implemented using a phase gate and removed from the parity matrix. After there are no more Pauli exponentials remaining, there is a remaining Clifford circuit, which can be synthesized by other means. We will follow a similar set of steps, but replace each of these with its distributed equivalent. The resulting method uses many of the same techniques that were discussed in Section \ref{sec:clifford}.

Suppose that we are given a Pauli exponential gate $P(A, \theta)$ where $A$ spans more than one partition. We wish to conjugate $P(A, \theta)$ by a sequence of generalized CNOT gates $C(A_i, B_i)$ so that it only spans one partition. Since, $U^\dagger P(A, \theta) U = P(U^\dagger A U, \theta)$ for any Clifford unitary $U$, this is the same as finding generalized CNOT gates that conjugate $A$ itself to span only one partition. Let us write $A = \bigotimes_{i} A_i$ in terms of blocks $A_i$ that act on a single partition $i$, similar to how we did in Section \ref{sec:cliffblock}. Then we already considered this problem in Section \ref{sec:distrowcol}, and arrived at the following procedure:
\begin{enumerate}
    \item Suppose we want to eliminate every block of $A$ except for $A_i$ acting on partition $i$, which we call the pivot. Then we define $S = \{ j \mid A_j \neq I\} \cup \{ i\}$ and find the Steiner tree $T = (V_T, E_T)$ of $S$ in the partition connectivity graph $G$, and orient its edges away from $i$.
    \item Now we want to ensure that every block $A_j$ of $A$ is able to be eliminated by its parent. Let $p_j$ be the parent of $j$ in $T$, then we must ensure that $A_{p_j} \neq I$. This can be done by picking $B \neq I$ and $C$ which anticommutes with $A_{j}$, and applying $C_{jp_j}(C, B)$. In Section \ref{sec:distrowcol}, we needed to take greater care to avoid destroying structure in the stabilizer tableau, but this is not necessary here.
    \item Finally we can eliminate every block $A_j$, except $A_i$, using its parent. We can pick any $B$ which anticommutes with $A_{p_j}$, and apply $C_{p_jj}(B, A_j)$. After this step, $A_j = I$ for all $j \neq i$.
\end{enumerate}
However, we have not yet specified how to pick the pivot partition $i$. Note that so long as $A_i \neq I$, any choice of pivot leads to the same number of generalized CNOT gates used to reduce $A$. Thus we must instead consider which pivot will reduce the number of generalized CNOT gates used in \emph{future} reductions of other Pauli exponentials.

Here is where will import some ideas from \cite{vandaele2022phase}. In particular, the heuristic they use for this operation is to try the reduction with every possible pivot, and pick the one for which minimizes the minimum number of gates needed to synthesize the next Pauli exponential; we will use a similar heuristic. However, in the distributed setting, applying any fixed generalized CNOT gate to an arbitrary Pauli string vary rarely leads to a non-identity block becoming the identity: as the block size increases, this probability decreases exponentially. Therefore, we expect that applying a generalized CNOT gate to the remaining Pauli exponentials will increase the number of non-identity blocks by one, any time it is applied to a pair of blocks of which exactly one is the identity. Using this approximation, we can pick the pivot which minimizes the likely increase in non-identity blocks, even without explicitly computing the result of the reduction for each pivot.

The second heuristic we will adapt from \cite{vandaele2022phase} is the ordering in which Pauli exponentials are processed. In particular, they compute the Steiner trees for all of the remaining options and pick the one that minimizes the number of gates that would needed to eliminate it. However, for phase polynomial circuits, all of the Pauli exponentials commute, whereas this is not the case for generic Clifford+R\textsubscript{Z} circuits that we are considering. To take this into account, we can first compute a circuit representation $(\mathcal{D}(C_N), C_L)$ according to Algorithm \ref{alg:localcliffpush} \emph{assuming that there is only one partition}. The resulting DAG $\mathcal{D}(C_N)$ precisely encodes the freedom of choice we have in the ordering of which Pauli exponentials to process. Because it was constructed with only one partition, $\mathcal{D}(C_N)$ will contain exclusively non-Clifford Pauli exponentials. This leads to the following procedure, which is formalized in Algorithm \ref{alg:distpauliexp}:
\begin{enumerate}
    \item Compute $(\mathcal{D}(C_N), C_L)$ assuming that there is only one partition using Algorithm \ref{alg:localcliffpush}.
    \item While $\mathcal{D}(C_N)$ is non-empty, pick the Pauli exponential $P(A, \theta)$ that requires the fewest gates to eliminate and is a source vertex in $\mathcal{D}(C_N)$, and delete it from $\mathcal{D}(C_N)$. Pick the pivot which minimizes the likely increase in the number of gates required in the next iteration. Perform the elimination, applying the same gates to every remaining gate in $\mathcal{D}(C_N)$ and adding it to $C_L$. After this $P(A, \theta)$ can be implemented using a local non-Clifford Pauli exponential.
    \item When no more Pauli exponentials remain, use \textsc{DistRowCol} to synthesize $C_L$, and compute the new circuit representation $(\mathcal{D}(C_N), C_L)$ with the correct set of partitions.
\end{enumerate}
For a circuit with $n$ qubits and at most $r$ phase gates, this uses at most $2(k - 1)(r + n)$ generalized CNOT gates.

\begin{algorithm}
    \caption{Distributing Clifford+R\textsubscript{Z} circuits with arbitrary connectivity.}\label{alg:distpauliexp}
    \KwIn{A Clifford+R\textsubscript{Z} circuit $C$, and a partition connectivity graph $G$\;}
    \KwOut{A representation $(\mathcal{D}(C_N), C_L)$ of $C$ which respects the connectivity in $G$\;}

    Apply Algorithm \ref{alg:localcliffpush} to find $(\mathcal{D}(C_N), C_L)$ \emph{assuming there is only one partition}\; 
    $C' \gets$ an empty circuit\;
    \While{$\mathcal{D}(C_N)$ is non-empty}{
        $V_\mathrm{src} \gets \{ v \in \mathcal{D}(C_N) \mid v \text{ is a source vertex}\}$\;
        \For{$v = P(A, \theta) \in V_\mathrm{src}$}{
            $S_v \gets \{ j \mid A_j \neq I\}$\;
            $(V_v, E_v) \gets$ the Steiner tree of $S$ in $G$\;
            $C_v \gets 2|V_v| - |S_v|$\;
        }
        $v \gets \arg\min \{ C_v \mid v \in V_\mathrm{src}\}$\;
        Delete $v$ from $\mathcal{D}(C_N)$\;
        $p \gets \arg\min \{ \sum_{P(B, \theta) \in \mathcal{D}(C_N)} \mathrm{PivotCost}(B, E_v, p) \mid p \in S_v\}$ \;

        \For{each edge $p_j \to j \in E_v$ in post-order from $p$}{
            \lIf{$A_{p_j} \neq I$}{continue}
            Let $B \neq I$ be arbitrary and $C$ anticommute with $A_j$\;
            Apply $C_{jp_j}(C, B)$ to $A$, append it to $C'$, and prepend it to $C_L$\;
            \lFor{$g \in \mathcal{D}(C_N)$}{$g \gets C_{jp_j}(C, B) \cdot g \cdot C_{jp_j}(C, B)$}
        }
        \For{each edge $p_j \to j \in E_v$ in post-order from $p$}{
            Let $B$ anticommute with $A_{p_j}$\;
            Apply $C_{p_jj}(B, A_j)$ to $A$, append it to $C'$, and prepend it to $C_L$\;
            \lFor{$g \in \mathcal{D}(C_N)$}{$g \gets C_{p_jj}(B, A_j) \cdot g \cdot C_{p_jj}(B, A_j)$}
        }

        Append $P(A, \theta)$ to $C'$\;
    }

    Synthesize $C_L$ using \textsc{DistRowCol} and append the output to $C'$\;
    Compute the final representation $(\mathcal{D}(C_N), C_L)$ using Algorithm \ref{alg:localcliffpush}\;
    \Return{$(\mathcal{D}(C_N), C_L)$}

    \BlankLine
    \SetKwProg{Fn}{procedure}{ do}{end}
    \Fn{$\mathrm{PivotCost}(B, E_v, p)$}{
        $S \gets 0$\;
        \For{each edge $p_j \to j \in E_v$ in post-order from $p$}{
            \If{($B_{p_j} \neq I$ and $B_{j} = I$) or ($B_{p_j} = I$ and $B_{j} \neq I$)}{
                $S \gets S + 1$\;
            }
        }
        \Return{$S$}
    }
\end{algorithm}

\subsection{Generalized CNOT Folding}
\label{sec:gencnotfold}

In this section, we will discuss a optimization for distributed Clifford+R\textsubscript{Z} circuits which is akin to the phase-folding \cite[Section 7.4.3]{kissingerwetering2024book} or phase-teleportation \cite{zhang2019optimizing,kissinger2020reducing} optimization in the non-distributed setting. Given two generalized CNOT gates $C(A, B_1)$ and $C(A, B_2)$ that act on the same pair of partitions and share one Pauli string argument, we can combine them:
$$ C(A, B_1) \cdot C(A, B_2) = \begin{cases}
    C(A, \overline{B_1B_2}) & \text{if } B_1B_2 = B_2B_1 \\
    [P(A, \frac{\pi}{2}) \otimes I] \cdot C(A, \overline{B_1B_2}) & \text{if } B_1B_2 \neq B_2B_1
\end{cases}$$
Applying this reduces the number of non-local gates in the circuit by at least one, and so we would like to do it wherever possible. In particular, we will use the greedy strategy of repeatedly finding and combining compatible gates in the circuit until none remain, giving priority to those with $B_1 = B_2$, since this will remove two non-local gates instead of just one. 

Gates can be merged if they satisfy above criteria and also it is possible to reorder the circuit (through a sequence of commutations) so that these gates are placed next to each other. This condition can be detected using the DAG representation $\mathcal{D}(C_N)$ of the circuit: given two vertices $u$ and $v$, there exists a topological order of $\mathcal{D}(C_N)$ where $u$ and $v$ are consecutive if and only if there is no intermediate vertex; that is, no successor of $u$ is a predecessor of $v$, or vice-versa.

In the case that $B_1$ and $B_2$ anti-commute, an extra local Clifford operation $P(A, \frac{\pi}{2})$ will be introduced. Likewise if $B_1B_2 = -I$, then $C(A, \overline{B_1B_2}) = P(A, \pi)$. In either case, we can remove it from $C_N$ by `pushing' it through all subsequent gates using the following two relations:
\begin{align*}
    P(B, -\frac{k\pi}{2}) \cdot P(A, \theta) \cdot P(B, \frac{k\pi}{2}) &= \begin{cases}
        P(A, \theta) & \text{if } AB = BA \\
        P(A, (-1)^\frac{k}{2} \theta) & \text{if } k \text{ even} \\
        P(\overline{AB}, (-1)^\frac{k-1}{2} \theta) & \text{if } k \text{ odd}
    \end{cases} \\
    (P(B, -\frac{k\pi}{2}) \otimes I) \cdot C(A, C) \cdot (P(B, \frac{k\pi}{2}) \otimes I) &= \begin{cases}
        C(A, C) & \text{if } AB = BA \\
        C((-1)^\frac{k}{2}A, C) & \text{if } k \text{ even} \\
        C((-1)^\frac{k-1}{2}\overline{AB}, C) & \text{if } k \text{ odd}
    \end{cases}
\end{align*}
The full procedure is described in Algorithm \ref{alg:gencnotfold}. While it was not necessary for our purposes, for extremely large circuits it may be beneficial to avoid reconstructing the DAG $\mathcal{D}(C_N)$ at each step, and use an online transitive closure algorithm \cite{sankowski2004dynamic} to quickly answer reachability queries.

\begin{algorithm}
    \caption{Generalized CNOT gate folding.}\label{alg:gencnotfold}
    \KwIn{The circuit representation $(\mathcal{D}(C_N), C_L)$.}
    \KwOut{An optimized circuit representation $(\mathcal{D}(C_N), C_L)$.}

    \While{can make progress}{
        \For{each pair of gates $a, b \in \mathcal{D}(C_N)$, with $a = b$ taken first}{
            \lIf{$a$ and $b$ are not of the form $C_{ij}(A, B_2)$, $C_{ij}(A, B_1)$}{continue}
            \lIf{$b$ is reachable from any child of $a$ or vice-versa}{continue}
            Apply Lemma \ref{thm:commdag} to find a topological order $\mathcal{D}(C_N)$ with $a$ and $b$ consecutive\;
            Let $[u_1, \dots, u_n, a, b, v_1, \dots, v_m]$ be that topological order\;
            \uIf{$B_1B_2 \neq B_2B_1$}{
                \lFor{$j = 1$ \KwTo $m$}{$v_j \gets P_i(A, -\frac{\pi}{2}) \cdot v_j \cdot P_i(A, \frac{\pi}{2})$}
                Prepend $P_i(A, \frac{\pi}{2})$ to $C_L$\;
            }\ElseIf{$B_1B_2 = -I$}{
                \lFor{$j = 1$ \KwTo $m$}{$v_j \gets P_i(A, \pi) \cdot v_j \cdot P_i(A, \pi)$}
                Prepend $P_i(A, \pi)$ to $C_L$\;
            }
            \uIf{$B_1B_2 \neq \pm I$}{
                $C_N \gets [u_1, \dots, u_n, C_{ij}(A, \overline{B_1B_2}), v_1, \dots, v_m]$\;
            }\Else{
                $C_N \gets [u_1, \dots, u_n, v_1, \dots, v_m]$\;
            }
            Reconstruct $\mathcal{D}(C_N)$ as in Algorithm \ref{alg:localcliffpush} and break\;
        }
    }
\end{algorithm}

\subsection{Finding Optimizable Subcircuits}
\label{sec:findsubcirc}

Since we have found algorithms for synthesizing distributed CNOT and Clifford circuits, one potential optimization for distributed Clifford+T circuits is to find CNOT or Clifford subcircuits and resynthesize them. Since the size of the output of these resynthesis algorithms does not depend strongly on the size of the input, we would like to find the largest possible subcircuits, and ideally this would be up to any possible reordering of gates. In the circuit representation $(\mathcal{D}(C_N), C_L)$, we phrase this as follows: we want to partition the DAG $\mathcal{D}(C_N)$ into as few bins as possible satisfying a given predicate (for instance, being a CNOT or Clifford subcircuit). For this to be well defined, we may only pick bins for which there is a topological ordering of $\mathcal{D}(C_N)$ where the bins are contiguous; this implies that there is an equivalent reordering of the gates of the circuit where the bins form genuine subcircuits. In the graph-theoretic terms, this is known as a \emph{convex acyclic partition} of a DAG \cite{herrmann2017acyclic}; we give an example in Figure \ref{fig:dagpart-example}. We define it as follows:

\begin{figure}
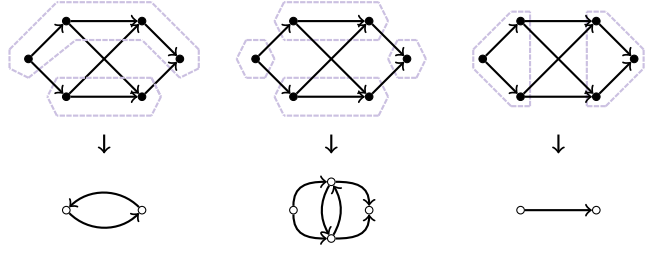

    \ctikzfig{dag-partition-example}
    \caption{An example of different types of partitions of a DAG, as defined in Definition \ref{def:dagpart}. Above shows the partitions, while below is the contracted graph. \emph{Left:} a partition that is neither convex nor acyclic. \emph{Center:} a partition that is convex, but not acyclic. \emph{Right:} a convex acyclic partition.}\label{fig:dagpart-example}
\end{figure}

\begin{definition}
    Given a DAG $D = (V, E)$, write $u \rightsquigarrow v$ if $v$ is reachable from $u$ in $D$. A subset $U \subseteq V$ is \emph{convex} if its predecessors $P = \{ v \in V \setminus U \mid \exists u \in U, v \rightsquigarrow u \}$ and successors $S = \{ v \in V \setminus U \mid \exists u \in U, u \rightsquigarrow v\}$ are disjoint. A \emph{down-set} is a convex subset such that $P = \emptyset$.
\end{definition}

\begin{definition}
    \label{def:dagpart}
    An \emph{acyclic partition} of a DAG $D = (V, E)$ is a partition of $V$ into disjoint subsets $U_1 \cup U_2 \cup \cdots \cup U_k = V$ such that the contracted graph $G_C = (V_C, E_C)$ with $V_C = \{ U_1, \dots, U_k \}$, $E_C = \{ i \to j \mid \exists a \in U_i, b \in U_j, a \rightsquigarrow b\}$ is acyclic. A \emph{convex acyclic partition} is an acyclic partition where every subset is convex.
\end{definition}

This is indeed sufficient to guarantee the existence of a topological ordering that makes the bins contiguous: find a topological order of each subset $U_i$ and concatenate them according to a topological order of $G_C$. Therefore, we define the following problem; the restriction on $f$ is to guarantee that every vertex can be fit into a bin, if this were not the case then no acceptable partition exists.

\begin{problem}[Bin Packing on a DAG]
    Suppose we have a DAG $D = (V, E)$ and a predicate $f \in 2^V \to \mathbb{B}$ such that $f(\{ v\}) = 1$ for all $v \in V$. Find an acyclic convex partition $U_1 \cup \cdots \cup U_k = V$ which minimizes $k$, subject to $f(U_i) = 1$ for all $i$.
\end{problem}

When $E = \emptyset$ and $f(U) = 1$ iff $\sum_{u \in U} u \leq C$ this reduces to the standard bin packing problem, and so it is NP-hard. On the other hand, if $D$ is a path graph then there is only one topological order of $D$, and the optimal bin boundaries can be decided in polynomial time via dynamic programming; the canonical example of this is the Knuth-Plass line-breaking algorithm used in \LaTeX ~\cite{knuth1981breaking}. Therefore, the difficult part of this problem is picking a good topological order of $D$. For this, we propose a heuristic extension of the Knuth-Plass method using beam search.

The first insight we need is that any prefix of the optimal partition must itself be optimal: suppose that $U_1, U_2, \cdots U_k$ is a partition with the bins listed in topological order, then if any prefix $U_1, \dots, U_i$ were suboptimal, we could replace it with a smaller set of bins $U'_1, \dots, U'_j$ that partition the same vertices. Second, note that the set of vertices in any prefix of a topological order must be a \emph{down-set} of vertices, because otherwise there are is at least one vertex that is ordered before its parent. Hence, let $\mathrm{OptPart}(V)$ be the function that gives the optimal partitioning of a down-set $V$, then we can give it an inductive definition (i.e the Bellman equation)
$$ \mathrm{OptPart}(V) = \min \left\{ \mathrm{OptPart}(V') \cup \{V \setminus V'\} \mid f(V \setminus V') = 1 \text{ and } V' \subseteq V \text{ is a down-set}\right\} $$
where the minimum is taken by comparing the number of bins. If we incrementally compute and memoize the values of $\mathrm{OptPart}(V)$ for all down-sets with increasing sizes, then computing the optimal partition of the whole DAG can be done in time polynomial in the number of down-sets. Unfortunately, for most DAGs, the number of down-sets grows exponentially with the number of vertices (in particular, if $\mathcal{D}(C_N)$ contains a set of $k$ consecutive mutually commuting gates then it has at least $2^k$ down-sets). To make this tractable, we apply beam search with beam width $B$: at each step, we keep only the top $B$ new down-sets with the fewest number of bins seen so far, and when computing larger down-sets we only consider possible $V'$ which were previously memoized.

This strategy is formalized as Algorithm \ref{alg:dag-bin-packing}, and runs in time $O(T_fk^3B^2\log B)$ for a DAG of $k$ vertices, where $T_f$ is the time needed to evaluate $f$ for any set of vertices. A more sophisticated implementation can improve this to $O(T_ik^3B\log B)$,where $T_i$ is the time needed to evaluate $f$ incrementally (i.e compute $f(V \cup \{v\})$ given $f(V)$), but we omit the details here. The application to resynthesizing Clifford subcircuits is straightforward; define the predicate $f$ as 
$$ f(V) = \begin{cases}
    1 & |V| = 1 \\
    1 & |V| > 1 \text{ and } \forall g \in V, ~g \text{ is Clifford} \\
    0 & \text{otherwise}
\end{cases} $$
and then find a convex acyclic partition of $\mathcal{D}(C_N)$. For each bin of at least two vertices (which are necessarily all Clifford), find the appropriate subcircuit, use \textsc{DistRowCol} to resynthesize it, and replace the subcircuit if the resynthesized output is smaller than the original subcircuit.

\begin{algorithm}
    \caption{A beam search method for bin packing on a DAG.}\label{alg:dag-bin-packing}
    \KwIn{A DAG $D = (V, E)$, a predicate $f \in 2^V \to \mathbb{B}$, and the beam width $B > 1$.}
    \KwOut{A convex acyclic partition $U_1 \cup \cdots \cup U_k = V$ such that $f(U_i) = 1$ for all $i$.}

    $M \gets$ an empty dictionary\;
    $S \gets$ an empty priority queue\;
    \For{each source vertex $v \in V$}{
        $M[\{v\}] \gets [\{v\}]$\; 
        Insert $\{v\}$ into $S$ with priority $1$\;
    }
    \For{$i = 1$ \KwTo $|V| - 1$}{
        $S' \gets$ an empty priority queue\;
        \For{each $V_f$ in $S$ in ascending order by priority}{
            \For{each source vertex $v$ of $D[V \setminus V_f]$}{
                $V_n \gets V_f \cup \{v\}$\;
                \For{each $V_p$ in $M$}{
                    $[U_1, \dots, U_k] \gets M[V_p]$\;
                    \lIf{$V_p = V_n$ or $V_p \nsubseteq V_n$ or $|M[V_n]| < k + 1$ or $f(V_n \setminus V_p) = 0$}{continue}
                    \If{$k + 1 <$ the maximum priority in $S'$}{
                        $V_{max}\gets$ the element in $S'$ with maximum priority\;
                        Delete $V_{max}$ from $S'$ and $M$\;
                    }
                    \If{$|S'| < B$}{
                        $M[V_n] \gets [U_1, \dots, U_k, V_n \setminus V_p]$\;
                        Insert $V_n$ into $S'$ with priority $k + 1$\;
                    }         
                }
            }
        }
        $S \gets S'$\;
    }

    \Return{$M[V]$}
\end{algorithm}

\subsection{Detecting CNOT Subcircuits}

Since \textsc{BlockRowCol} can be approximately optimal in some scenarios where \textsc{DistRowCol} is not (i.e parity matrices with low-rank off-diagonal blocks), we would also like to find subcircuits of CNOT gates in order to resynthesize them. In fact, we can be slightly more general: suppose that a subcircuit of generalized CNOT gates $U$ can be written as $U = C U' C^\dagger$ for some local Clifford circuit $C$ and a circuit $U'$ of ZX-type generalized CNOT gates, then we can apply \textsc{BlockRowCol} to $U'$ and conjugate the result by $C$ to resynthesize $U$. The remainder of this subsection is dedicated to showing that, given $U$, we can determine if $C$ exists, and construct it if it does, in polynomial time. This allows us to use the methods of the previous subsection to find large optimizable subcircuits.

First we give the following lemma which characterizes when a set of Pauli strings $S$ can be conjugated by a Clifford unitary to obtain all Z-type or X-type Pauli strings. An alternative way to think about this is to decide whether there exists a Clifford unitary $U$ for which each element of $S$ is either a product of stabilizers of $U$ or a product of destabilizers of $U$. Note that since
$$ (U \otimes I) C(A, B) (U^\dagger \otimes I) = C(UAU^\dagger, B) $$
then if we have $A \in S$ for every gate in a circuit of generalized CNOT gates, conjugating the corresponding partition by $U$ would transform all generalized CNOT gates to be Z-type or X-type on that partition. Hence, we apply this result to each partition of the circuit; if the premises of the lemma can be satisfied in each partition in a way that are compatible with each other, then it is equal to a circuit of ZX-type generalized CNOT gates up to conjugation. The subsequent lemma shows that this satisfiability problem can be decided in polynomial time.

\begin{lemma}
    \label{thm:zxstabsfrompaulis}
    Given two sets of mutually commuting $n$-qubit Pauli strings $U = \{ U_i \}$ and $V = \{ V_i \}$ such that the set of all Pauli strings generated by $U$ and $V$ do not intersect, there exists a Clifford unitary $C$ and sets of indices $A_i$ and $B_i$ such that $C U_i C^\dagger = \prod_{j \in A_i} Z_j$ and $C V_i C^\dagger = \prod_{j \in B_i} X_j$ for all $i$.
\end{lemma}

\begin{proof}
    First, assume without loss of generality that $U$ and $V$ are each linearly independent sets. Let $M$ be a $\mathbb{F}_2$-matrix of $|U|$ rows and $|V|$ columns such that $M_{ij} = 1$ if and only if $U_i V_i \neq V_i U_i$. Suppose were were to redefine $U$ by taking $U_i \gets U_i U_j$, then this would modify $M$ by adding the $j$th row to the $i$th row. Similarly redefining $V$ by $V_i \gets V_iV_j$ corresponds to adding the $j$th column of $M$ to the $i$th column. Note that both of these operations do not change the set of strings generated by $U$ and $V$. Therefore, we do both row and column reduction to $M$, updating $U$ and $V$, respectively, whenever row and column operations are performed. $M$ will now be zero everywhere except for the top-left block of size $k \geq 0$ which will be the identity matrix. Let $U'$ and $V'$ contain the first $k$ elements of $U$ and $V$, then we have $U'_i U'_j = U'_j U'_i$, $V'_i V'_j = V'_j V'_i$, and $U'_i V'_j = V'_j U'_i$ for all $i \neq j$, and $U'_i V'_i \neq V'_i U'_i$ for all $i$. Therefore, proceeding as in the proof of Lemma \ref{thm:blockfindanticommute}, we can construct a Clifford transformation $C$ such that $C U'_i C^\dagger = Z_i$ and $C V'_i C^\dagger = X_i$ for all $i \leq k$. 
    
    The elements in $U \setminus U'$ and $V \setminus V'$ commute with every Pauli string in $U$ and $V$, because the corresponding rows and columns of $M$ are all zero. Therefore, for every $P \in (U \setminus U') \cup (V \setminus V')$ we must have $C P C^\dagger$ be the identity on the first $k$ qubits, since they commute with $Z_i$ and $X_i$ for all $i \leq k$. Moreover, because $(U \setminus U') \cup (V \setminus V')$ is mutually commuting, there exists a change of basis $C'$ which co-diagonalizes the whole set \cite{murairi2023reducing}, and maps $U'$ and $V'$ to themselves. Since by assumption $U$ and $V$ are each linearly independent, $C'$ can be chosen so that each string in $U \setminus U'$ and $V \setminus V'$ has support on a single qubit, and since the sets of strings generated by $U$ and $V$ are disjoint, $C'$ can be chosen to map $U \setminus U'$ and $V \setminus V'$ onto disjoint sets of qubits. In particular, pick $C'$ so that $C'C$ maps $U \setminus U'$ to $\{ Z_i \mid k < i \leq |U|\}$ and maps $V \setminus V'$ to $\{ X_i \mid |U| < i < |U| + |V| - k\}$. Now, since the map $C'C$ has transformed all of the redefined elements of $U$ and $V$ to be Pauli strings of the form $Z_i$ and $X_i$, respectively, it must transform the original elements of $U$ and $V$, which are linear combinations of these, so that each $U_i$ is mapped to $\prod_{j \in A_i} Z_j$ and each $V_i$ is mapped to $\prod_{j \in B_i} X_j$ for some sets of indices $A_i$ and $B_i$. 
\end{proof}

\begin{definition}
    The vectorization of an $n$-qubit Pauli string $P$ is a $2n$-dimensional binary vector $\vec{v}(P)$ with $\vec{v}(P)_i = 1$ iff $P_i \in \{Z, Y\}$ and $\vec{v}(P)_{n + i} = 1$ iff $P_i \in \{X, Y\}$ for all $1 \leq i \leq n$.
\end{definition}

\begin{lemma}
    For any circuit of generalized CNOT gates, it can be determined in polynomial time whether there is a local Clifford transformation $C$ under which every generalized CNOT gate becomes ZX-type, and if $C$ exists, it can be constructed explicitly.
\end{lemma}

\begin{proof}
    For each partition $p$, we define a pair of sets $U_p$ and $V_p$, to which we will assign all the Pauli strings corresponding to the portion of each generalized CNOT gate that acts on partition $p$. From Lemma \ref{thm:zxstabsfrompaulis}, $C$ exists if and only if we can find such an assignment for which $U_p$ and $V_p$ are each mutually commuting and do not intersect for all $p$. Consider a set of binary variables $x_i$, where each $i$ corresponds to a Pauli string $P_i$ associated to one half of a generalized CNOT gate. Define $p_i$ as the partition to which $P_i$ belongs, and $h(i)$ as the label of the Pauli string associated to the other half of the generalized CNOT gate containing $P_i$. 
    
    We will define constraints on these variables, which, if solved, provide an assignment of each $i$ to $U_{p_i}$ (when $x_i = 1$) or $V_{p_i}$ (when $x_i = 0$). The constraints are as follows: each generalized CNOT gate must have one Z-type Pauli string and one X-type Pauli string, which implies $x_i \neq x_{h(i)}$ for all $i$. Each set $U_p$ and $V_p$ must mutually commute. Hence, for all $i$ and $j$ with $p_i = p_j$, we have $x_i \neq x_j$ if $P_iP_j \neq P_jP_i$. The last and most complicated condition is that the set of strings generated by $U_p$ and $V_p$ cannot intersect, for all $p$. Let $S_p$ be the set of all indices $i$ with $p_i = p$, and define a matrix $M_p$ with columns given by $\vec{v}(P_i)$ for all $i \in S_p$. Then $U_p$ and $V_p$ do not intersect if and only if no linear dependence in $M_p$ has a non-zero contribution from both columns corresponding to elements of $U_p$ and columns corresponding to elements of $V_p$. This can be checked in the following way: let $N_p$ be a basis for the null space of $M_p$, then for every $\vec{n} \in N_p$, add the constraint that $\bigoplus_{i \in S_i} x_i n_i P_i = 0$. This constraint is saying that every linear dependence $\vec{n}$ has a net-zero contribution from the columns labeled with $x_i = 1$, which are exactly those assigned to $U_p$. Note that since these constraints are linear, it is sufficient to consider only basis vectors of the null-space.

    Since each of the constraints added are linear over $\mathbb{F}_2$ (as $x_i \neq x_j$ is equivalent to $x_i \oplus x_j = 1$), and the number of constraints is polynomial in the total number of qubits, the whole system can be solved by Gaussian elimination in polynomial time. Given an assignment $\{ x_i \}$, constructing the local transformation $C$ can be done by applying Lemma \ref{thm:zxstabsfrompaulis} to each pair of sets $U_p$ and $V_p$ to determine a transformation $C_p$ on each partition $p$, and then taking $C = \bigotimes_p C_p$.
\end{proof}
Note that the method given in this proof will be very slow as the number of qubits and gates increases, because the number of rows in the constraint matrix increases quadratically. Therefore, in practice we optimize it in two ways: first, we perform a row reduction on each matrix $M_p$ individually, and remove all of the zero rows. Second, we solve all of the inequality constraints ahead of time using a union-find method \cite{tarjan1984worst}, and eliminate any columns of the constraint matrix corresponding to dependent variables.  

\section{Circuit Resynthesis}
\label{sec:resynth}

Here we briefly describe a strategy to convert the $(\mathcal{D}(C_N), C_L)$ circuit representation back into a standard quantum circuit, while respecting the qubit connectivity of the underlying architecture. Although our main goal is to minimize the number of non-local gates, we would also like to avoid generating many local two-qubit gates, where possible. We propose a simple greedy strategy based on Steiner trees, similar to \cite{nash2020quantum}. Note that good algorithms have been developed for architecture-aware phase-polynomial synthesis \cite{vandaele2022phase} and it is likely that these could be adapted to our use-case; we leave this for future work.

We incrementally construct a topological order of $\mathcal{D}(C_N)$: at each step, find a source node in $\mathcal{D}(C_N)$ that corresponds to a gate $g$ with minimal synthesis cost. Then synthesize it, delete it from $\mathcal{D}(C_N)$, and update the gates in $C_N$ and $C_L$ with the local Clifford transformations that were applied during synthesis. After all non-local and non-Clifford gates have been processed, $C_L$ can be synthesized using any architecture-aware Clifford synthesis algorithm, for example \cite{winderl2023architecture}. This is formalized as Algorithm \ref{alg:resynth}. For each synthesis operation of a gate $g$, we proceed as follows, where the cost is counted as the number of two-qubit gates:
\begin{enumerate}
    \item First, we apply single-qubit Clifford transformations: if $g = P_i(A, \theta)$ for some $A$, then we will apply single-qubit gates so that $A$ consists only of $Z$s. For each qubit $k$ where $A_k \notin \{I, Z\}$, apply either $H$ or $SH$ to map it to $Z$. If $g = C_{ij}(A, B)$, we do the same procedure for both $A$ and $B$.
    \item Next we will pick one qubit for each partition affected by $g$, that we call the pivot, where the non-Clifford or non-local gates will be performed. For $g = P_i(A, \theta)$, we pick the pivot in $i$ arbitrarily as any qubit $k$ for which $A_k \neq I$. For $g = C_{ij}(A, B)$, we also need to take into account that not all qubits in a partition may be able to perform non-local gates. Let $Q_{i \to j}$ be a list of qubits of partition $i$ that are able to perform non-local gates targeting partition $j$, and $G_i$ be the qubit connectivity graph of partition $i$. Then we pick the pivot in partition $i$ as the $k \in Q_{i \to j}$ which minimizes the size of the Steiner tree of $\{ l \mid A_l \neq I  \} \cup \{k\}$ in $G_i$. Similarly, the pivot in partition $j$ is $k \in Q_{j \to i}$ which minimizes the size of the Steiner tree of $\{ l \mid B_l \neq I  \} \cup \{k\}$ in $G_j$.
    \item Now we can eliminate all non-pivot qubits. For each partition $i$ affected by $g$, let $A$ be the corresponding Pauli string and $k$ be the pivot. Find the Steiner tree $T = (V_T, E_T)$ of $\{ l \mid A_l \neq I  \} \cup \{k\}$ in $G_i$ rooted at $k$, and orient the edges away from $k$. First, while there is $a \in V_T$ so that $A_a = I$, we find any neighbor $b$ of $a$ in $T$ such that $A_b = Z$ and perform $CNOT_{ab}$. Then, we traverse the edges $a \to b$ of the tree from the leaves to the root, and apply $CNOT_{ab}$.
    \item Finally, we can place the actual non-local or non-Clifford gate: if $g = P_i(A, \theta)$, then we add $R_Z(\theta)$ to the pivot in $i$. If $g = C_{ij}(A, B)$, then we apply a CZ gate between the pivots in $i$ and $j$.
\end{enumerate}
In practice, we slightly tweak this construction to try and minimize also the number of single-qubit Clifford gates, which has the added benefit of always synthesizing ZX-type generalized CNOT gates as a subcircuit of CNOT gates; we omit the details here.

\begin{algorithm}
    \caption{A procedure to convert $(\mathcal{D}(C_N), C_L)$ to a Clifford+R\textsubscript{Z} circuit.}\label{alg:resynth}
    \KwIn{The representation $(\mathcal{D}(C_N), C_L)$, the qubit connectivity graphs of each partition $\{ G_i \}$, and the list of non-local capable qubits $\{ Q_{i \to j} \}$ for all partitions $i$ and $j$.}
    \KwOut{A Clifford+R\textsubscript{Z} circuit.}

    $C \gets$ an empty circuit\;
    \While{$\mathcal{D}(C_N)$ is non-empty}{
        $S \gets \{ g \in \mathcal{D}(C_N) \mid g \text{ has no incoming edges} \}$\;
        $g \gets \arg\min \{ \text{size of } C_g \mid \forall g \in S, ~(C_g, g_N) \gets \mathrm{SynthGate}(g) \}$\;
        Delete $g$ from $\mathcal{D}(C_N)$\;
        $C_g, g_N \gets \mathrm{SynthGate}(g)$\;
        Append $C_g$ and $g_N$ to $C$\;
        Prepend $C_g^\dagger$ to $C_L$\;
        \lFor{$U \in \mathcal{D}(C_N)$}{$U \gets C_gUC_g^\dagger$}
    }
    $C'_L \gets$ synthesize $C_L$ using the method of \cite{winderl2023architecture} with connectivity graph $\bigcup_i G_i$\;
    \Return{$C + C'_L$}
    \BlankLine
    \SetKwProg{Fn}{procedure}{ do}{end}
    \Fn{$\mathrm{SynthGate}(g)$}{
        \uIf{$g = C_{ij}(A, B)$}{
            $C_A, p \gets \mathrm{PickPivot}(A, G_i, Q_{i \to j})$\;
            $C_B, q \gets \mathrm{PickPivot}(B, G_j, Q_{j \to i})$\;
            \Return{$C_A + C_B$, $CZ_{pq}$}
        }\ElseIf{$g = P_i(A, \theta)$}{
            $C_A, p \gets \mathrm{PickPivot}(A, G_i, \{ p \mid A_p \neq I \})$\;
            \Return{$C_A$, $R_Z(p, \theta)$}
        }
    }
    \BlankLine
    \Fn{$\mathrm{PickPivot}(A, G, Q)$}{
        $S \gets \{ i \mid A_i \neq I \}$\;
        \uIf{$S \cap Q \neq \emptyset$}{
            $p \gets$ any element of $S \cap Q$\;
        }\Else{
            $p \gets \arg\min \{ |V_T| \mid \forall p \in Q, (V_T, E_T) \gets \text{the Steiner tree of } S \cup \{ p \} \text{ in } G \}$\;
        }
        $V_T, E_T \gets$ the Steiner tree of $S \cup \{p\}$ in $G$, rooted at $p$\;
        $C \gets$ an empty circuit\;
        \For{$i \in S$}{
            \lIf{$A_i = X$}{append $H_i$ to $C$}
            \lElseIf{$A_i = Y$}{append $S_i$, $H_i$ to $C$}
        }
        \For{each edge $p_j \to j \in E_T$ in post-order from $p$}{
            \lIf{$A_{p_j} = I$}{append $CNOT_{p_jj}$ to $C$}
        }
        \For{each edge $p_j \to j \in E_T$ in post-order from $p$}{
            Append $CNOT_{jp_j}$ to $C$\;
        }
        \Return{$C$, $p$}
    }

\end{algorithm}

\section{Benchmarks}
\label{sec:benchmarks}

We will now present some preliminary benchmarking results for the algorithms developed in this work. We look at three different scenarios: random CNOT circuits, random Clifford circuits, and random Clifford+T circuits. These are generated by picking gates uniformly at random: for Clifford and Clifford+T circuits, we set the probability of generating a single-qubit gate as $40\%$, evenly divided amongst $H$ and $S$ or $H$, $S$, and $T$, respectively. If a single-qubit gate is not generated, a uniformly random CNOT gate is generated instead. We consider two different families of architectures:
\begin{enumerate}
    \item $n = 36$ qubits grouped into $k = 9$ blocks of four qubits each. We apply this to three different connectivities: all-to-all, a $3\times 3$ 2D grid, and linear nearest neighbor.
    \item $k = 2$ blocks of between $4$ and $16$ qubits, leading to $n \in [8, 32]$ (there is only one possible connectivity to consider for $k = 2$).
\end{enumerate}
For each of these, we compare three different methods:
\begin{enumerate}
    \item The algorithms developed in this work. In particular, we use \textsc{BlockRowCol} for CNOT circuits, \textsc{DistRowCol} for Clifford circuits, and Algorithm \ref{alg:distpauliexp} for Clifford+T circuits. In each case, we also apply both standard phase-folding \cite[Section 7.4.3]{kissingerwetering2024book} and then generalized CNOT gate folding to the resulting circuit using Algorithm \ref{alg:gencnotfold}. We verify that the synthesis is correct by comparing stabilizer tableaux, and using the \texttt{full-reduce} method of QuiZX \cite{quizx} for Clifford+T circuits.
    \item The implementation of the methods described in \cite{wu2023entanglement} and \cite{andresmartinez2024distributing}, as available in the \texttt{pytket-dqc} Python package. These methods work by grouping non-local gates together in a circuit to minimize the number of qubit teleportations between partitions that are needed to implement them. This is achieved by mapping the problem to a combination of hypergraph partitioning and vertex covers. In particular, we used the \textsc{CoverEmbeddingSteinerDetached} strategy, since \cite{andresmartinez2024distributing} suggests that it is usually the best performing of the strategies suggested there. Note that \texttt{pytket-dqc} reorders which qubits appear in which partitions in order to minimize the non-local gate count. Our methods do not do this, but since this can only bias the results against our methods, we consider this acceptable.
    \item Since \texttt{pytket-dqc} only optimizes the placement and teleportation of qubits, and does not perform any optimization or resynthesis of the underlying circuit, we also consider combining \texttt{pytket-dqc} with a standard resynthesis method. For CNOT circuits we use RowCol \cite{wu2023optimization}, for Clifford circuits we use the method of van den Berg \cite{vandenberg2020simple}, and for Clifford+T circuits we run Algorithm \ref{alg:distpauliexp} with one qubit per partition. After the resynthesis, we apply \textsc{CoverEmbeddingSteinerDetached} as above.
\end{enumerate}
The results are presented in Figures \ref{fig:benchcnot}, \ref{fig:benchclifford}, and \ref{fig:benchclifft}. Regrettably, some data points are quite noisy due to computational constraints; the data presented here already represent approximately 1000 core-hours of computation. Specific conclusions are discussed further in the figure captions, and we also draw some general conclusions here:
\begin{itemize}
    \item For shallow circuits, \texttt{pytket-dqc} without resynthesis is usually the best performer, followed by \texttt{pytket-dqc} with resynthesis. This is unsurprising: generally speaking, circuit resynthesis is unhelpful for small circuits, and qubit reordering becomes more helpful in this situation. This is especially relevant for Clifford+T circuits, where resynthesis seems to be unhelpful even for relatively deep circuits.
    \item For deeper circuits, our methods usually outperform or are comparable to \texttt{pytket-dqc} with resynthesis. This is particularly apparent for larger numbers of qubits and less well-connected architectures like linear nearest neighbor; \textsc{DistRowCol} is especially insensitive to connectivity. The exception to this is for small numbers of qubits (e.g $n = 8$), for which our methods do relatively badly (especially Algorithm \ref{alg:distpauliexp}); this may be due to the lack of qubit reordering.
\end{itemize}
Lastly, we note that our algorithms are \emph{much} faster than \texttt{pytket-dqc} for the same size of circuits (up to four orders of magnitude). While some of this is likely due to the difference in implementation language (Rust vs Python), this probably does not account for all of it (indeed, one bottleneck is the KaHyPar library used for hypergraph partitioning \cite{schlag2022highquality}, which is written in C++). This is not intended to discredit \texttt{pytket-dqc}, since it is performing a much more sophisticated method; however, one benefit of simple tableau-based algorithms like we have presented here is that they can scale very large. In Figure \ref{fig:benchspeed}, we show a comparison of total runtime for distributed CNOT circuit synthesis with different architecture configurations, including scaling of \textsc{BlockRowCol} up to $n = 512$ qubits, $k = 128$ partitions, and circuits containing more than $2\cdot 10^5$ gates. Our implementation is available on GitHub \cite{github}. In the future, we would like to give a more detailed set of benchmarks on more realistic circuit structures and architectures, and compare to a wider range of algorithms.

\begin{figure}
    \includegraphics[width=\textwidth]{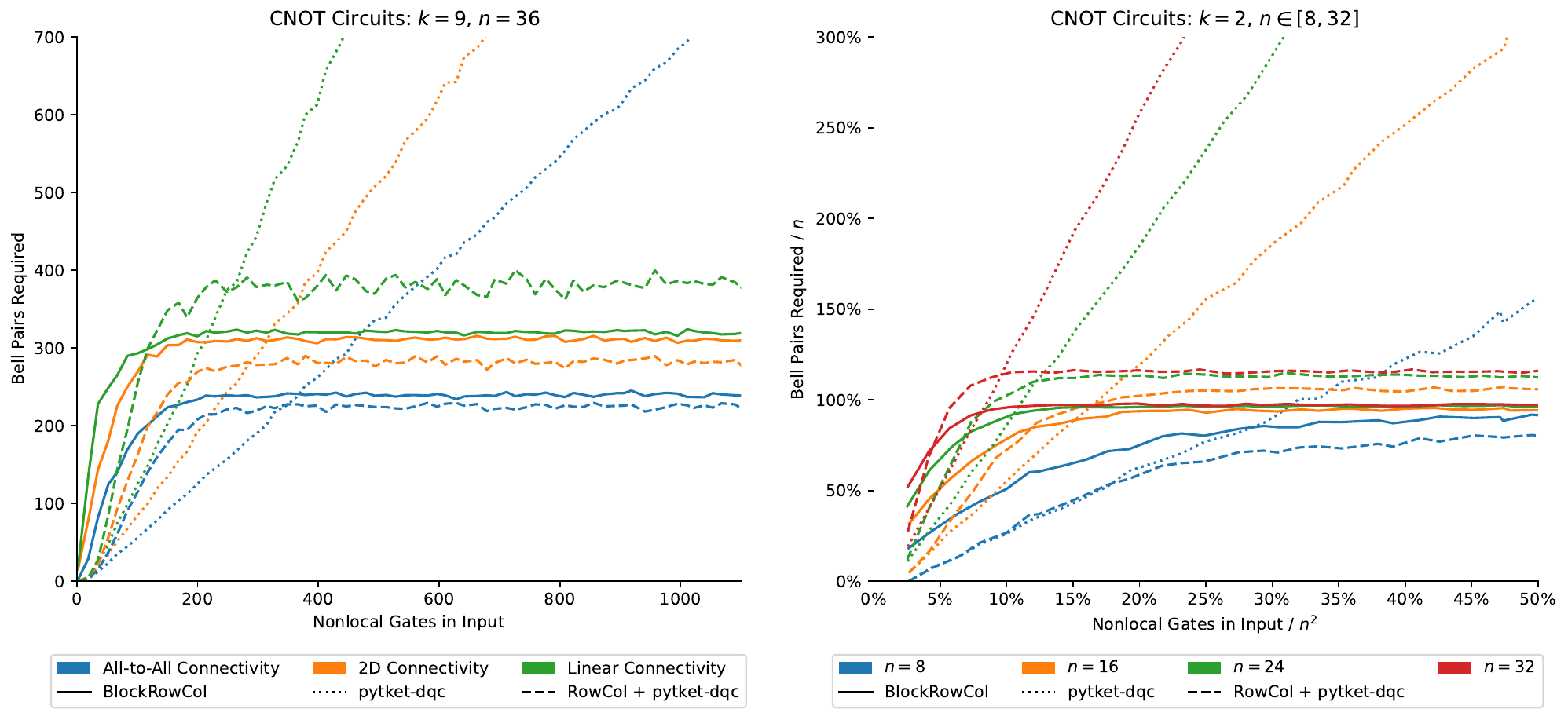}
    \caption{A comparison of distributed CNOT circuit synthesis methods. \emph{Left:} The case of $k = 9$ partitions and $n = 36$ qubits, organized into three architectures: all-to-all, a 3x3 2D grid, and a line. For shallow circuits, \texttt{pytket-dqc} is the clear winner, while for deeper circuits, we see that there is a significant difference between different connectivities; \textsc{BlockRowCol} and \texttt{pytket-dqc} with \textsc{RowCol} are comparable for all-to-all connected, while \textsc{BlockRowCol} is better for linear connectivity and worse for 2D connectivity. \emph{Right:} The case of $k = 2$ partitions with a varying number of qubits $n$. The number of non-local input gates is normalized against $n^2$ for simpler comparability, while the output is normalized against $n$ (which is how we expect it to scale, given that $k$ is fixed). We see that \textsc{BlockRowCol} becomes relatively better as the depth and number of qubits increases. In the case of $n = 8$ \textsc{BlockRowCol} consistently looses. This may be because it does not support qubit reordering, which may make a larger difference for smaller qubit counts.}\label{fig:benchcnot}
\end{figure}

\begin{figure}
    \includegraphics[width=\textwidth]{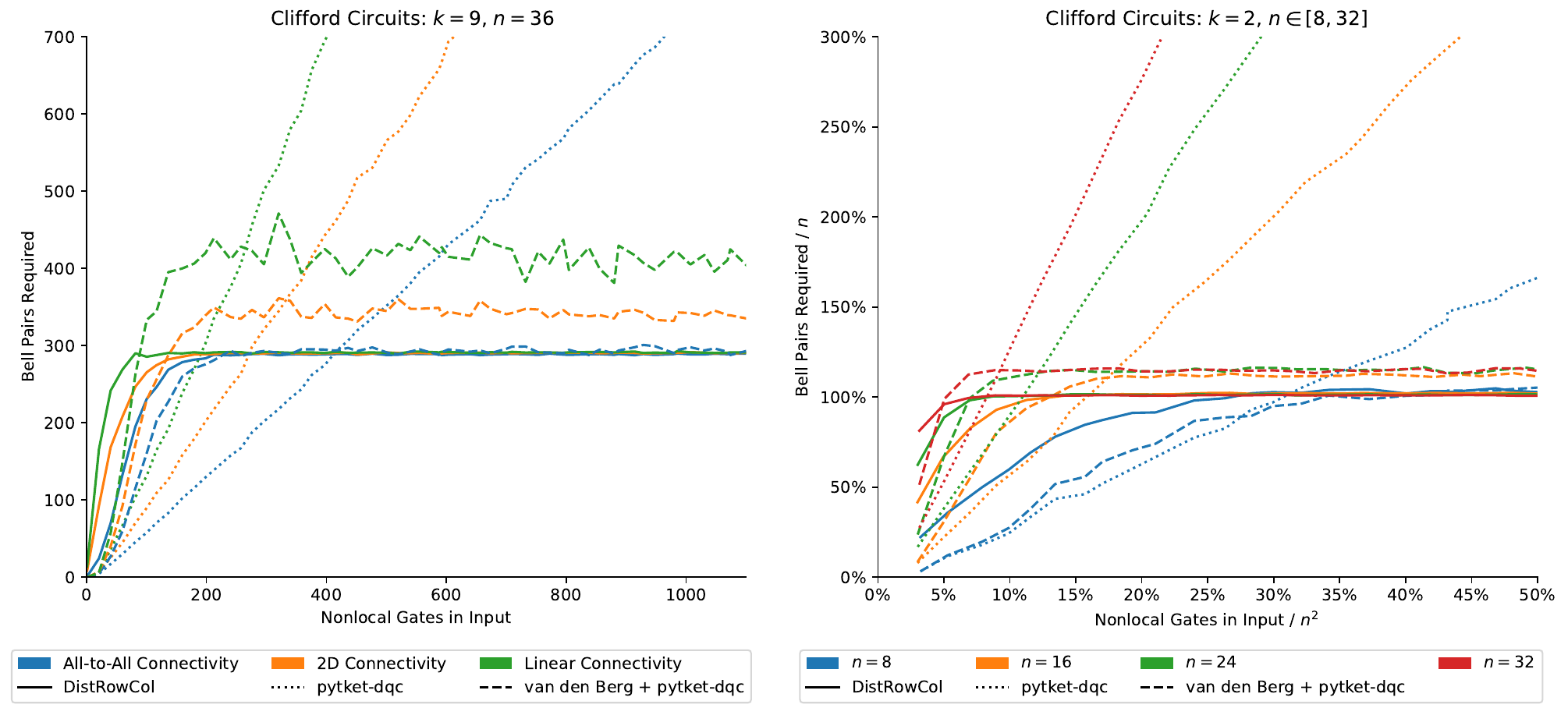}
    \caption{A comparison of distributed Clifford circuit synthesis methods. \emph{Left:} The case of $k = 9$ partitions and $n = 36$ qubits, organized into three architectures: all-to-all, a 3x3 2D grid, and a line. For shallow circuits, \texttt{pytket-dqc} is the clear winner, while for deeper circuits, \textsc{DistRowCol} is superior. \textsc{DistRowCol} exhibits essentially no variation across different connectivities, while \texttt{pytket-dqc} does. \emph{Right:} The case of $k = 2$ partitions with a varying number of qubits $n$. The number of non-local input gates is normalized against $n^2$ for simpler comparability, while the output is normalized against $n$ (which is how we expect it to scale, given that $k$ is fixed). We see that \textsc{DistRowCol} has a small but persistent advantage for deeper circuits with $n > 8$. In the case of $n = 8$, it is worse than \texttt{pytket-dqc} with resynthesis except for very deep circuits. This may be because it does not support qubit reordering.}\label{fig:benchclifford}
\end{figure}

\begin{figure}
    \includegraphics[width=\textwidth]{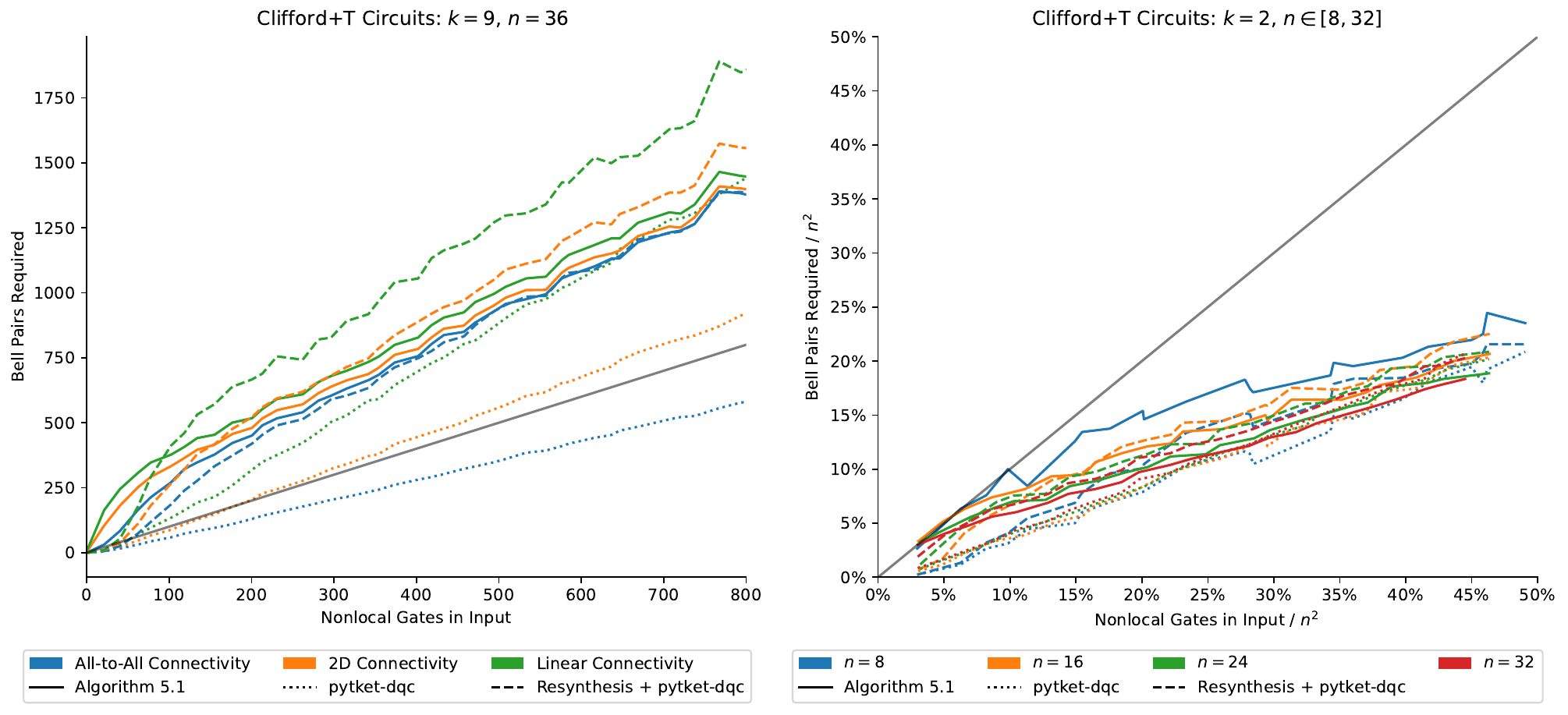}
    \caption{A comparison of distributed Clifford+T circuit synthesis methods. \emph{Left:} The case of $k = 9$ partitions and $n = 36$ qubits, organized into three architectures: all-to-all, a 3x3 2D grid, and a line. \texttt{pytket-dqc} is the clear winner here. Algorithm \ref{alg:distpauliexp} outperforms \texttt{pytket-dqc} with resynthesis, but it is still always worse than break-even. This suggests that resynthesis is generally a bad idea for Clifford+T circuits with more partitions. \emph{Right:} The case of $k = 2$ partitions with a varying number of qubits $n$. The number of non-local input and output gates are both normalized by the same $n^2$ factor, since we expect them to scale similarly, to enable easier comparability. We see that all methods are relatively similar here, especially as the number of qubits $n > 8$ increases. In the case of $n = 8$, Algorithm \ref{alg:distpauliexp} is notably bad. This may be because it does not support qubit reordering. The break-even line is marked in gray.}\label{fig:benchclifft}
\end{figure}

\begin{figure}
    \includegraphics[width=\textwidth]{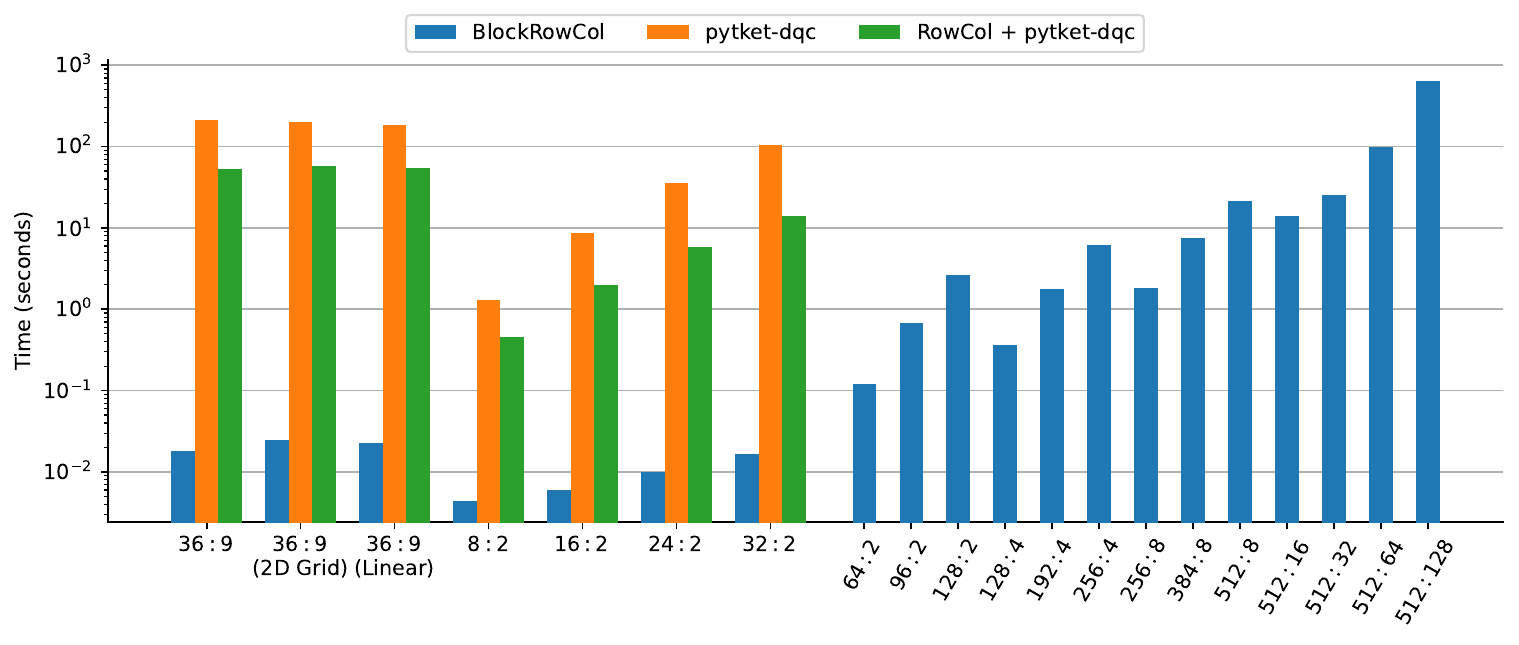}
    \caption{A comparison of the average total runtime used by each of the three distributed CNOT circuit synthesis methods compared in this section. Each bar is labeled as $n : k$, and represents an all-to-all connected architecture except where otherwise notated. Each instance ran on one core of an AMD EPYC 7551 processor. Note the log-scale on the Y-axis. On the left, we see the two families of architectures that are compared in Figure \ref{fig:benchcnot}. We see that there is a significant slowdown for all methods when increasing the number of qubits (and hence circuit depth) or the number of partitions, but that they are relatively independent of choice of connectivity. On the right, we see how \textsc{BlockRowCol} scales as the number of qubits is increased to $n = 512$, for a range of numbers of partitions up to $k = 128$. }\label{fig:benchspeed}
\end{figure}

\section{Discussion}
\label{sec:discussion}

In the previous sections we have given methods for optimizing various kinds of circuits for distributed architectures. As also discussed in the introduction, in this section we provide further motivation for these constructions by seeing how they can be applied to three scenarios: logical CNOT circuits in general CSS codes \cite{calderbank1996good} and phantom codes \cite{koh2026entangling}, Pauli exponential synthesis in the bicycle architecture \cite{yoder2025tour}, and classical simulation of quantum circuits via tree tensor networks \cite{seitz2023simulating}. Finally, we will discuss improvements and extensions to our methods that could be made in the future.

\subsection{Application to CSS Codes}
\label{sec:csscodes}

\newcommand{\lCNOT}{\ensuremath{\overline{\mathrm{CNOT}}}~}

Recent architecture proposals \cite{yoder2025tour, webster2026pinnacle} for fault-tolerant quantum computation focus on the use of quantum error correcting codes that encode more than one logical qubit in each code block. These can be much more efficient than other methods, such as those based on the surface code \cite{litinski2019game}, but in exchange, the fault-tolerant operations that can be performed between blocks are limited in addressability; operations targeting specific qubits within a block may not be available. Moreover, due to physical layout constraints, it may not be possible to apply operations between an arbitrary pair of blocks. One non-local operation which is available for any CSS code \cite{calderbank1996good}, is a transversal (and hence fault-tolerant) \lCNOT operation between blocks of the code that performs a CNOT gate between all pairs of qubits. That is, given the encoder map $E$, we have:
\ctikzfig{phantom-code-encoder}
In fact, any code that has a (weakly) transversal \lCNOT must be (locally equivalent to) a CSS code \cite{zalcman2023css}. While it is possible to implement any specific logical non-local CNOT gate using local operations and \lCNOT gates, applying this translation naively to a CNOT circuit would be inefficient and hard to parallelize. Instead, it would be beneficial to group together many non-local gates and implement them using a few \lCNOT gates, while respecting the inter-block connectivity constraints of the architecture. We will now show how to adapt \textsc{BlockRowCol} for this task, by constructing arbitrary block row additions using a constant number of \lCNOT gates. 

First, note that when performing a block row addition with coefficient matrix $R$, \textsc{BlockRowCol} uses Lemma \ref{thm:blockrowaddgates} to generate a set of $r = \mathrm{rank}(R)$ ZX-type generalized CNOT gates $G_1, \dots, G_r$. By construction, these must all commute, and so we can find a pair of local CNOT circuits $C_A$ and $C_B$ such that all $G_i$ can be performed simultaneously:
\ctikzfig{phantom-code-rank-decomp}
The CNOT gates in the middle of the circuit are acting on the top $r$ qubits of each block. In the case that $R$ is full rank, which will occur relatively frequently in random CNOT circuits, this can be implemented exactly with one \lCNOT gate. In the general case, we can always decompose this into at most three \lCNOT operations. To see this, note that:
\ctikzfig{phantom-code-rank}
This identity shows us how to implement a rank-one block row addition using two \lCNOT operations in the two-qubit case. By repeating this $r$ times on separate pairs of qubits and combining the \lCNOT gates, we can implement a rank-$r$ block row addition also using two \lCNOT operations, whenever $r \leq \frac{1}{2} \frac{n}{k}$. Note that the non-local gates in the two \lCNOT operations will cancel out on the $\frac{n}{k} - 2r$ remaining untouched qubits. In the case that $r > \frac{1}{2} \frac{n}{k}$, we can instead apply non-local CNOT gates to $\frac{n}{k} - r$ qubits using the previous technique, and then an additional \lCNOT operation which cancels with these to obtain $r$ non-local CNOT gates as desired. 

Strategies to implement local operations can vary greatly depending on the code (as demonstrated in the next section), but a generic strategy that applies to many codes is surgery \cite{cowtan2026parallel}, which can implement logical Pauli product measurements. Assuming that at least one ancilla qubit is available to each code block, it is known that any $\frac{n}{k}$-qubit Clifford circuit can be implemented using at most $O(\frac{n}{k})$ such measurements. This is done by first decomposing the circuit into at most $\frac{n}{k} + 1$ unitaries of the form $P(A, \frac{\pi}{2})$ \cite{pllaha2021decomposition}, each of which can be implemented by preparing an ancilla, performing a joint Pauli product measurement, and finally destructive measurement of the ancilla \cite[Figure 11(b)]{litinski2019game}. 

Since \textsc{BlockRowCol} performs $2k(k-1)$ block row additions followed by $k$ local CNOT circuits, the above construction shows that we can write any CNOT circuit using at most $6k(k-1)$ \lCNOT gates and $O(k^2)$ local CNOT circuits. This leads to the following theorem:

\begin{theorem}
    Suppose $n$ logical qubits are encoded in $k$ blocks of a CSS code, where each block has an ancilla available. Then \textsc{BlockRowCol} can be used to implement any logical CNOT circuit with at most $6k(k-1)$ inter-block \lCNOT operations and $O(nk)$ intra-block Pauli product measurements, regardless of the inter-block connectivity.
\end{theorem}

\subsection{Application to Phantom Codes}
\label{sec:phantom}

Phantom codes were recently introduced \cite{koh2026entangling} as a collection of CSS codes characterized by the property that any logical CNOT gates between qubits within one block of the code can be realized as permutations of the physical qubits, which are essentially free, and trivially fault-tolerant. To synthesize logical CNOT circuits for phantom codes, we can therefore use any combination of local logical CNOT gates along with $\overline{\mathrm{CNOT}}$. 

This fits well in our framework because only the non-local $\overline{\mathrm{CNOT}}$ gates actually contribute any physical gates, so if we were to minimize the number of non-local gates, we directly minimize the number of two-qubit physical gates and hence the error. Since \textsc{BlockRowCol} performs at most $2k(k-1)$ block row additions, the results of the previous section give the following theorem. This is quite similar to Theorem 1 from \cite{koh2026entangling}, but slightly more general, because it allows for any number of code blocks (rather than a power of two), and does not require all-to-all connectivity between blocks. Importantly, however, their result also shows a linear bound on the depth of logical CNOT circuits which is not the case for \textsc{BlockRowCol}.

\begin{theorem}
    Suppose $q$ physical qubits encode $n$ logical qubits which are organized in $k$ blocks of a $[\![\frac{q}{k}, \frac{n}{k}, d]\!]$ phantom code. Then \textsc{BlockRowCol} can be used to implement any logical CNOT circuit with at most $6k(k - 1)$ \lCNOT operations, regardless of the inter-block connectivity. This corresponds to at most $6q(k-1)$ physical CNOT gates.
\end{theorem}

\subsection{Application to the Bicycle Architecture}

Recently \cite{yoder2025tour} introduced the bicycle architecture, which is a design for a modular quantum computer that uses multiple bivariate bicycle code blocks coupled together with shared entanglement. In particular, the modules have a linear nearest neighbor connectivity, with one qubit in each capable of non-local gates, and a T-gate factory which can produce non-Clifford phase gates in one module. In this architecture, the highest-error operations are non-local gates, which have two orders of magnitude higher logical error rates than local gates. Therefore, this is a good fit for our framework, especially since its native gate-set is easily phrased in terms of Pauli exponentials. \cite[Section 3]{yoder2025tour} details a compilation pipeline which implements $d$ (non-local) Pauli exponentials across $k$ modules using $kd$ non-local gates. In contrast, the methods presented in Section \ref{sec:distpauliexp} use $2kd$ non-local gates for the same task in the worst case, although in practice for dense circuits we would expect the performance to be closer to $kd$ (since the probability of the Steiner tree needing fill-in decreases exponentially as the block size increases, for random circuits). \cite{yoder2025tour} achieves this by using ancillas and measurement, which our representation does not account for; further work would be required to match this.

\subsection{Application to Tree Tensor Networks}

Tensor network methods are a family of algorithms for classically simulating quantum circuits. Popular variants include the matrix product state and tree tensor network methods \cite{seitz2023simulating}. These aim to speed up classical simulations relative to statevector methods by taking advantage of the low-rank structure of the connection between neighboring qubits. Consequently, for these methods, the application of quantum gates between distant qubits can be either extremely expensive, because the connection between all the intermediate qubits may be affected, or induce errors if the network must be truncated to conserve memory. Therefore, we would like to rewrite our circuits in such a way that these gates are minimized. Here, we will consider CNOT circuits; unfortunately it is unclear how to apply the same technique for Clifford circuits, except via Clifford normal forms \cite{aaronson2004improved}.

We focus on tree tensor networks, and specifically those organized a perfect binary tree. In these tensor networks, qubits are organized as the leaves of a binary tree. Suppose we are given two qubits $i$ and $j$, then we say that a gate between $i$ and $j$ is a level-$l$ gate if the distance from $i$ and $j$ to their nearest common ancestor in the tree is $l$. The cost to apply a gate scales as $2^{2^l}$ in the worst case \cite{seitz2023simulating}, which is super-exponential in the level. Thus, minimizing the number of highest-level gates is especially important, and in fact \textsc{BlockRowCol} can be used for this task. Suppose that we are given a system of $n = 2^a$ qubits, then the highest level gates are at level $a - 1$, and the cost to apply them will be $2^\frac{n}{2}$. Now consider a CNOT circuit $C$ and apply \textsc{BlockRowCol} on $k = 2$ blocks; this uses at most three block row additions. As in Section \ref{sec:phantom}, each of these commute, and so we can organize them into three layers of at most $\frac{n}{2}$ level-$(a-1)$ CNOT gates
\ctikzfig{tree-tensor-network-decomp}
and eight local CNOT circuits $C_A$ through $C_H$. Now suppose that were were to apply \textsc{BlockRowCol} recursively to each of these, with $k = 2$ blocks each time, until only one qubit remains in each block. The total cost to apply all of these gates can be defined by the following recurrence relation:
$$ C(n) \leq 3 \cdot \frac{n}{2} \cdot 2^{\frac{n}{2}} + 8C\left(\frac{n}{2}\right) \implies C(n) \leq \frac{3n}{2}\sum_{k = 0}^{\log_2 n} 2^{2k + \frac{n}{2^{k+1}}}$$ 
Consider the function $f(k) = 2k + n2^{-(k+1)}$, then $f'(k) = 2 - \ln(2)n2^{-(k+1)} < 0$ whenever $k \leq \log_2(n) - 3$. Thus, all but the last four terms of the sum are decreasing, so we have:
$$ C(n) \leq \frac{3n}{2}\left[\log_2{n} \cdot 2^\frac{n}{2} + \sum_{i = 0}^3 2^{2(\log_2{n} - i) + n2^{i - \log_2{n} - 1}}\right] \leq \frac{3n}{2}\left[\log_2{n} \cdot 2^\frac{n}{2} + 3n^2\right] = O(2^\frac{n}{2}n\log_2{n})$$
By comparison, applying Patel-Markov-Hayes' algorithm \cite{patel2003efficient} would generate $O(n^2 / \log_2 n)$ highest-level gates, with cost $O(2^\frac{n}{2}n^2/\log_2{n})$, and hence \textsc{BlockRowCol} can do asymptotically better by a factor of $n / \log_2^2{n}$. A similar conclusion would hold whenever the cost of applying a gate grows sufficiently quickly as the level increases, even if it is not super-exponential as assumed here. We leave a practical implementation and benchmark of this technique to future work.

\subsection{Future Work}
\label{sec:future}

We conclude by discussing some directions for future research. Perhaps most impactful improvement to our methods would be to find more effective heuristics for the many steps of the algorithms given here that are not fully fixed by the problem. Some of these have previously been considered at length for non-distributed settings, like pivot selection \cite{winderl2023architecture}, and could be adapted. On the other hand, there are additional free variables that are not present in the non-distributed setting, like the choice of anticommuting Pauli strings when eliminating a block using generalized CNOT gates, for which customized heuristics would need to be developed. 

We could also consider extensions to the circuit representation, for instance, by integrating these methods with previous embedding-based approaches to distributed circuits \cite{wu2023entanglement}, or by considering new block gates that are both non-local and non-Clifford (for example, replacing generalized CNOT gates with generalized controlled-phase gates). In the same direction, we would like to consider the case of synthesis up to a permutation of the qubits, since the ability to reorder qubits between partitions appears to have played a role in our benchmarking results. While synthesis up to a permutation of \emph{partitions} is fairly trivial to implement (it is straightforward to adapt the PermRowCol \cite{meijer2023dynamic} and reverse-traversal \cite{li2019tackling} methods), permutations of qubits between partitions appears more difficult. More generally, tableau-based elimination methods are not the only algorithms for synthesizing CNOT, Clifford, and Clifford+T circuits: it would be interesting to see if SAT-based methods \cite{shaik2025cnot}, the LazySynth framework \cite{martiel2022architecture}, template-based methods \cite{bravyi2021clifford}, or peephole optimization \cite{kliuchnikov2013optimization} could be adapted to use block operations.

Finally, we close with an open question: in Theorem \ref{thm:blockrowcolapproxopt}, we were able to show that \textsc{BlockRowCol} is approximately optimal for $k = 2$ because we could define the concept of the rank of a submatrix, which is modified in a controlled way when a ZX-type generalized CNOT gate is applied. Is there an analogous concept of rank for subsections of a stabilizer tableau, so that (generic) generalized CNOT gates look like `rank-one updates'? This does not seem readily apparent, even considering the representation of tableaux as symplectic matrices.

\subsection*{Acknowledgments}

\noindent The author acknowledges support from the US Department of Energy under Award No. DE-SC0020264, and thanks Pablo Andr\'es-Mart\'inez for discussions regarding tree tensor network simulation, Richie Yeung for useful information on Clifford circuit synthesis, Benjamin Rodatz and Boldizs\'ar Po\'or for suggesting the multipartite setting and the application to block codes, and the QPL 2025 and 2026 reviewers for helpful comments.

\vspace{2mm}

\noindent This work is dedicated to the late Nuno F. Loureiro, without whose patience and support it would not have been possible. My thanks to the whole Loureiro group for creating such a welcoming environment.

\bibliographystyle{quantum}
\bibliography{refs}

\end{document}